%% file: main.tex
\documentclass[12pt]{article}
\usepackage[letterpaper, margin=1in]{geometry}

\usepackage{times}
\usepackage{bm}
\usepackage{natbib}

\usepackage[plain,noend]{algorithm2e}
\usepackage{bibunits}

\makeatletter
\renewcommand{\algocf@captiontext}[2]{#1\algocf@typo. \AlCapFnt{}#2} 
\def\@algocf@capt@plain{top}
\renewcommand{\algocf@makecaption}[2]{%
  \addtolength{\hsize}{\algomargin}%
  \sbox\@tempboxa{\algocf@captiontext{#1}{#2}}%
  \ifdim\wd\@tempboxa >\hsize
    \hskip .5\algomargin%
    \parbox[t]{\hsize}{\algocf@captiontext{#1}{#2}}
  \else%
    \global\@minipagefalse%
    \hbox to\hsize{\box\@tempboxa}
  \fi%
  \addtolength{\hsize}{-\algomargin}%
}
\makeatother

\usepackage{amsmath}
\usepackage{url}
\usepackage{bbm,bm}
\usepackage{array,multicol,enumitem}
\usepackage{amsfonts} 
\usepackage{amssymb,mathtools}
\usepackage{amsmath,amsthm}
\usepackage{soul}
\usepackage{verbatim}
\usepackage{rotating}
\usepackage[xcolor]{changebar}   
\graphicspath{{./figures/} }

\def\v{{\varepsilon}}

\renewcommand{\hat}[1]{\widehat{#1}}
\renewcommand{\tilde}[1]{\widetilde{#1}}
\newcommand{\E}{E}
\renewcommand{\P}{pr}
\newcommand{\R}{\mathbb{R}}

\newcommand{\bTheta}{{\Theta}}
\newcommand{\bbeta}{\beta}

\newcommand{\y}{y}

\newcommand{\X}{X}
\newcommand{\bbxi}{\Xi}

\newcommand{\bOmega}{\Omega}
\newcommand{\bPsi}{\Psi}

\newcommand{\bPhi}{\Phi}
\newcommand{\bepsilon}{\varepsilon}
\newcommand{\var}{\text{var}}

\newcommand{\bs}[1]{\boldsymbol{#1}}

\newcommand{\s}[1]{\mathcal{#1}}

\newcommand{\U}{{U}}

\renewcommand{\v}{{v}}

\newcommand{\indep}{\perp \!\!\! \perp}

\newtheorem{theorem}{Theorem}
\newtheorem{lemma}[theorem]{Lemma}
\newtheorem{proposition}[theorem]{Proposition}
\newtheorem{condition}[theorem]{Condition}

\newtheorem{corollary}[theorem]{Corollary}
\newtheorem{definition}[theorem]{Definition}

\newtheorem{remark}[theorem]{Remark}

\title{Regression of exchangeable relational arrays}

\author{Frank W. Marrs$^1$, Bailey K. Fosdick$^2$, and Tyler H. McCormick$^3$ \\[1ex]
LA-UR-20-30061 \\[1ex] 
}
\date{%
    $^1$Statistical Sciences, Los Alamos National Laboratory,\\ 
    \textit{fmarrs3@lanl.gov}\\[1ex]%
    $^2$Department of Statistics, Colorado State University\\[1ex]%
    $^3$Departments of Statistics and Sociology, University of Washington\\[2ex]%
    \today
}

\begin{document}

\maketitle

\begin{abstract}
Relational arrays represent 
measures of association between pairs of actors, often in varied contexts or over time.  Trade flows between countries, financial transactions between individuals, contact frequencies between school children in classrooms, and dynamic protein-protein interactions are all examples of relational arrays.
Elements of a relational array are often modeled as a linear function of observable covariates. 
Uncertainty estimates for regression coefficient estimators -- and ideally the coefficient estimators themselves --  must account for dependence between elements of the array (e.g. relations involving the same actor) and existing estimators of standard errors that recognize such relational dependence rely on estimating extremely complex, heterogeneous structure across actors. 
This paper develops
a new class of parsimonious coefficient  and standard error estimators for regressions of relational arrays. 
We leverage an exchangeability assumption to derive standard error estimators that pool information across actors and are substantially more accurate than existing estimators in a variety of settings. 
 This exchangeability assumption is pervasive in network and array models in the statistics literature, but not previously considered when adjusting for dependence in a regression setting with relational data. 
 We demonstrate improvements in inference theoretically, via a simulation study, and by analysis of a data set involving international trade. \\
 
\end{abstract}

\section{Introduction}

Entries in relational arrays quantify pairwise interactions between actors that may be of multiple types or observed over time. 
Examples 
include annual flows of migrants between countries \citep{aleskerov2016network} and  interactions between students over the course of a semester \citep{han2016using}.  In economics, 
relational arrays are used to describe monetary transfers between individuals as part of informal insurance markets ~\citep[see][]{bardham1984land, fafchamps2006development, foster2001imperfect,attanasio2012risk,banerjee2013diffusion}).
Other examples of data that are naturally represented as relational arrays include gene expressions \citep{zhang2005general} and international relations \citep{fagiolo2008topological}. 

A relational array 
$Y=\left(y_{ijr} \right)$ where $(i,j = 1,\ldots,n, i\not=j, r = 1, \ldots ,R)$,
is composed of a series of $R$ matrices of size $(n \times n)$, each of which describes the directed pairwise relationships among $n$ actors of type $r$, e.g. time period $r$ or relation context $r$. 
The diagonal elements of each matrix, for example $y_{iir}$, are assumed to be undefined, as we do not consider, e.g., international relations of a country with itself.  
The relationship from actor $i$ to actor $j$ may differ than that from $j$ to $i$, $y_{ijr} \neq y_{jir}$, however, the methods we propose extend to the symmetric relation case (see Section~\ref{app_undir}).

The primary goal in our setting is inference for linear regressions exploring the effects of exogenous covariates on the values in the relational array, expressed as
\begin{align}
y_{ijr} = {\beta}^T x_{ijr}+\xi_{ijr},  \quad  (i,j=1,\ldots, m; \ i \not= j; \ r = 1, \ldots ,R),
\label{eq:lm}
\end{align}  
where $y_{ijr}$ is a continuous directed measure of the $r$th relation from actor $i$ to actor $j$, $x_{ijr}$ is a $(p \times 1)$ vector of covariates, and  $\xi_{ijr}$ is an unobserved, scalar random error. 
For example, considering informal insurance markets, \cite{fafchamps2007formation} examine how covariates such as geographical proximity and kinship relate to risk sharing relations after economic shocks.

A core  challenge in making inference on $\beta$ arises from the innate dependencies among error relations involving the same actor. 
For example, dependence often exists between trade relations involving the same country or between economic transfers originating from the same individual. This dependence may arise due to variation unaccounted for in the covariates, for example, from differences in production levels between nations or
from individual differences in risk aversion. 
Standard regression techniques 
may lead to poor estimates of $\beta$ and/or incorrect conclusions regarding the significance of the estimate of $\beta$.
Approaches to account for error dependence in relational arrays have appeared in the statistics, biostatistics, and econometrics literatures and can be characterized into two broad classes. 

The first set of approaches impose a  parametric model on the errors.  
Specifically, they either use latent variables to model the array measurements as conditionally independent given the latent structure \citep[see][]{holland1983stochastic, wang1987stochastic, hoff2002latent, li2002unified, hoff2005bilinear} or model the error covariance structure directly subject to a set of simplifying assumptions ~\citep[see][]{hoff2011separable,fosdick2014separable,hoff2015multilinear}.  While these methods allow for (possibly) improved estimation of $\beta$ and appropriate standard error estimators in the presence of relational dependence,
the accuracy of inference on $\beta$ depends on the extent to which the true error structure is consistent with the specified parametric model. In addition,  many of these models are estimated in a Bayesian paradigm using Markov chain Monte Carlo approaches, which are commonly computationally expensive to estimate.

The second set of approaches to accounting for relational dependence relies heavily on empirical estimates of the error structure based on the regression residuals, 
first proposed by \cite{fafchamps2007formation} and based on the spatial dependence work of \cite{conley1999gmm}.  
This framework is model agnostic, making as few assumptions as possible about the data generating process. One empirical approach estimates the regression coefficients using ordinary least squares, and then utilizes a sandwich covariance estimator -- which is robust to a wide array of error structures --  for the standard errors of the regression coefficients \citep{fafchamps2007formation, aronow2015cluster}. 
In finite samples, this estimator 
is hindered by the need to estimate a large number of covariance parameters with limited
observations (\cite{king2014robust} for a discussion in other contexts), and is the reason why~\cite{wakefield2013bayesian} suggests such estimators be labeled empirical rather than robust. 
We observe that standard errors from this empirical framework are often highly variable and are anticonservative.

In this work, we introduce an empirical estimation approach for relational arrays that incorporates an exchangeability assumption. 
This assumption is implicit in many of the model-based approaches discussed previously and is a hallmark of Bayesian hierarchical models within, and outside, the relational context (\cite{orbanz2015bayesian}). 
Our key contribution is to define the covariance matrix (and a corresponding estimator) of the relational error array under exchangeability.
Use of our parsimonious estimator produces superior estimates of $\beta$ and its standard errors relative to  existing approaches, 
and our estimator is easier to compute than existing Bayesian model-based  and exchangeable bootstrapping \citep{menzel2017bootstrap,  green2017bootstrapping} approaches. 
Reproduction code is available at \url{https://github.com/fmarrs3/netreg_public} and methods are implemented in \url{R} package \url{netregR}.

\section {Inference in relational regression}
\subsection{Estimation of regression coefficients}

We employ a least squares framework to perform inference on ${\beta}$ in the relational regression model in \eqref{eq:lm} \cite{aitkin1935least}. 
An unbiased estimator for $\beta$ is the ordinary least squares estimator, 
$\hat{\beta} = \left( X^T X \right)^{-1}X^T y$, where $X$ is an $\big(R n(n-1) \times p\big)$ matrix of $(p \times 1)$ covariate vectors $( x_{ijr} )$ and ${y}$ is a vectorized representation of $( y_{ijr} )$.
The least squares estimator is the best linear unbiased estimator for $\beta$ when the covariance matrix $\Omega = \var(y~\mid~X)$ is proportional to the identity matrix. Dependence is expected in relational data, e.g. between relations $(i,j,r)$ and $(i,k,r)$ which share actor $i$.
If 
$\Omega$ were known, 
the best linear unbiased estimator for $\beta$ is the generalized least squares estimator of \cite{aitkin1935least}, 
\begin{align}
    \hat{\beta}_{GLS} = (X^T \Omega^{-1} X)^{-1} X^T \Omega^{-1} y.
    \label{eq:gls}
\end{align}
In practice, $\Omega$ is unknown and must be estimated. Given estimator  $\hat{\Omega}$, alternating  estimation of ${\Omega}$ with \eqref{eq:gls} (replacing $\Omega$ with $\hat{\Omega}$ at each iteration) is termed feasible generalized least squares. When $\hat{\Omega}$ is consistent, feasible generalized least squares is asymptotically efficient for $\beta$ \citep{greene2003econometric, hansen2015econometrics}.

Regardless of whether the ordinary least squares estimator or \eqref{eq:gls} is used to estimate $\beta$, uncertainty estimates are required for inference. A common approach is to approximate 
the distribution  of the $\beta$ estimate 
as a multivariate normal random variable and construct confidence intervals using an estimator of its variance: in the ordinary least squares setting,
\begin{align}
    \var \left( \hat{\beta} \mid X \right) 
    &= (X^T X)^{-1} X^T \Omega X (X^T X)^{-1},
     \label{eq:ols_var}
\end{align}
and in the generalized least squares setting, 
\begin{align}
    \var \left( \hat{\beta}_{GLS}\mid X\right) 
    &= (X^T \tilde{\Omega}^{-1} X)^{-1}  X^T \tilde{\Omega}^{-1} \Omega \tilde{\Omega}^{-1} X (X^T \tilde{\Omega}^{-1} X)^{-1},
    \label{eq:gls_var}
\end{align}
where 
$\widetilde{\Omega}$ is the final estimate of $\Omega$ from the generalized least squares procedure.
Variance estimators are often constructed by substituting an estimator for $\Omega$ in  \eqref{eq:ols_var} and \eqref{eq:gls_var}, and are commonly termed sandwich estimators \citep{huber1967behavior,white1980heteroskedasticity}. 
Thus, inference for $\beta$ requires an estimator for $\Omega$,
regardless of how $\beta$ is estimated,
and
properties of the estimator of $\var( \hat{\beta} \mid X)$ depend strongly on the estimator of $\Omega$.

\subsection{Dyadic clustering estimator} 
\label{sec:DCest}
\cite{fafchamps2007formation}, \cite{cameron2011robust},  \cite{aronow2015cluster}, and~\cite{tabord2017inference}
propose and describe the properties of a flexible standard error estimator for relational regression which makes the sole assumption that two relations $(i,j,r)$ and $(k,l,s)$ are independent if $(i,j)$ and $(k,l)$ do not share an actor.
This assumption implies that $\text{cov}(y_{ijr},y_{kls}\mid X) =\text{cov}(\xi_{ijr},\xi_{kls} \mid X)=0$ for non-overlapping relation pairs, but  places no restrictions on the covariance elements for pairs of relations that share an actor.  Let $\Omega_{DC}$ denote the covariance matrix of $\xi$, subject to this non-overlapping pair independence assumption.  \cite{fafchamps2007formation} propose estimating each nonzero entry of $\Omega_{DC}$ with a product of residuals, i.e. using $e_{ijr}e_{iks}$ to estimate $\text{cov}(\xi_{ijr},\xi_{iks})$, where $e_{ijr} = y_{ijr} - \hat{\beta}^T x_{ijr}$.
The estimator $\hat{\Omega}_{DC}$ can be seen as that which takes the empirical covariance of the residuals defined by ${e} {e}^T$, where $e$ is a vector of the set of residuals $( e_{ijr} )$, and  introduces zeros to enforce the non-overlapping pair independence assumption.  
\cite{fafchamps2007formation} propose a sandwich variance estimator for $\var ( \hat{\beta} \mid X )$ in \eqref{eq:ols_var} based on $\hat{\Omega}_{DC}$, 
\begin{equation}
\hat{V}_{DC} =
(X^TX)^{-1} X^T \hat{\Omega}_{DC} X (X^TX)^{-1}.
\label{eq:varc}
\end{equation}
We refer to $\hat{V}_{DC}$ as the { dyadic clustering  estimator} as it owes its derivation to the extensive literature on cluster-robust standard error estimators.

The dyadic clustering estimator  in \eqref{eq:varc} 
has the attractive properties that it is asymptotically consistent under wide range of error dependence structures 
and is fast to compute. 
However,
$\hat{\Omega}_{DC}$ estimates ${O}(R^2 n^3)$ nonzero covariance elements separately based on ${O}(R^2 n^2)$ dependent observations, and thus,
$\hat{V}_{DC}$ is inherently quite variable. 
Only when there is extreme heterogeneity in the true covariance structure is the dyadic clustering method optimal and it will suffer a loss of efficiency otherwise.

\section{Standard errors under exchangeability} 
\label{sec:structured}

\subsection{Exchangeability in relational models} 

 A common modeling assumption for relational and array structured errors is exchangeability.  Defined by de Finetti for a univariate sequence of random variables, exchangeability was generalized to array data by \cite{hoover1979relations} and \cite{aldous1981representations}.
 The errors in a relational data model are jointly exchangeable if the probability distribution of  the error array $\Xi = ( \xi_{ijr} )$ is invariant under any simultaneous permutation of the rows and columns, and secondary permutation of the third dimension.  Mathematically, this means
 $\text{pr}(\Xi)=\text{pr}\{ \Pi(\Xi) \},$ where $\Pi(\Xi) = \{\xi_{\pi(i)\pi(j)\nu(r)}\}$ is the error array with its indices reordered according to permutation operators $\pi$ and $\nu$.
 Intuitively, exchangeability in the regression context 
 means the observed covariates are sufficiently informative such that the labels of the rows and columns in the error array are uninformative. Similarly the ordering of the third dimension of the error array is uninformative to its distribution. This assumption may be appropriate when the third dimension of the array represents different contexts of observations (such as economic trade sectors) that have no inherent ordering, or when the third dimension represents time periods, but the bulk of the temporal variation is accounted for in the covariates. 
 Many of the 
conditionally independent 
parametric latent variable models cited in the Introduction have this joint exchangeability property  \citep{hoff2008modeling, bickel2009nonparametric}. 

\begin{figure}
  \begin{center}
\begin{tabular}{cc}
\includegraphics[height=.18\textheight]{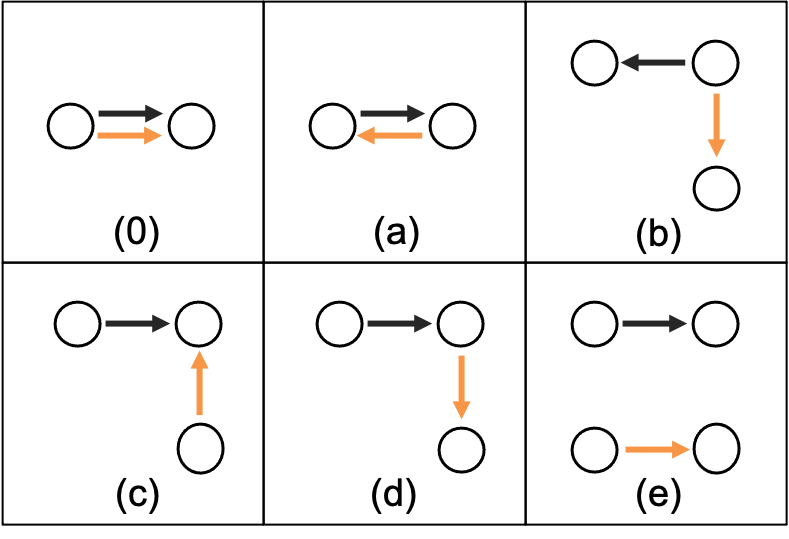}
&
\includegraphics[height=.2\textheight]{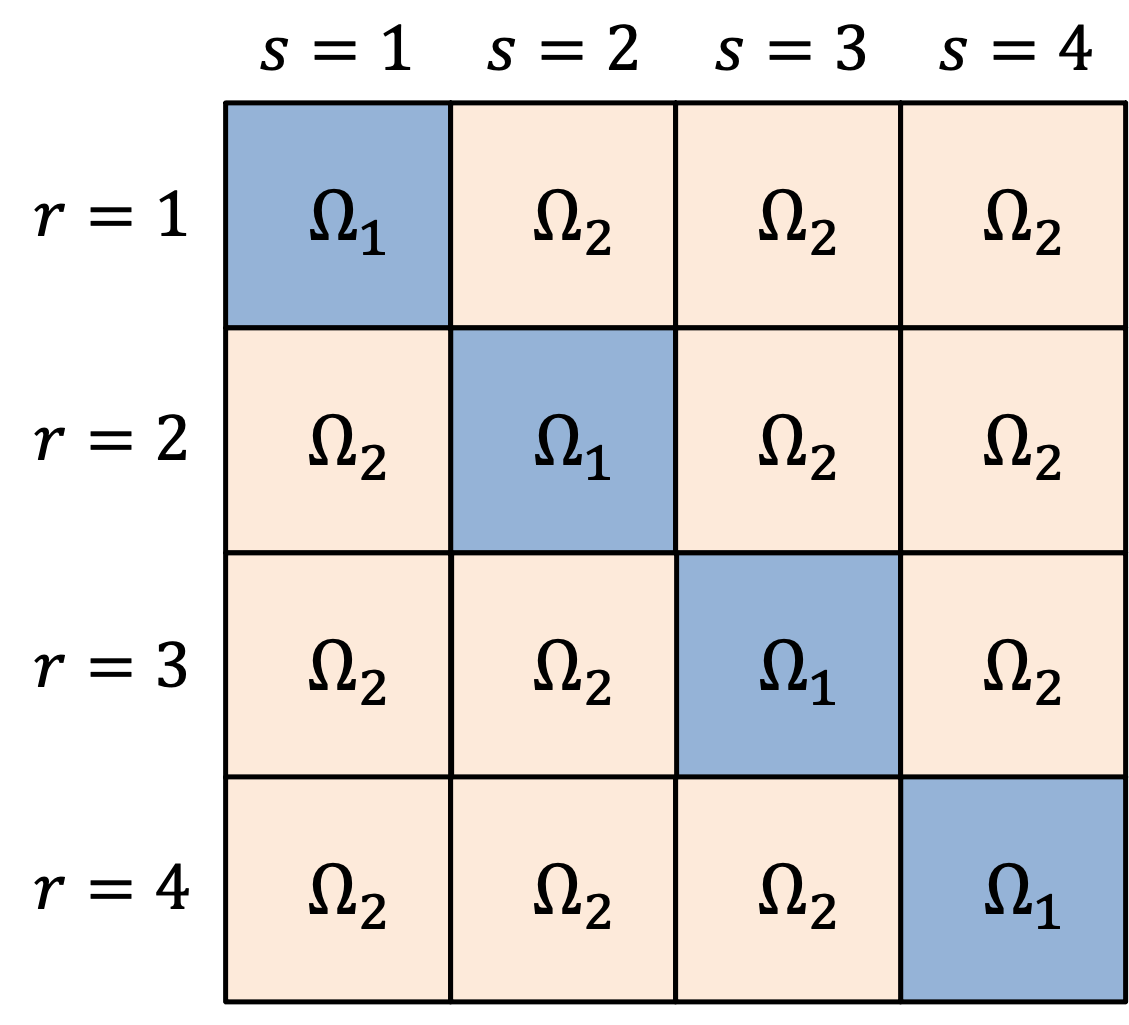}
\end{tabular}
  \end{center}
  \caption{Left: Six distinguishable configurations of relation pairs (black and orange arrows) in an exchangeable relational model with unlabeled actors, corresponding to twelve covariance values (six each for $s=r$ and $s\neq r$). 
  Right: Block structure in $\Omega_E$ for $R=4$, where blocks $\Omega_1$ and $\Omega_2$ correspond to $s=r$ and $s \neq r$, respectively.  } 
    \label{fig:directeddyad}
  \end{figure}

\subsection{Impact of exchangeability on covariance structure}
\cite{li2002unified} and \cite{hoff2005bilinear} describe several particular random effects models for $R=1$. The corresponding error covariance matrices have different entries depending on the model, yet all covariance matrices have at most six unique entries. A key contribution of this paper is to formalize and extend this observation, showing that {any} jointly exchangeable model for relational array $\Xi$ results in an $\Omega$ of the same form, with at most six unique terms when $R=1$ and at most twelve unique terms when $R>1$.

\begin{proposition}
If a probability model for a directed relational array  $\Xi$ is jointly exchangeable and has finite second moments, then the covariance matrix of $\Xi$ contains at most twelve unique values. 
\label{prop}
\end{proposition}
The twelve (possibly) unique entries in $\Omega$ correspond to the
twelve distinguishable configurations of relation pairs $(i,j,r)$ and $(k,l,s)$ with unlabeled actors. The twelve configurations can be separated into two sets of six identical configurations of  relations $(i,j)$ and $(k,l)$ with unlabeled actors, as depicted in Figure~\ref{fig:directeddyad}, where each set corresponds to $r=s$ and $r\neq s$. A proof is provided in Section~\ref{app_prop_proof}.

\subsection{Covariance matrices of exchangeable relational arrays}
\label{sec:exch_cov_mat}
Similar to the dyadic clustering estimator, we assume non-overlapping relation pairs are independent, such that $\text{cov}(\xi_{kls} , \xi_{ijr}) = 0$, for any $s$ and $r$ when $(i,j,k,l)$ are distinct. This assumption sets two of the twelve parameters in $\Omega$ to zero.  We introduce a new class of covariance matrices which contain ten  possibly nonzero entries, $\phi_0^{(\eta)}, \phi_a^{(\eta)},\phi_b^{(\eta)}, \phi_c^{(\eta)},\phi_d^{(\eta)}, (\eta= 1,2)$, associated with the left panel in Figure~\ref{fig:directeddyad}.  The separation of covariances by $r=s$ and $r
\neq s$ implies that $\Omega$ consists of blocks of $\Omega_1$ and $\Omega_2$, each consisting of five nonzero terms for $r=s$ ($\eta=1$) and $r \neq s$ ($\eta=2$), respectively (Figure~\ref{fig:directeddyad}, right panel).
We define an {exchangeable covariance matrix} as any covariance matrix of this form and denote it $\Omega_{E}$.

\section{Exchangeable estimator definition and evaluation}
 \subsection{Exchangeable covariance estimator}
 \label{subsec:exch_estimator}
Consider relational regression models with $\Omega$  of the exchangeable form  $\Omega_E$.
The proposed {exchangeable estimator} of $\var(\hat{\beta} \mid X )$ is then
 \begin{equation}
\hat{V}_{E}= (X^TX)^{-1} X^T \hat{\Omega}_{E} X (X^TX)^{-1}, \hspace{.1in}    \hspace{.1in} 
\hat{\Omega}_{E} = 
\sum_{\eta=1}^2 \sum_{u=0}^d \hat{\phi}^{(\eta)}_u \mathcal{S}^{(\eta)}_u,
\label{eq:exch.est}
 \end{equation}
where $\mathcal{S}^{(\eta)}_u$ denotes the $\{ R n(n-1) \times R n(n-1) \}$ binary matrix with $1$'s in the entries corresponding to relation pairs of type $(u =0,a,b,c,d; \eta = 1,2)$ as defined in Figure \ref{fig:directeddyad}. 
We propose estimating the ten parameters in $\Omega$ by averaging the residual products that share the same index configurations, corresponding to (0)-(d) in Figure \ref{fig:directeddyad}.  For example, the estimate of ${\text{cov}}(\xi_{kls},\xi_{ijr})$, corresponding to $u=b$ and $\eta=2$, is
\begin{align}
\hat{\phi}_b^{(2)} =\binom{R}{2}^{-1}\frac{1} {n(n-1)(n-2)} \sum\limits_{r \neq s} \sum\limits_{i} \sum\limits_{j \neq i} e_{ijr}  \Big( \sum\limits_{k \neq j} e_{iks} - e_{ijs} \Big). \label{eq:resid_avg}
\end{align}
The remaining nine estimators for $(s=0,a,\ldots,e; \eta=1,2)$, are defined analogously,
and $\hat{\Omega}_{E} $ may be interpreted as the projection of $\hat{\Omega}_{DC}$ into the vector space over symmetric matrices of the form of $\Omega_E$. 
We provide theoretical details for the proposed exchangeable estimator, including asymtotic normality of least squares under exchangeability and consistency of $\hat{V}_E$ in Sections~\ref{app:theory}-~\ref{app:consistency}.

\subsection{Comparison of exchangeable estimator with dyadic clustering}
It is intuitive that the moment-based exchangeable estimator is consistent, and more efficient than the dyadic clustering estimator whenever the exchangeability assumption is satisfied. One might expect the highly parametrized dyadic clustering estimator to trade off high variance for reduced bias. However, we derive the result that the dyadic clustering estimator is biased downwards, and that this bias is larger than twice the bias of the exchangeable estimator.
One concludes that one tradeoff for the robustness of the dyadic clustering estimator is anticonservatism. The proof of Theorem~\ref{thm:bias} is provided in Section~\ref{app:bias}. We provide additional theoretical details for the proposed exchangeable estimator (comparison of mean-square error with dyadic clustering) in Section~\ref{app:mse}.

\begin{theorem}
\label{thm:bias}
Consider error vector $\xi$ and normally distributed covariate vector $x$ with mean zero, exchangeable covariance matrix, and bounded fourth moments, where
\[y_{ij} = \beta_1 + x_{ij} \beta_2 + \xi_{ij}.\]
Then, the dyadic clustering estimator for $\var( \hat{\beta}_2 )$ is biased downwards, 
\[ n^{2} \text{Bias}( \hat{V}_{DC} ) +  O(n^{-1/2}) \le -2n^{2}  |\text{Bias}( \hat{V}_{E} )| +  O(n^{-1/2}) \le 0, \]
noting that both $n^{2} \text{Bias}( \hat{V}_{DC} )$  and  $|n^{2} \text{Bias}( \hat{V}_{E} )|$ are $O(1)$.
\end{theorem}

We conducted a simulation study to compare the bias and 95\% confidence interval coverage when using the exchangeable and dyadic clustering estimators. We simulated from a model with three covariates (one each of binary, positive real, and real-valued) with exchangeable and non-exchangeable error models, for $R=1$ and $n=20,40,80,160,320$ (see Section~\ref{app_sim} for details). 
Figure~\ref{fig:sim_coverage} shows the estimated mean coverage and middle 95\% of coverages across various $X$ realizations. In all settings, the estimated mean coverage of the exchangeable estimator is closer to the nominal $0.95$ level than the dyadic clustering estimator.
This difference is most pronounced for the binary covariate, where there is reduced signal-to-noise relative to the other covariates.
The average bias of the dyadic clustering estimator is typically more than four times that of the exchangeable estimator under the exchangeable error model, confirming Theorem~\ref{thm:bias} and driving the poorer coverage performance 
(see Section~\ref{sec:sim_bias} for details).

 \begin{figure}[ht!]
\centering
  \begin{tabular}{rccc}
  \hspace{-.2in}&\hspace{-.2in}
    \hspace{1pt} binary

  \hspace{-.22in}&\hspace{-.22in}
  \hspace{1pt} positive real
  \hspace{-.22in}&\hspace{-.22in}
  \hspace{1pt} real-valued \\
  \begin{sideways}\hspace{.55in} Estimated coverage \end{sideways} 
  \hspace{-.2in}&\hspace{-.2in}
\includegraphics[width=.315\textwidth]{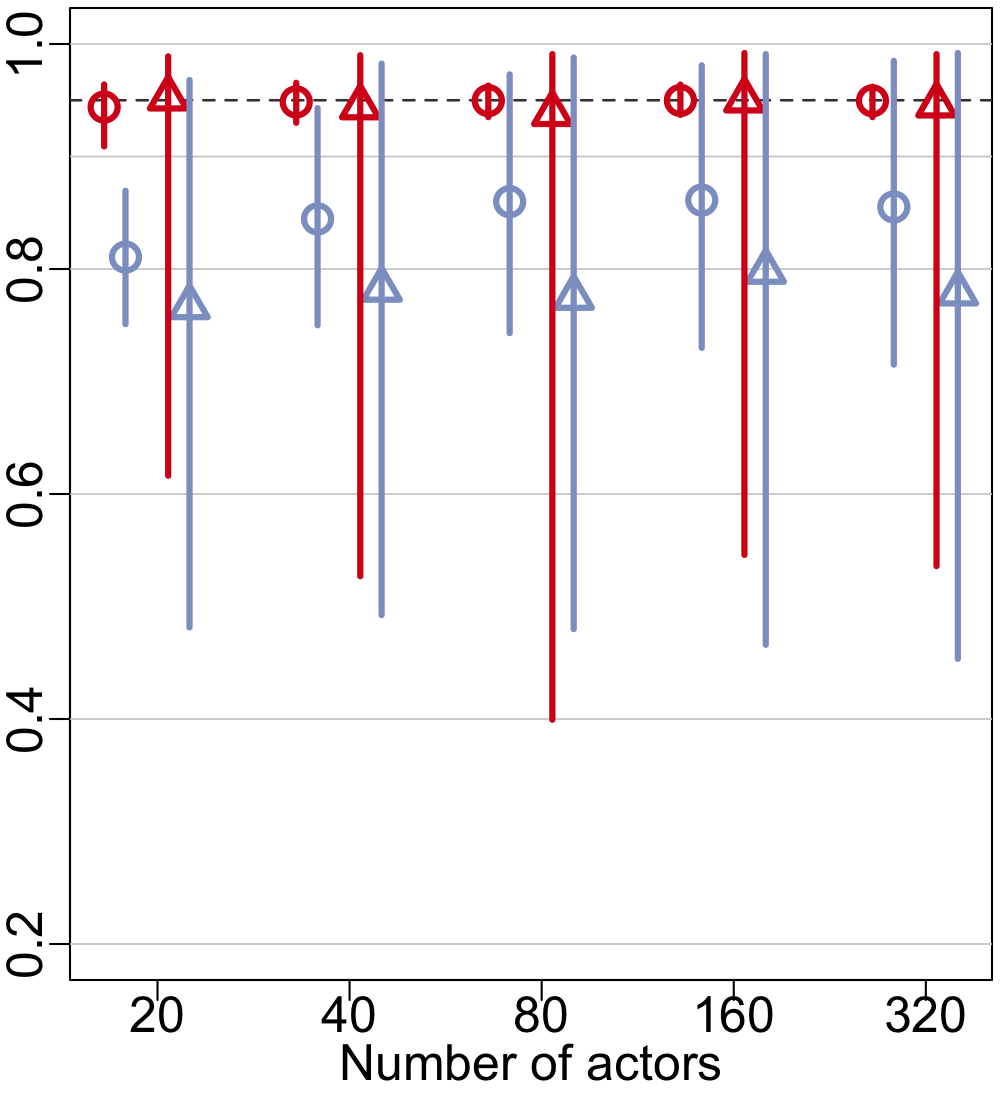} 
\hspace{-.22in}&\hspace{-.22in}
\includegraphics[width=.315\textwidth]{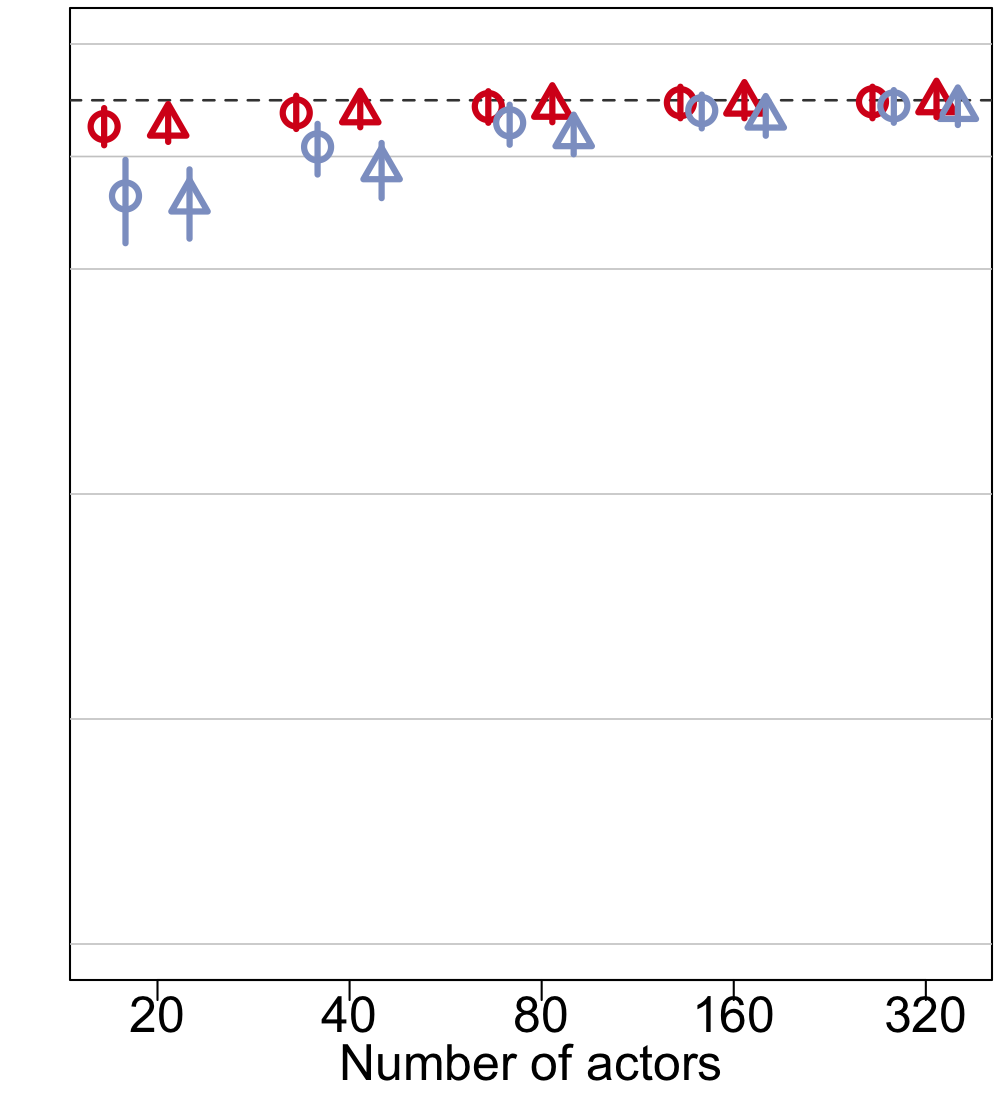} 
  \hspace{-.22in}&\hspace{-.22in}
\includegraphics[width=.315\textwidth]{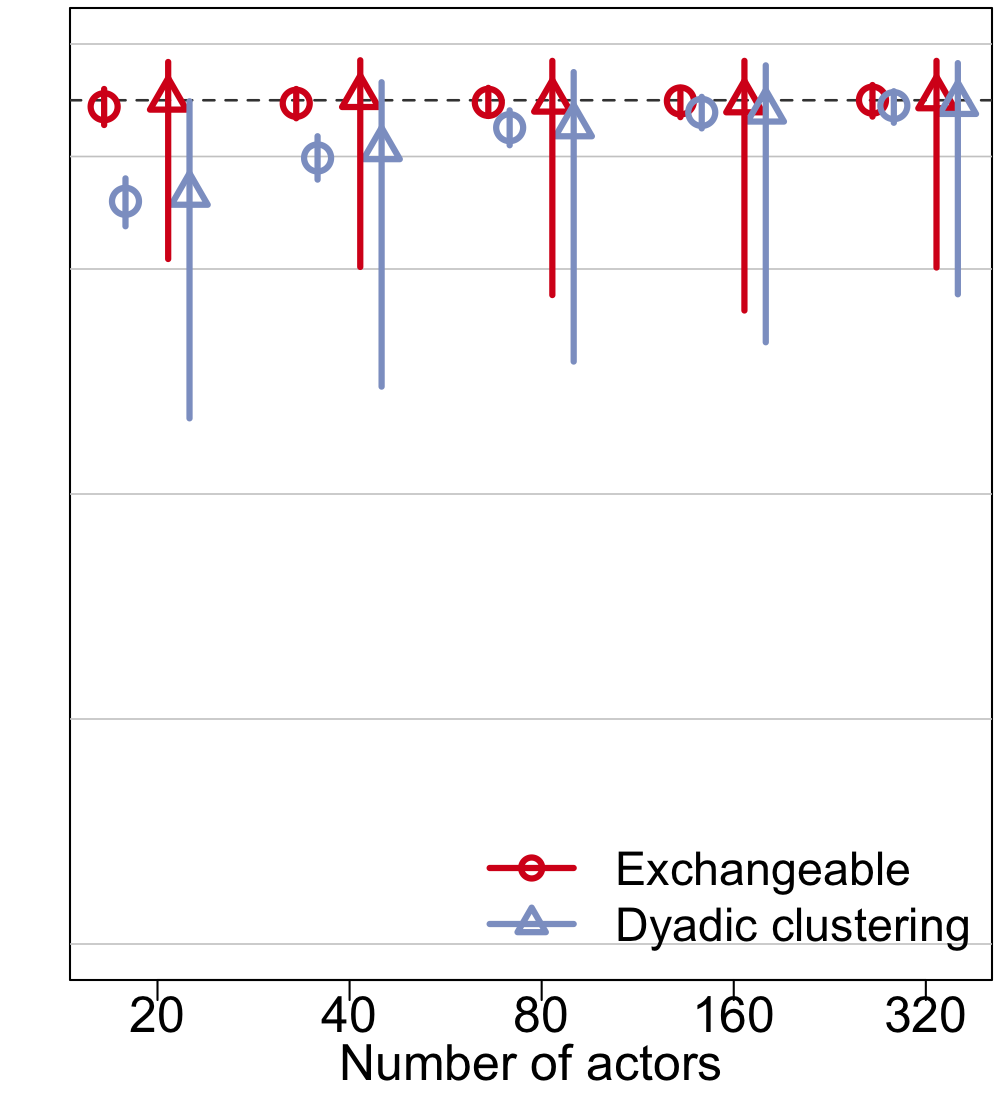} 
  \\
  \end{tabular}
  \caption{Estimated probability that true coefficient is in 95\% confidence interval for each of three covariates (binary, positive, and real-valued) when the errors are generated from exchangeable (circles) and non-exchangeable (triangles) models. Points denote mean estimated coverage and lines represent the middle 95\% of coverages for exchangeable (black) and dyadic clustering (gray)
  estimators. }
      \label{fig:sim_coverage}
\end{figure}

\section{Patterns in international trade}
\label{sec:trade}
We demonstrate the implications of using our exchangeable estimator in a study of international trade among $58$ countries over $R=T=20$ years.  These data were previously analyzed and made available by~\cite{westveld2011mixed}.  Following~\cite{westveld2011mixed},~\cite{ward2007persistent},  and~\cite{tinbergen1962shaping}, we use a modified gravity mean model to represent log yearly trade between each pair of countries as linear function of seven covariates in years 1981-2000.
\cite{westveld2011mixed}  propose a model -- which we refer to as the mixed effects model -- which explicitly decomposes the regression error term $\epsilon_{ijt}$ for each time period and pair of actors into time-dependent sender and receiver effects, resulting in 13 error covariance parameters which are estimated using a latent variable representation and Bayesian Markov chain Monte Carlo methodology. 
We propose estimating the gravity mean model using feasible generalized least squares, assuming the errors are jointly exchangeable. As noted in Section~\ref{subsec:exch_estimator}, the proposed approach estimates ten error covariance parameters.
We provide a method for efficient inversion of $\Omega_E$ in this setting in Section~\ref{app_Omega_inv}.

\begin{figure}[h]
    \centering
    \includegraphics[width=.8\textwidth]{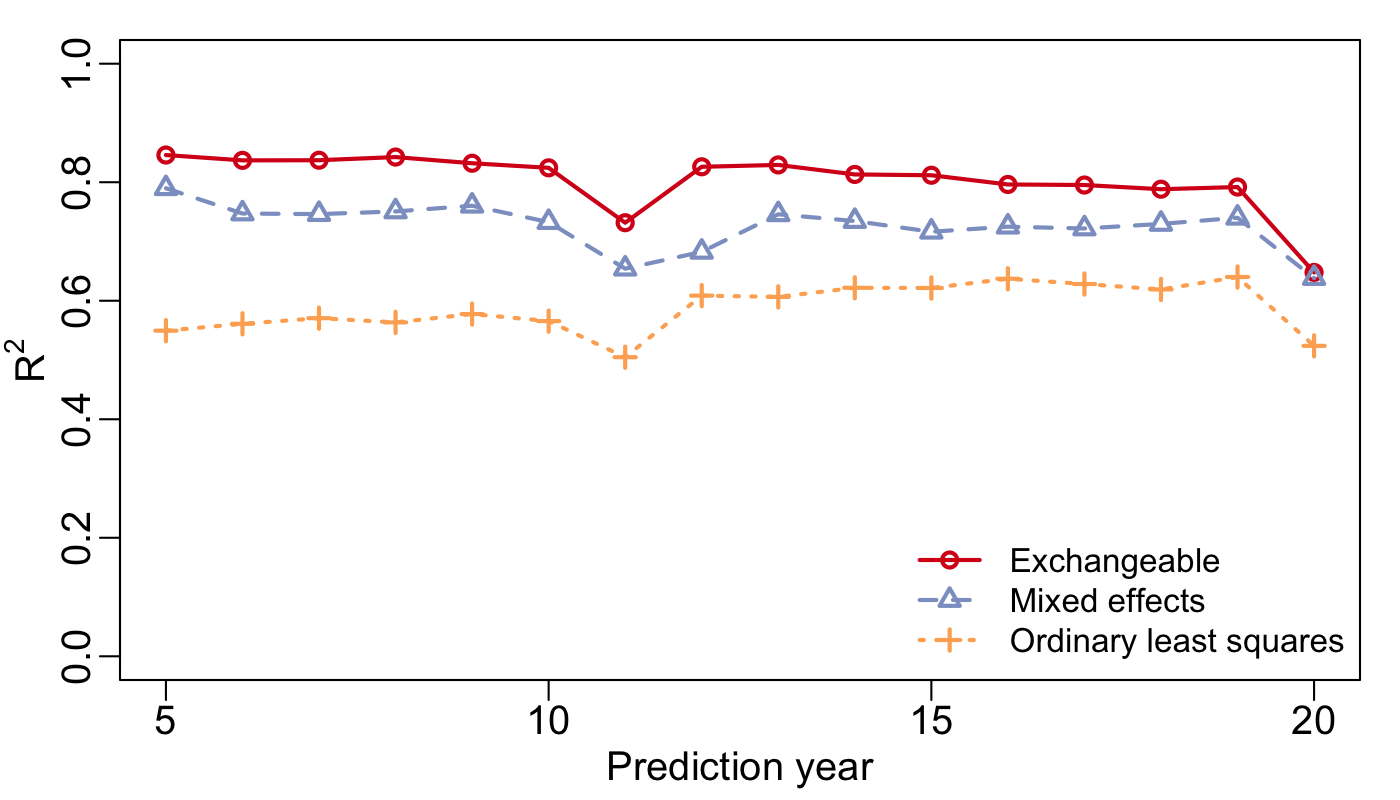} 
    \caption{$R^2$ when predicting one-year-ahead trade flows using exchangeable, mixed effects, and ordinary least squares approaches.}
    \label{fig:prediction}
\end{figure}

We compared the exchangeable and mixed effects approaches (and ordinary least squares as a baseline) in an out-of-sample prediction study. 
Here we estimated the regression coefficients using the first $K$ years of trade data for $K = 4,...,19$ and used the estimates to predict trade values in the following year.  Figure~\ref{fig:prediction} provides the coefficient of determination, $R^2$, for the three procedures when predicting trade flows in years 5 through 20.  There is a median increase in $R^2$ of about 10\% (30\%) when using the proposed exchangeable approach relative to the mixed effects approach (and ordinary least squares, respectively). The proposed approach performs better than the other approaches for all time periods, although the gap in performance decreases as $T$ increases. 
These results suggest the more parsimonious exchangeable approach represents the data better than the mixed effects model, and yet, the exchangeable approach runs in a small fraction of the time of the mixed effects approach. 
See Section~\ref{app_trade} for additional details.

We compared the coefficients of the ordinary least squares, mixed effects, and exchangeable approaches. The $R^2$ between the ordinary least squares and mixed effects coefficients is about $0.46$, while the exchangeable and mixed effects coefficients have an $R^2$ of $0.78$. About $40\%$ of the ordinary least squares coefficients (using dyadic clustering standard errors) were significantly different from the mixed effects coefficients, while only $1\%$ of the exchangeable coefficients were significantly different from the the mixed effects coefficients. Finally, the ordinary least squares coefficients have standard errors that are, on average, over $1.5$ times the standard errors of the exchangeable coefficients (with exchangeable standard errors). See Section~\ref{sec:trade_insample} for more details.
Together with the one-year-ahead prediction results, the coefficient and standard error comparisons suggest that the proposed approach
can revise ordinary least squares coefficients in the direction of a higher fidelity model, giving more precise estimates of the coefficients, while requiring few modeling decisions and with limited runtime penalty. Based on the success of the exchangeable approach in the trade data analysis and simulation study, we recommend that researchers use the feasible generalized least squares estimator of the coefficient vector $\beta$ as demonstrated here (unless an unbiased estimator of the coefficient vector is specifically desired).

\section{Discussion} \label{sec:concl}

The proposed exchangeable estimator leverages exchangeability for maximal symmetry (and thus maximal parsimony) in the covariance matrix of relational array $Y$.  The exchangeability assumption may not be appropriate when the true error covariances are substantially heterogeneous. 
We propose using a permutation test based on the dyadic clustering estimator for testing the hypothesis of exchangeable errors. The procedure consists of generating a null distribution of $\hat{V}_{DC}$ in \eqref{eq:varc} by randomly permuting the residual array in a manner consistent with exchangeability. If the observed estimator is extreme relative to the null distribution,  this suggests the errors are non-exchangeable. Details and simulations are available in Section~\ref{app_test}. 

\section*{Supplementary material}
\label{SM}
The supplementary material contains details about the estimator and analyses, proofs, further theoretical details, and a proposed test for exchangeability.


\bibliographystyle{apalike}
\bibliography{main.bib}

\clearpage
\include{supplement_biometrika}

\end{document}

%% file: supplement_biometrika.tex

\textbf{\Large \centering Supplementary material for ``Regression of exchangeable relational arrays''}




\newcommand{\x}{x}. 

\setcounter{section}{0}
\setcounter{page}{1}
\setcounter{figure}{0}  
\setcounter{equation}{0}

\renewcommand\thefigure{S\arabic{figure}}  
\renewcommand{\theequation}{S\arabic
{equation}}
\renewcommand{\thesection}{S\arabic{section}}
\renewcommand{\thetheorem}{S\arabic{theorem}}
\renewcommand{\thetable}{S\arabic{table}}
\renewcommand{\theproposition}{S\arabic{proposition}}
\renewcommand{\thedefinition}{S\arabic{definition}}


\vspace{.25in}
\noindent{{\large Contents\hfill Page No.}\vspace{.1in} \par}
\contentsline {section}{\ref{sec:appx_patterns}. Patterns of exchangeable covariance matrices}{2}
\contentsline {section}{\ref{sec:appx_ estimator}. Exchangeable estimator details}{3} 
\contentsline {section}{\ref{app_undir}. Undirected arrays}{3} 
\contentsline {section}{\ref{app_prop_proof}. Proof of Proposition~\ref{prop}}{4} 
\contentsline {section}{\ref{app_Omega_eigen_parrent}. Eigenvalues of  exchangeable covariance matrix}{5} 
\contentsline {section}{\ref{app:bias}. Proof of Theorem~\ref{thm:bias}: bias of the exchangeable and dyadic clustering estimators}{6} 
\contentsline {section}{\ref{app:theory}. Theoretical properties of the exchangeable estimator}{15} 
\contentsline {section}{\ref{app:AN}. Proof of asymptotic normality of ordinary least squares}{17} 
\contentsline {section}{\ref{app:consistency}. Proof of consistency of the exchangeable estimator}{22}  
\contentsline {section}{\ref{app:mse}. Proof of Theorem~\ref{thm_mse}: mean-square error of the exchangeable and dyadic clustering estimators}{25} 
\contentsline {section}{\ref{app_sim}. Simulation study}{35} 
\contentsline {section}{\ref{app_DC_inv}. Dyadic clustering covariance matrix invertibility}{37} 
\contentsline {section}{\ref{app_Omega_inv}. Efficient inversion of the exchangeable covariance matrix}{38} 
\contentsline {section}{\ref{app_trade}. Trade data analysis}{40} 
\contentsline {section}{\ref{app_test}. A test for exchangeability}{43}

\clearpage
\section{Patterns of exchangeable covariance matrices}
\label{sec:appx_patterns}
Figure \ref{exchCOVMAT} shows the structure of $\Omega_{E}$ for a relational matrix with four actors $\{A,B,C,D\}$ and $R=1$. 

\begin{figure}[h]
\centering
\begin{tabular}{m{.35\textwidth} m{.55\textwidth} }
 \includegraphics[width=.28\textwidth]{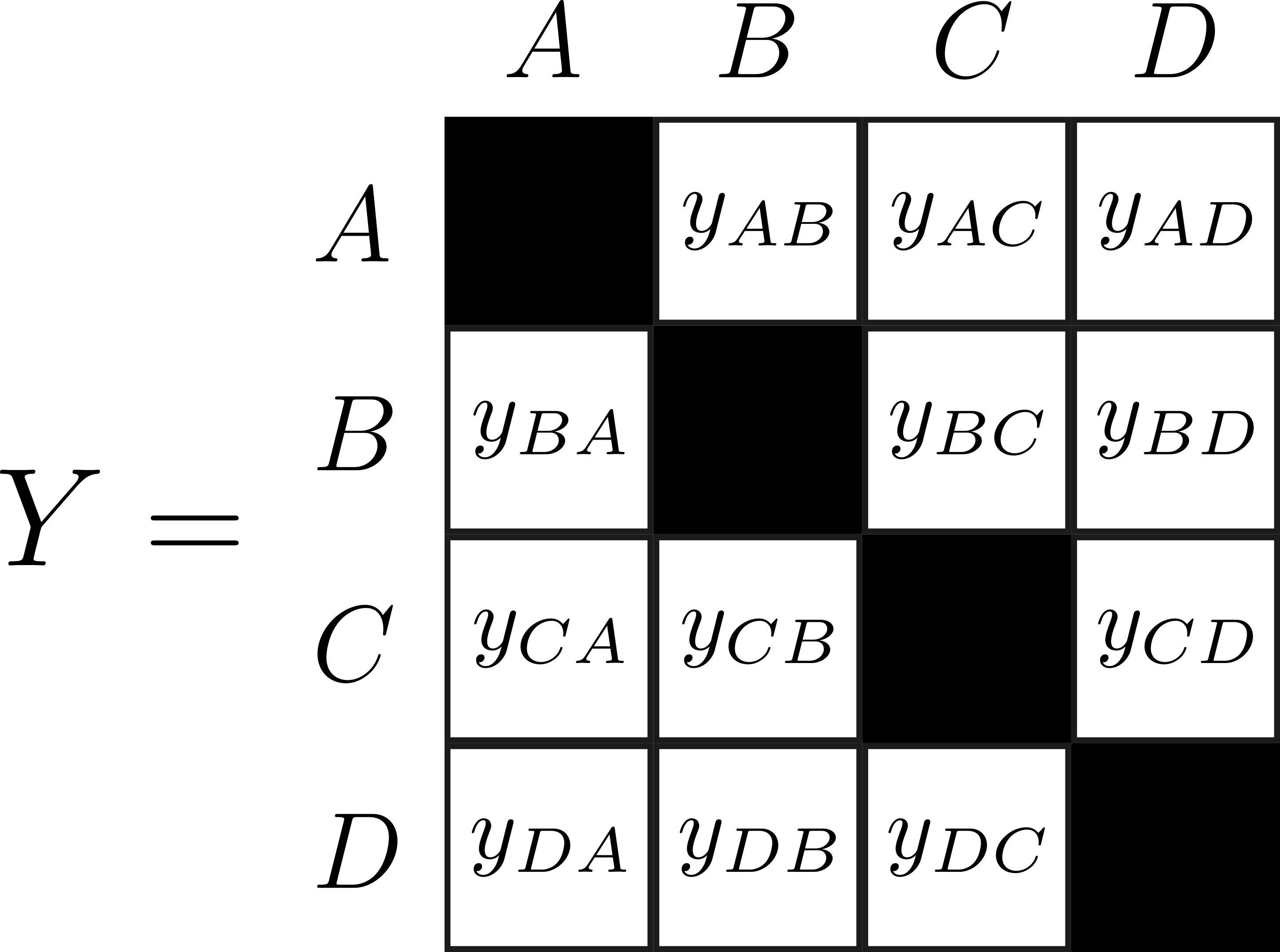} \hspace{3pt} & \hspace{3pt}
 \includegraphics[width=.53\textwidth]{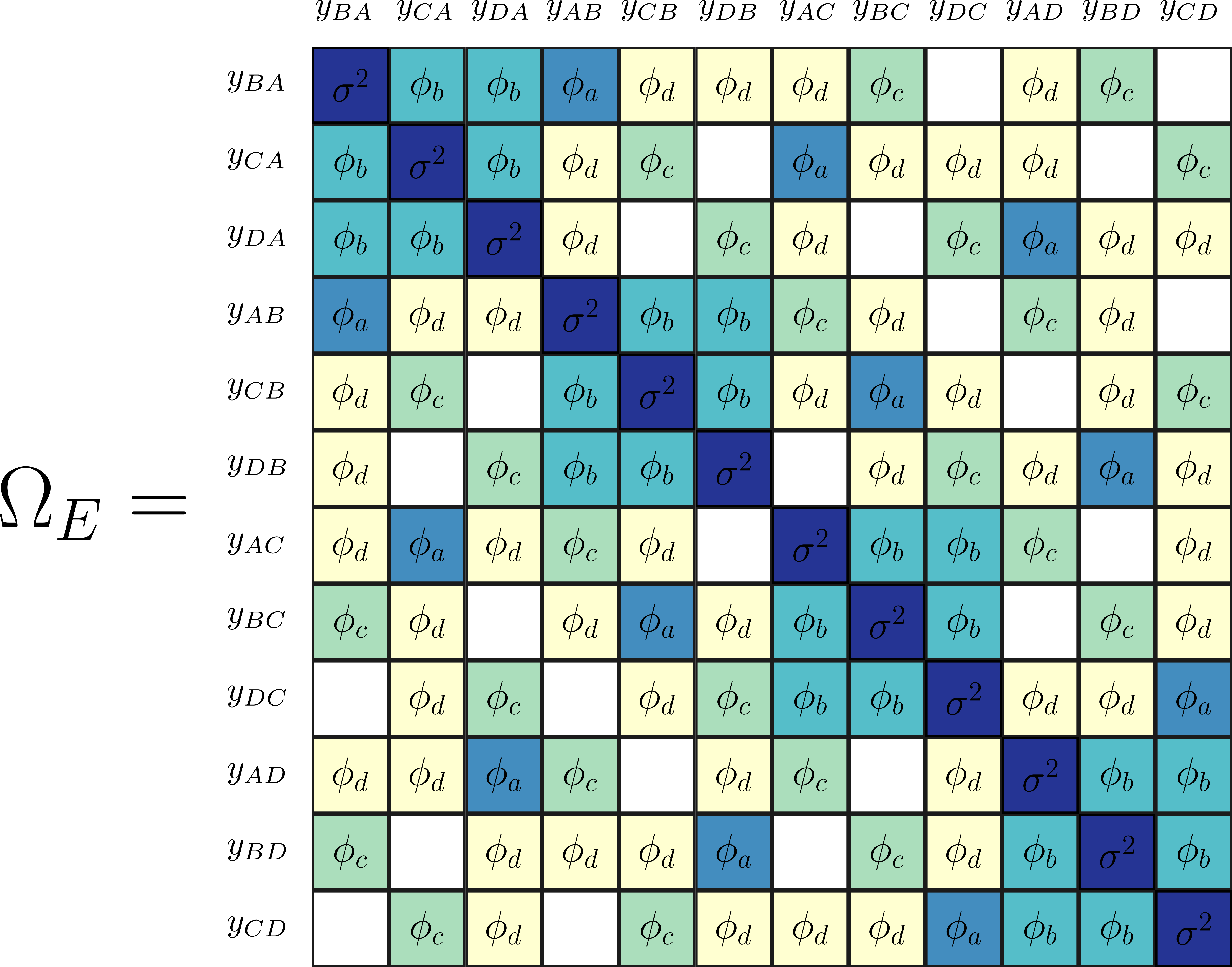}  \vspace{-1pt}
 \end{tabular}
\caption{Consider a matrix $Y$ containing the relations among four actors $\{A,B,C,D\}$ shown on the left.  Since the relation between an actor and itself is undefined, the diagonal entries (blacked out in the picture) are not regarded as part of $Y$.  Assuming joint exchangeability of the actors and that relations involving non-overlapping sets of actors are independent, the covariance matrix $\Omega_E$ contains five unique values. }
\label{exchCOVMAT}
\end{figure}

\section{Exchangeable estimator details}
\label{sec:appx_ estimator}
 The text gives an example for the estimator of $\phi_b^{(2)}$. We explicitly define the  empirical mean estimates used in the exchangeable estimator in \eqref{eq:exch.est} below:
{ \small
\begin{align*}
&\hat{\phi}_0^{(1)} = \frac{1} {Rn(n-1)} \sum_{r} \sum\limits_{i} \sum\limits_{j \neq i} e_{ijr}^2, \\
&\hat{\phi}_a^{(1)} = \frac{1} {Rn(n-1)} \sum_{r} \sum\limits_{i} \sum\limits_{j \neq i} e_{ijr} e_{jir}, \\
&\hat{\phi}_b^{(1)} = \frac{1} {Rn(n-1)(n-2)}\sum_{r} \sum\limits_{i} \sum\limits_{j \neq i} e_{ijr}  \Big( \sum\limits_{k \neq i} e_{ikr} - e_{ijr} \Big), \\
&\hat{\phi}_c^{(1)} = \frac{1} {Rn(n-1)(n-2)} \sum_{r} \sum\limits_{i} \sum\limits_{j \neq i} e_{ijr}  \Big( \sum\limits_{k \neq j} e_{kjr} - e_{ijr} \Big), \\
&\hat{\phi}_d^{(1)} =  \frac{1} {2Rn(n-1)(n-2)}\sum_{r} \sum\limits_{i} \sum\limits_{j \neq i} e_{ijr} \Big( \sum\limits_{k \neq i} e_{kir} + \sum\limits_{k \neq j} e_{jkr} - 2e_{jir} \Big), \\
&\hat{\phi}_0^{(2)} = \binom{R}{2}^{-1}\frac{1} {n(n-1)} \sum_{r\neq s} \sum\limits_{i} \sum\limits_{j \neq i} e_{ijr}e_{ijs}, \\
&\hat{\phi}_a^{(2)} = \binom{R}{2}^{-1} \frac{1} {n(n-1)} \sum_{r \neq s} \sum\limits_{i} \sum\limits_{j \neq i} e_{ijr} e_{jis}, \\
&\hat{\phi}_b^{(2)} = \binom{R}{2}^{-1}\frac{1} {n(n-1)(n-2)}\sum_{r\neq s} \sum\limits_{i} \sum\limits_{j \neq i} e_{ijr}  \Big( \sum\limits_{k \neq i} e_{iks} - e_{ijs} \Big), \\
&\hat{\phi}_c^{(2)} = \binom{R}{2}^{-1} \frac{1} {n(n-1)(n-2)} \sum_{r\neq s} \sum\limits_{i} \sum\limits_{j \neq i} e_{ijr}  \Big( \sum\limits_{k \neq j} e_{kjs} - e_{ijs} \Big), \\
&\hat{\phi}_d^{(2)} =  \binom{R}{2}^{-1} \frac{1} {2n(n-1)(n-2)} \sum_{r\neq s} \sum\limits_{i}  \sum\limits_{j \neq i} e_{ijr} \Big( \sum\limits_{k \neq i} e_{kis} + \sum\limits_{k \neq j} e_{jks} - 2e_{jis} \Big).
\end{align*}
}
Recall from the text that $\hat{\phi}_0^{(1)}$ is an estimator of ${\text{var}}(\xi_{ijr})$, $\hat{\phi}_a$ is an estimator of 
${\text{cov}}(\xi_{ijr}, \xi_{jir})$,
$\hat{\phi}_b^{(1)}$ is an estimator of 
${\text{cov}}(\xi_{ijr}, \xi_{kjr})$,
$\hat{\phi}_c^{(1)}$ is an estimator of 
${\text{cov}}(\xi_{ijr}, \xi_{ikr})$,
and $\hat{\phi}_d^{(1)}$ is an estimator of 
${\text{cov}}(\xi_{ijr}, \xi_{jkr})$ (and equivalently an estimator of ${\text{cov}}(\xi_{ijr}, \xi_{kir})$). The estimator $\hat{\phi}_i^{(2)}$, for $i = 0,a,b,c,d$, is the analogous estimator to $\hat{\phi}_i^{(1)}$, only when $r \neq s$.

\section{Undirected arrays}
\label{app_undir}
This section specializes the results presented in the manuscript to undirected  relational data. Consider the case when $R=1$ and suppose the relational data contains the relations among $n$ actors. The covariance of the errors $\Omega$ contains six possibly unique elements
\begin{align*} 
&\text{var}(\xi_{ijr})  = \theta_0^{(1)}, \ \ \ \ \ \text{cov}(\xi_{ijr},\xi_{kir})  = \theta_a^{(1)}, \ \ \ \ \ \text{cov}(\xi_{ijr},\xi_{klr})  = 0, \\
&\text{cov}(\xi_{ijr}, \xi_{ijs})  = \theta_0^{(2)}, \ \ \ \ \ \text{cov}(\xi_{ijr},\xi_{kis})  = \theta_a^{(2)}, \ \ \ \ \ \text{cov}(\xi_{ijr},\xi_{kls})  = 0, \ \ \ r \neq s.
\end{align*}
As in the directed case, we assume covariances corresponding to relations which share no common actor are zero.  We again propose estimating the remaining nonzero covariances using the corresponding means of residuals.

\section{Proof of Proposition~\ref{prop}}
\label{app_prop_proof}
\begin{proof} 
Consider a probability model for a directed relational array  $\Xi$ that satisfies the joint exchangeability and second moment criteria defined in the proposition.  
Now consider arbitrary permutations $\pi(\cdot)$ of the actor set $( 1,...,n  )$ and $\nu(\cdot)$ of the actor set $(1,...,R)$.  Exchangeability implies the that the probability distribution of $\xi_{ijr}$ is the same as $\xi_{\pi(i) \pi(j) \nu(r)}$. Thus, every entry in $\Xi$ is marginally identically distributed, and 
\[
\var(\xi_{ijr}) =  \var\{ \xi_{\pi(i) \pi(j) \nu(r)} \}.
\]
This is one of the entries in $\Omega$, the covariance matrix of $\Xi$.

We now move onto the covariances in $\Omega$; it remains to show that there are 11 unique entries beyond the variance. This result follows from two facts. The first fact is that exchangeability implies that the bivariate distribution of  the pair $(\xi_{ijr},\xi_{klr} )$ must be the same as bivariate distribution of $\{ \xi_{\pi(i)\pi(j) r}, \xi_{\pi(k)\pi(l) r} \}$, and thus, 
\[ 
\text{cov}(\xi_{ijr},\xi_{klr}) = \text{cov}\{ \xi_{\pi(i)\pi(j) r}, \xi_{\pi(k)\pi(l) r} \}.
\]
The second fact is that permutations preserve equality and inequality, meaning that if $i=j$, then $\pi(i) = \pi(j)$, and if $i \neq j$, then $\pi(i) \neq \pi(j)$, for any permutation $\pi$. We will enumerate the 11 patterns of indices that have the same bivariate distributions, as implied by the equality-inequality preservation of permutations, being mindful of the fact that the first two indices in $\xi_{ijr}$ refer to the same index set $( 1,\ldots,n )$.

First take the case when $i=k$ and $j=l$. When $r=s$, we have the variance previously discussed. When $r\neq s$, we obtain the second set of equivalent bivariate distributions (after the variance):
\[
(\xi_{ijr}, \xi_{ijs}),  \ \ (i = 1,\ldots n; j = 1, \ldots, n; i \neq j; r=1, \ldots, R; s=1, \ldots, R; r\neq s),
\]
where it is understood that $i \neq j$ due to the network setting. 
Moving on, consider the case when $i=k$ and $j \neq l$, noting that, by assumption, $i\neq j$ and $k \neq l$. Then, we obtain two sets with equivalent bivariate distributions, pertaining to $r=s$ and $r \neq s$:
\begin{align}
    (\xi_{ijr}, \xi_{ilr}),&  \ \ (i = 1,\ldots n; j = 1, \ldots, n; l=1,\ldots, n; l \neq j; r=1, \ldots, R), \nonumber \\
    (\xi_{ijr}, \xi_{ils}),&  \ \ (i = 1,\ldots n; j = 1, \ldots, n; l=1,\ldots, n; l \neq j; r=1, \ldots, R; s=1, \ldots, R; r\neq s). \nonumber 
\end{align}
Thus far, we have enumerated four sets of equivalent bivariate distributions under exchangeability. The remaining eight may be determined by cycling through combinations of indices in $(i,j)$ equal to, and not equal to, indices in $(k,l)$, for both $r=s$ and $r\neq s$. Table~\ref{tab:bivariate_indices} summarizes the twelve index patterns that result; the first and third row in the table have already been discussed. The second row, first column, states that pairs of entries in $\Xi$ with the index pattern $(\xi_{ijr}, \xi_{jir})$ form a set with the same bivariate distribution. 
\begin{table}[h]
    \caption{Twelve index patterns that form a set of equivalent bivariate distributions with the entry in $\Xi$ corresponding to indices $(i,j,r)$. The row labels correspond to those in Figure~\ref{fig:directeddyad} and the column labels indicate index sets that share the same third dimension ``slice'' of $\Xi$ with index set $(i,j,r)$.}
    \label{tab:bivariate_indices}
    \centering
    \begin{tabular}{ccc}
    & $s=r$ & $s \neq r$ \\
    
        (0) & $(i,j,r)$ & $(i,j,s)$  \\
        (a) &$(j,i,r)$ & $(j,i,s)$  \\
        (b) &$(i,l,r)$ & $(i,l,s)$  \\
        (c) &$(k,j,r)$ & $(k,j,s)$  \\
        (d) &$(k,i,r)$ & $(k,i,s)$  \\
        (e) &$(k,l,r)$ & $(k,l,s)$  \\
    \end{tabular}
\end{table}

Finally, one might at first think that we have missed the bivariate distribution of $(\xi_{ijr}, \xi_{jlr})$. However, we note that the bivariate distribution of random variables $(A,B)$ is equal to that of $(B,A)$, and thus, the bivariate distribution of $(\xi_{ijr}, \xi_{jlr})$ is captured in the first column of row $(d)$ of Table~\ref{tab:bivariate_indices}. Since we have enumerated all entries in $\Omega$, and there are twelve possibly unique sets with equivalent covariances (or variances), the result is shown for $R>1$. When $R=1$, it is sufficient to consider only the first column of Table~\ref{tab:bivariate_indices}, which establishes that there are at most six unique entries in $\Omega$ for $R=1$. 
\end{proof}

\section{Eigenvalues of exchangeable covariance matrix}
\label{app_Omega_eigen_parrent}

\subsection{Positive definite $\hat{\Omega}_E$}
\label{app_Omega_eigen}
Since the entries in the exchangeable covariance matrix estimator $\hat{\Omega}_E$ are empirical averages, it is possible the estimate $\hat{\Omega}_E$ is not positive definite.  Here we briefly investigate the constraints on the parameters that guarantee the resulting covariance matrix is positive definite when $R=1$.  Note that for computing the sandwich estimator variance of $\hat{\bbeta}$ and making inference on $\hat{\bbeta}$, positive definiteness of $\hat{\Omega}_E$ is not necessary.  However, if generalized least squares is employed to estimate $\beta$, the inverse of the covariance matrix estimator is required, and hence it is desirable that $\hat{\Omega}_E$ be positive definite. 

\subsection{Undirected  relational data}

We first focus on the undirected case, where the exchangeable covariance matrix contains two distinct nonzero entries: a variance $\sigma^2$ and a parameter $\phi$ in the off-diagonal representing the correlation between any pairs of relations that share an actor.  Below we consider the correlation matrix, rather than the covariance matrix, which contains only a single nonzero correlation value.  We denote this value by $a= \theta_a^{(1)}/\theta_0^{(1)}$. 

Based on a thorough empirical investigation, we conjecture that the exchangeable correlation matrix corresponding to an undirected set of relations among  $n$  actors, which has nonzero value $a$ in select off-diagonal entries, has exactly three eigenvalues as given below.

\begin{center}
\begin{tabular}{ rr } 
 Eigenvalue & Multiplicity \\  
 $1+2(n-2)a$ & $1$  \\ 
 $1-2a$ & $\frac{1}{2}n(n-3)$  \\
 $1+(n-4)a$ & $n-1 $ \\
\end{tabular}
\end{center}

The correlation matrix is positive definite if and only if all eigenvalues are positive.  Thus, if $a \in \left(1/({2(n-2)}), 1/2\right)$, the correlation matrix is positive definite.  Notice that the upper bound on $a$ does not vary with $n$.  Using the relation between $a$ and $( \theta_0^{(1)}, \theta_a^{(1)} )$, this constraint can be re-expressed as a constraint on the covariance parameters.

\subsection{Directed relational data}
Since matrices of the form of $\Omega_E$ form a five-dimensional linear subspace over square symmetric matrices, it is intuitive that the eigenspace of such matrices is limited.
We find empirically that square symmetric matrices of the form of $\Omega_E$ have five unique eigenvalues (when $R=1$). 
As in the undirected case, we again focus on the exchangeable correlation matrix which contains four nonzero off-diagonal elements $\{a,b,c,d\}$ corresponding, respectively, in placement to $\{\phi_a^{(1)},\phi_b^{(1)},\phi_c^{(1)},\phi_d^{(1)}\}$ in the exchange covariance matrix $\Omega_E$.  Note $a = \phi_a^{(1)}/\phi_0^{(1)}$, $b = \phi_b^{(1)}/\phi_0^{(1)}$, and so on.  Based on 
empirical studies, we conjecture the eigenvalues for the  exchangeable correlation matrix associated with a directed set of relations among   $n$  actors has exactly five eigenvalues as given below.

\begin{center}
\begin{tabular}{ rr} 
 
 Eigenvalue(s) & Multiplicity\\ 
 $1+a+(n-2)(b+c)+2(n-2)d$ & 1  \\ 
 $1+a-(b+c+2d)$ & $(n-1)(n-2)/2 - 1$\\
 $1-(a+b+c)+2d$ & $(n-1)(n-2)/2$\\
 $\{ (n-3)(b+c)-2d+2 \}/2 \pm {(\alpha + \beta)^{1/2}}/2$ & $n-1$\\
 
\end{tabular}
\end{center}
where  $\alpha = (c^2+b^2)(n^2-2n+1)+4d^2(n^2-6n+9)+2bc(1-n^2+2n)$ and $\beta = ad(8n-24)+(b+c)d(12-4n)+4a \{a-(b+c) \}$.  As in the undirected case, these constraints can be re-expressed as constraints on the original five covariance parameters.

\section{Proof of Theorem~\ref{thm:bias}: bias of the exchangeable estimator}
\label{app:bias}

\subsection{Notation and proof preliminaries}
In this section we prove that the dyadic clustering estimator is biased downwards, and that this bias is greater in absolute value than the bias of the exchangeable estimator, in the case when $R=1$.

First, we make some definitions. Recall that the dyadic clustering estimator is based on the outer product of residuals, with  entries corresponding to pair of relations that do not share an actor set to zero. This  operation may be viewed as a projection onto a vector space of square, symmetric matrices of dimension $n(n-1)$ with zeros in the appropriate entries.  The dimension of this vector space is $O(n^3)$. We define the vector space below:

\begin{definition}
\label{def:dc_space}
{Vector space $\s{D}$: } Let $M_{jk,lm}$ be the square matrix of dimension $n(n-1)$ with a `1' in the entry corresponding to relations $jk$ and $lm$, and zeros elsewhere. Then, 
\[\s{D} = \text{span}\left\{  M_{jk,lm} : (j,k) \cap (l, m) = \varnothing  \right\}. \]
We define the orthogonal projection of matrix $A$ onto $\s{D}$ as $P_\s{D} \left( A \right )$.
\end{definition}

A similar definition to Definition~\ref{def:dc_space} may be made concerning the space of covariance matrices corresponding to jointly exchangeable random variables, as defined below. 
\begin{definition}
\label{def:e_space}
{Vector space $\s{E}$: } Let $\s{S}_{i}$ be an indicator matrix of entries in square, symmetric matrices of dimension $n(n-1)$ corresponding to pairs of relations that share an actor in the $i \text{th}$ manner, for $i = 1,\ldots,5 $. Then, 
\[\s{E} = \text{span}\left( \s{S}_i : i=1, \ldots, 5 \right). \]
We define the orthogonal projection of matrix $A$ onto $\s{E}$ as $P_\s{E} \left( A \right)$.
\end{definition}

It will be useful to write the dyadic clustering and exchangeable estimators as functions of the projections in Definitions~\ref{def:dc_space} and \ref{def:e_space}, supposing that $X$ is orthogonal, that is, $X^TX = f(n) I$, where $f(n)$ is a function of the number of actors $n$. The dyadic clustering estimator of $\var( \hat{\beta}_k )$ is then
{\small
\begin{align}
( \hat{V}_{DC} )_{kk} &= f(n)^{-2} X_{k}^T P_{\s{D}} \left\{ (y - X\hat{\beta})  (y - X\hat{\beta})^T \right\} X_{k},  \label{eq:dc_trace} \\
&= f(n)^{-2} \text{tr} \left\{ (y - X\hat{\beta})  (y - X\hat{\beta})^T  P_{\s{D}} \left(  X_{k} X_k^T \right) \right\},  \nonumber 
\end{align}
}
where $X_k$ is column $k$ of covariate matrix $X$. The second equality in \eqref{eq:dc_trace} results from the application of a property of inner products, which we apply to the inner product over matrices, where, for any projection operator $P( . )$, we have that $\text{tr} \left\{ A P ( B ) \right\} = \text{tr} \left\{ P ( A ) B \right\} = \text{tr} \left\{ P ( A ) P( B ) \right\}$. An analogous definition to \eqref{eq:dc_trace} holds for the exchangeable estimator of $\var( \hat{\beta}_k )$ when replacing the projections onto $\s{D}$ in \eqref{eq:dc_trace} with projections onto $\s{E}$. 

Finally, to prove Theorem~\ref{thm:bias}, we need some relationships between the error covariance parameters $( \phi_3^{(1)}, \phi_4^{(1)}, \phi_5^{(1)} )$ when $\Omega$ is positive definite. We summarize these relationships in the inequalities in the following proposition. We drop the superscript `$(1)$' to lighten notation. 

\begin{proposition}
\label{prop:eval_relations}
Let $\Omega \in \s{E}$ be an exchangeable correlation matrix, with $\Omega = {I} + \sum_{i=2}^5 \phi_i \s{S}_i$. Then, 
\begin{align}
0 \le \phi_3 + \phi_4 \le 1, \nonumber \\
    \phi_3^2 + \phi_5^2 \le \phi_3, \nonumber \\
    \phi_4^2 + \phi_5^2 \le \phi_4. \nonumber \\
        \phi_3^2 + \phi_4^2 + 2 \phi_5^2 \le 1, \nonumber 
\end{align}
\end{proposition}

\begin{proof}
The inequalities result from the requirement that eigenvalues of $\Omega$ are positive for all $n$. Taking $n$ arbitrarily large, the eigenvalues of $\Omega$ are
\begin{align}
    n(\phi_3 + \phi_4 + 2\phi_5) & \ge 0, \nonumber \\
    1 + \phi_2 - \phi_3 - \phi_4 - 2\phi_5 & \ge 0, \nonumber \\
    1 - \phi_2 - \phi_3 - \phi_4 + 2\phi_5 & \ge 0, \nonumber \\
    n \left[ \phi_3 + \phi_4  \pm \left\{ (\phi_3 - \phi_4 )^2 + 4 \phi_5^2 \right\}^{1/2} \right] & \ge 0. \nonumber 
\end{align}
The second and third eigenvalues give that $\phi_3 + \phi_4  \le 1$. The fourth and fifth eigenvalues imply $\phi_3 + \phi_4  \ge 0$.

Using the fourth and fifth eigenvalue and rearranging,
\begin{align}
   4 \phi_5^2 + (\phi_3 -  \phi_4)^2&\le (\phi_3 +  \phi_4)^2 \label{eq:phi5_bound} \\
   4 \phi_5^2 & \le 4 \phi_3 \phi_4.  \nonumber 
\end{align}
Then, we have that
\begin{align}
    \phi_3^2 + \phi_5^2 &\le \phi_3^2 + \phi_3 \phi_4, \label{eq:last_bound} \\
    &\le \phi_3^2 + \phi_3(1-\phi_3)\nonumber \\
    &\le \phi_3, \nonumber
\end{align}
which establishes the second result. The third result follows from an analogous argument to \eqref{eq:last_bound} for $\phi_4^2 + \phi_5^2$. The final result follows from summing the previous two, such that
\begin{align}
    \phi_3^2 + \phi_4^2 + 2\phi_5^2 &\le \phi_3 + \phi_4 \le 1. \nonumber
\end{align}

\end{proof}

\subsection{Proof of Theorem~\ref{thm:bias}}
In this proof, we start by establishing the expression for the true variance of $\hat{\beta}_2$. Then, we write the bias of the dyadic clustering and exchangeable  estimators as a function of the variance parameters of $x$ and $\xi$. Then, we write the bias of the dyadic clustering estimator as a function of the bias of the exchangeable estimator, and then bound the absolute value of the bias of the exchangeable estimator, which gives the result. Throughout, we assume that the bias of the exchangeable and dyadic clustering estimators, scaled by $n^2$, are $O(1)$. We verify this assumption at the end of the proof.
Without loss of generality, we prove for the scaled problem where $E(xx^T) = \Psi = \sum_{i=1}^5 \psi_i \s{S}_i$ with $\psi_1 = 1$.

We begin by expressing the true variance of $\hat{\beta}_2$. We require that $n^{1/2}( \hat{\beta}_2 - \beta_{2})$ converges in distribution to a normal random variable with finite variance, which may be expected from Proposition 3.2 in \cite{tabord2017inference}. By assumption, $\hat{\beta}_2 = (x^T x)^{-1} x^T y$. The true variance of $\hat{\beta}_2$ is then
\begin{align}
    n \var( \hat{\beta}_2 ) &= n (x^T x)^{-2} x^T \Omega x, \nonumber \\
    &= O(n^{-3}) x^T \Omega x + O(n^{-1})= O(1), \nonumber
\end{align}
where we use that $x^T x= n^2 + O(n^{-1})$ by assumption. Using \eqref{eq:dc_trace}, it will be useful to rewrite the true variance
\begin{align}
    n \var( \hat{\beta}_2 ) &= n (x^T x)^{-2} \text{tr} \left\{  \Omega P_{\s{D}} \left(  x x^T \right) \right\}, \nonumber \\
    &= n (x^T x)^{-2} \text{tr} \left\{ \Omega P_{\s{E}} \left(  x x^T \right) \right\}, \nonumber
\end{align}
where we use the fact that $\Omega$ is a member of both spaces $\s{D}$ and $\s{E}$ by assumption. 

We now write the expectation of the dyadic clustering estimator, taking only the leading terms. 
First, the expectation
\[ \E\{ (y - X\hat{\beta})  (y - X\hat{\beta})^T \mid X\} = (I - P_X) \Omega (I - P_X) ,\] 
where $P_X$ is the projection onto the column space of the covariate matrix $X$. By assumption, we have that
\begin{align}
    P_X &= \{ n(n-1)\}^{-1} {1} {1}^T + (x^T x)^{-1} xx^T, \nonumber \\
    &= \{ n(n-1)\}^{-1} \left\{ {1} {1}^T + (1 + O(n^{-1})) xx^T\right\}, \nonumber \\
    &= \{ n(n-1)\}^{-1} \left( {1} {1}^T +  xx^T\right) + O(n^{-1}), \nonumber
\end{align}
 where ${1} $ is a vector of 1's of appropriate length (in this case, length $n(n-1)$). Using the expression of the dyadic clustering estimator in \eqref{eq:dc_trace}, its expectation is
{\small
\begin{align}
n^2\E \{ ( \hat{V}_{DC} )_{22} \}  
&=n^2(x^T x)^{-2}\text{tr} \Big\{ \Big( I - P_X \Big) \Omega \Big( I - P_X \Big) P_{\s{D}} \left(  xx^T \right)  \Big\} , \nonumber \\
&=n^{-2} \text{tr} \Big\{ \Big( I - P_X \Big) \Omega \Big( I - P_X \Big) P_{\s{D}} \left(  xx^T \right)  \Big\} + O(n^{-1}), \nonumber \\
&=n^2 \var ( \hat{\beta}_2 ) +  n^{-2} \text{tr} \Big\{  P_X \Omega P_X P_{\s{D}} \left(  xx^T \right)  \Big\} - 2 n^{-2} \text{tr} \Big\{ \Omega P_X P_{\s{D}} \left(  xx^T \right)  \Big\} + O(n^{-1}),  \nonumber \\
&=n^2\var( \hat{\beta}_2 )  
 + n^{-6} \text{tr} \Big\{ {1} {1}^T \Omega  xx^T  P_{\s{D}} \left(  xx^T \right)  \Big\}  - 2n^{-4} \text{tr} \Big\{ \Omega  {1} {1}^T   P_{\s{D}} \left(  xx^T \right)  \Big\}
\ldots   \label{eq:dc_bias_expanded} \\
& \hspace{.5in} + n^{-6} \text{tr} \Big\{  xx^T \Omega  xx^T  P_{\s{D}} \left(  xx^T \right)  \Big\} - 2n^{-4} \text{tr} \Big\{ \Omega  xx^T  P_{\s{D}} \left(  xx^T \right)  \Big\} \ldots \nonumber \\
& \hspace{.5in} + n^{-6} \text{tr} \Big\{ {1} {1}^T \Omega  11^T  P_{\s{D}} \left(  xx^T \right)  \Big\}
+ O(n^{-1}).  \nonumber
\end{align}
}
\noindent The second term in \eqref{eq:dc_bias_expanded} is zero in probability, as the row sums of $\Omega$ are all the same, and are $O(n)$. Thus, 
\[ {1}^T \Omega x = O(n) {1}^Tx. \]
Then, the second term is
\begin{align}
   n^{-6} \text{tr} \Big\{ {1} {1}^T \Omega  xx^T  P_{\s{D}} \left(  xx^T \right)  \Big\} &= (n^{-2} {1}^T x ) ( n^{-3} x^T  P_{\s{D}} \left(  xx^T \right) {1} ) = O(n^{-1}), \label{eq:dc_term22}
\end{align}
since $n^{-2}1^T x = O(n-1)$ by assumption,
which implies that the second term vanishes relative to the rest.

The third term in \eqref{eq:dc_bias_expanded} may be written
\begin{align}
- 2 n^{-4} \text{tr} \Big\{ \Omega  {1} {1}^T   P_{\s{D}} \left(  xx^T \right)  \Big\}
&= - 2  n^{-4} c_n {1}^T   P_{\s{D}} \left(  xx^T \right)  {1},\nonumber  
\end{align}
where $c_n = O(n)$ is the row sum of $\Omega$, which again is the same for all rows. 
Further, using the Cauchy-Scharz inequality, 
 \[ {1}^T   P_{\s{D}} \left(  xx^T \right)  {1} = \sum_{jk, lm \in \Theta_0} x_{jk}x_{lm} \le \left( \sum_{jk, lm \in \Theta_0} x^2_{jk}x^2_{lm}  \right)^{1/2} = || P_{\s{D}}\left( xx^T \right) ||_F. \]
Substituting, the third term is 
\begin{align}
- 2 n^{-3} \text{tr} \Big\{ \Omega  {1} {1}^T   P_{\s{D}} \left(  xx^T \right)  \Big\}
&= n^{-3} || P_{\s{D}}\left( xx^T \right) ||_F. \nonumber
\end{align}
We now bound the order of $|| P_{\s{D}}\left( xx^T \right) ||_F$. By assumption, $x^T x = n(n-1) + O(n^{-1})$, and thus
\[ || P_{\s{D}}\left( xx^T \right) ||^2_F = \sum_{jk, lm \in \Theta_0} x_{jk}^2 x_{lm}^2 = O(n^3). \]
Using this bound, the third term is 
\begin{align}
- 2 n^{-4} \text{tr} \Big\{ \Omega  {1} {1}^T   P_{\s{D}} \left(  xx^T \right)  \Big\}
&= O(n^{-3/2}), \label{eq:dc_term3}
\end{align}
and the third term vanishes relative to the other terms in \eqref{eq:dc_bias_expanded}. Further, the sixth term in \eqref{eq:dc_bias_expanded} is of the same order, 
\begin{align}
    n^{-6} \text{tr} \Big\{ {1} {1}^T \Omega  11^T  P_{\s{D}} \left(  xx^T \right)  \Big\} &= n^{4} c_n \text{tr} \Big\{ {1} 1^T  P_{\s{D}} \left(  xx^T \right) \Big\} = O(n^{-3/2}), \label{eq:dc_term6}
\end{align}
and may also be neglected.

Using \eqref{eq:dc_term22}, \eqref{eq:dc_term3}, and \eqref{eq:dc_term6}, the expression for the bias of the dyadic clustering estimator based on \eqref{eq:dc_bias_expanded} becomes
\begin{align}
n^2\text{Bias}\left\{ (\hat{V}_{DC} )_{22} \right\} 
&= 
 n^{-6} \text{tr} \Big\{  xx^T \Omega  xx^T  P_{\s{D}} \left(  xx^T \right)  \Big\} - 2 n^{-4} \text{tr} \Big\{ \Omega  xx^T  P_{\s{D}} \left(  xx^T \right)  \Big\}  \label{eq:dc_bias_expanded2} 
+ O(n^{-1}). 
\end{align}
Since, by the properties of projections, $|| P_{\s{E}} \left(  xx^T \right) ||_F \le || P_{\s{D}} \left(  xx^T \right) ||_F$ since $\s{E} \subset \s{D}$,  we may write the bias of the exchangeable estimator by analogy: 
\begin{align}
n^2\text{Bias}\left\{ (\hat{V}_{E} )_{22} \right\}  
&= 
 n^{-6} \text{tr} \Big\{  xx^T \Omega  xx^T  P_{\s{E}} \left(  xx^T \right)  \Big\} - 2 n^{-4} \text{tr} \Big\{ \Omega  xx^T  P_{\s{E}} \left(  xx^T \right)  \Big\}  \label{eq:e_bias_expanded2} 
+ O(n^{-1}). 
\end{align}

We now explicitly determine the value of the exchangeable bias in \eqref{eq:e_bias_expanded2}. By assumption,
\begin{align}
     |\Theta_i |^{-1} x^T \s{S}_i x = \psi_i + O(n^{-1/2}), \quad (i = 1,\ldots, 5 ), \nonumber
\end{align}
recalling that $E( xx^T ) = \Psi = \sum_{i = 1}^5 \psi_i \s{S}_i$. Then, the projection of $xx^T$ onto the exchangeable space is
\begin{align}
    P_{\s{E}} \left(  xx^T \right) =  \sum_{i = 1}^5 \psi_i \s{S}_i + O(n^{-1/2}), \label{eq:pe_a}
\end{align}
and the scaled true variance is
\begin{align}
    n\text{var}(\hat{\beta}_2) = n^{-3} x^T \Omega x &= \sum_{i=1}^5 \frac{| \Theta_i |}{n^3} \psi_i \phi_i + O(n^{-1/2}). \label{eq:v_true_a}
\end{align}
Using \eqref{eq:pe_a} and \eqref{eq:v_true_a}, the first term in the exchangeable bias in \eqref{eq:e_bias_expanded2} is
\begin{align}
 n^{-6} \text{tr} \Big\{  xx^T \Omega  xx^T  P_{\s{E}} \left(  xx^T \right)  \Big\} &= (n^{-3}x^T \Omega x)\left\{n^{-3}x^T  P_{\s{E}} \left(  xx^T \right)  x \right\}, \nonumber \\
 &= \left(  \sum_{i=1}^5 \frac{| \Theta_i |}{n^3} \psi_i \phi_i  \right) \left(  \sum_{j=1}^5 \frac{| \Theta_j |}{n^3} \psi_j^2  \right)  + O(n^{-1/2}). \nonumber 
\end{align}
The second term in the exchangeable bias in \eqref{eq:e_bias_expanded2} is 
\begin{align}
- 2 n^{-4} \text{tr} \Big\{ \Omega  xx^T  P_{\s{E}} \left(  xx^T \right)  \Big\} &= - 2 n^{-4} \sum_{i = 1}^5 \sum_{j=1}^5 \phi_i \psi_j x^T \s{S}_i \s{S}_j x. \nonumber 
\end{align}
Now we must determine $n^{-4} x^T \s{S}_i \s{S}_j x$ for each $i$ and $j$. Recalling that $\s{S}_{i}$ is a matrix that indicates pairs of relations $jk$ and $lm$ that share an actor in the $i^{th}$ manner, 
\begin{align}
    x^T  \s{S}_i \s{S}_j x &= \sum_{(rs, tu) \in \Theta_i} \sum_{(rs, ab) \in \Theta_j} x_{tu} x_{ab}. \label{eq:sumnew_sisj}
\end{align}
Then, the sums \eqref{eq:sumnew_sisj} may be determined by examining the relationships between the actors in relations $tu$ and $ab$. First, for any $i$ and $j$, $tu$ must share at least one actor with $ab$, as the expectation  $\E( x_{tu} x_{ab} )$ is zero otherwise. This result implies that every sum in \eqref{eq:sumnew_sisj} is at most $O(n^4)$.  
As an example, we consider $i=3$ and $j=3$. In this case, $tu$ is of the form $ru$ and $ab$ is of the form $rb$. Whenever $b \neq u$, the expectation $\E( x_{ru} x_{rb} ) = \psi_3$. When $b=u$, the expectation $\E( x_{ru} x_{rb} ) = \psi_1$, although there are an order of magnitude fewer of these terms. Finally, similar to the expectation argument, the variance of the sums in \eqref{eq:sumnew_sisj} consists of at most $O(n^6)$ finite terms, and thus tends to zero. Then, the sum tends to the limit of its expectation. Summarizing, we have that
\begin{align}
    n^{-4}\sum_{(rs, tu) \in \Theta_3} \sum_{(rs, ab) \in \Theta_3} x_{tu} x_{ab} &=
 n^{-4} \sum_{(rs, tu) \in \Theta_3} \sum_{(rs, ab) \in \Theta_3} \E( x_{tu} x_{ab} ) + O(n^{-1/2}), \nonumber \\
 &= \frac{n}{2} n^{-4} |\Theta_3|\psi_3 + n^{-3} |\Theta_1| \psi_1 +  O(n^{-1/2}), \nonumber \\
 &= \psi_3 + O(n^{-1/2}). \nonumber
 \end{align}
Similar counting operations give the following result, for $i \le j$,
\begin{align}
    n^{-4}x^T  \s{S}_i \s{S}_j x  
    &= \begin{cases} 
    \psi_3  + O(n^{-1/2}), &i=j=3 \\
    \psi_4  + O(n^{-1/2}), &i=j=4 \\
    \psi_3 + \psi_4  + O(n^{-1/2}), &i=j=5 \\
    \psi_5  + O(n^{-1/2}), &i=3, \ j=5 \\
    \psi_5  + O(n^{-1/2}), &i=4, \ j=5 \\
    O(n^{-1/2}) &\mbox{otherwise. } 
    \end{cases} \nonumber
\end{align}
Substituting, we have the following expression for the bias of the exchangeable estimator
\begin{align}
    n^2\text{Bias}\left\{ (\hat{V}_{E} )_{22} \right\}  
&= (\phi_3 \psi_3 + \phi_4 \psi_4 + 2\phi_5 \psi_5) (\psi_3^2 + \psi_4^2 + 2\psi_5^2) \ldots \label{eq:bias_e_asymptotic}  \\
 &\hspace{.5in} -2 \left\{ \phi_3(\psi_3^2 + \psi_5^2) +  \phi_4(\psi_4^2 + \psi_5^2) + 2\phi_5 \psi_5(\psi_3 + \psi_4)  \right\}
 + O(n^{-1/2}). \nonumber
\end{align}

We now return to the bias of the dyadic clustering estimator in \eqref{eq:dc_bias_expanded2}. The first term in the bias of the dyadic clustering estimator is 
{\small 
\begin{align}
 n^{-6} \text{tr} \Big\{  xx^T \Omega  xx^T  P_{\s{D}} \left(  xx^T \right)  \Big\} 
&= \left(  \sum_{i=1}^5 \frac{| \Theta_i |}{n^3} \psi_i \phi_i  \right) \left\{ n^{-3} x^T  P_{\s{D}} \left(  xx^T \right) x  \right\}  + O(n^{-1/2}), \nonumber \\
&= \left(  \sum_{i=1}^5 \frac{| \Theta_i |}{n^3} \psi_i \phi_i  \right) \left( n^{-3} \sum_{i=1}^5 \sum_{jk, lm \in \Theta_i}x_{jk}^2 x_{lm}^2 \right)  + O(n^{-1/2}). \label{eq:dc_sum}
\end{align}
}
As with the exchangeable estimator, the sum in \eqref{eq:dc_sum} converges to the limit of its expectation, as the variance tends to zero. Then, using Iserlis' theorem, 
\begin{align}
 n^{-3} \sum_{i=1}^5 \sum_{jk, lm \in \Theta_i}x_{jk}^2 x_{lm}^2  &= n^{-3} \sum_{i=1}^5 \sum_{jk, lm \in \Theta_i}\E( x_{jk}^2 x_{lm}^2 ) + O(n^{-1/2}), \nonumber \\
 &= n^{-3} \sum_{i=1}^5 |\Theta_i| (1 + 2 \psi_i^2) + O(n^{-1/2}), \nonumber
\end{align}
using the assumption that $\psi_1 = 1$.
Then, the first term in the bias of the dyadic clustering estimator in \eqref{eq:dc_bias_expanded2} is 
\begin{align}
    n^{-6}  \text{tr} \Big\{  xx^T \Omega  xx^T  P_{\s{D}} \left(  xx^T \right)  \Big\}
    &= \left(  \sum_{i=1}^5 \frac{| \Theta_i |}{n^3} \psi_i \phi_i  \right) \left( 4 + 2\sum_{j=1}^5 \frac{| \Theta_j |}{n^3} \psi_j^2  \right)  + O(n^{-1/2}). \nonumber
\end{align}
We now turn to the second term in the bias of the dyadic clustering estimator in \eqref{eq:dc_bias_expanded2}. 
\begin{align} 
- 2 n^{-4}  x^T  P_{\s{D}} \left(  xx^T \right) \Omega  x &= - 2 n^{-4} \sum_{i=1}^5 \sum_{j=1}^5\phi_i x^T \s{S}_i \left( xx^T \circ \s{S}_j \right) x, \label{eq:dc_term33}
\end{align}
where `$\circ$' represents element-wise multiplication. As with the exchangeable estimator, the sum in \eqref{eq:dc_term33} has variance that tends to zero, and converges to its expectation. Writing in sum form and using Isserlis' theorem,
\begin{align} 
n^{-4}  x^T \s{S}_i \left( xx^T \circ \s{S}_j \right) x &= n^{-4} \sum_{rs, tu \in \Theta_i} \sum_{rs, ab \in \Theta_j}  \E( x_{rs} x_{tu} x_{ab}^2 ) + O(n^{-1/2}), \nonumber \\
&= n^{-4} \sum_{rs, tu \in \Theta_i} \sum_{rs, ab \in \Theta_j}   \E( x_{rs} x_{tu} ) + 2\E( x_{ab} x_{rs}) \E( x_{ab} x_{tu} ) + O(n^{-1/2}) \nonumber, \\
&= n^{-3}|\Theta_j| \psi_i + 2n^{-4}  \sum_{rs, tu \in \Theta_i}  \sum_{rs, ab \in \Theta_j} \E( x_{ab} x_{rs} ) \E( x_{ab} x_{tu} ) + O(n^{-1/2}) \nonumber, \\
&= n^{-3}|\Theta_j| \psi_i + 2n^{-4}  \sum_{rs, tu \in \Theta_i}  \sum_{rs, ab \in \Theta_j} \psi_j \E( x_{ab} x_{tu} ) + O(n^{-1/2}) \nonumber, \\
&= n^{-3}|\Theta_j| \psi_i + 2n^{-4}  \psi_j x^T \s{S}_i \s{S}_j x + O(n^{-1/2}). \nonumber 
\end{align}
Substituting into \eqref{eq:dc_term33}, the second term in the bias of the dyadic clustering estimator is 
\begin{align}
   - 2 n^{-4}  x^T  P_{\s{D}} \left(  xx^T \right) \Omega  &= -8 \left(  \sum_{i=1}^5 \frac{| \Theta_i |}{n^3} \psi_i \phi_i  \right) 
   - 4 n^{-4} \sum_{i = 1}^5 \sum_{j=1}^5 \phi_i \psi_j x^T \s{S}_i \s{S}_j x
   + O(n^{-1/2}). \label{eq:dc_term3_final}
\end{align}
Then, combining \eqref{eq:dc_sum} and \eqref{eq:dc_term3_final}, we obtain the following expression for the bias of the dyadic clustering estimator,
{ \small
\begin{align}
n^2 \text{Bias}\left\{ (\hat{V}_{DC} )_{22} \right\} 
&=  \left(  \sum_{i=1}^5 \frac{| \Theta_i |}{n^3} \psi_i \phi_i  \right) \left( -4 + 2\sum_{j=1}^5 \frac{| \Theta_j |}{n^3} \psi_j^2\right) 
- 4 n^{-4} \sum_{i = 1}^5 \sum_{j=1}^5 \phi_i \psi_j x^T \s{S}_i \s{S}_j x
+ O(n^{-1/2}), \nonumber
\end{align}
}
wherein we recognize the bias of the exchangeable estimator, and
\begin{align}
n^2 \text{Bias}\left\{ (\hat{V}_{DC} )_{22} \right\} 
&=
2n^2 \text{Bias}\left\{ (\hat{V}_{E} )_{22} \right\} 
- 4 \left(  \sum_{i=1}^5 \frac{| \Theta_i |}{n^3} \psi_i \phi_i  \right) 
 + O(n^{-1/2}). \label{eq:bias_equality}
\end{align}
Since we have $\Omega$ positive definite and 
\begin{align}
    0 \ge -4 n^{-3} x^T \Omega x &= - 4 \left(  \sum_{i=1}^5 \frac{| \Theta_i |}{n^3} \psi_i \phi_i  \right) + O(n^{-1/2}). \nonumber
\end{align}
Thus, we have shown that the bias of the dyadic clustering estimator is less than twice the bias of the exchangeable estimator. To establish that this same property holds for the absolute values of the biases, it remains to show that $ n^2 \text{Bias}\{ (\hat{V}_{E} )_{22} \} \le x^T \Omega x + O(n^{-1/2})$, that is, that the right hand side of \eqref{eq:bias_equality} is negative and bounded away from zero. 
Combining terms in \eqref{eq:bias_e_asymptotic}, the bias of the exchangeable estimator is 
\begin{align}
    n^2 \text{Bias} \left\{ (\hat{V}_{E} )_{22} \right\} 
&= \phi_3 \psi_3 \{ \psi_3^2 + \psi_4^2 + 2\psi_5^2 - 2 \psi_3(1 + \psi_5^2 / \psi_3^2) \} \ldots \nonumber \\
& \hspace{.5in}\phi_4 \psi_4 \{ \psi_3^2 + \psi_4^2 + 2\psi_5^2 - 2 \psi_4(1 + \psi_5^2 / \psi_4^2) \} \ldots \nonumber \\
 &\hspace{.5in} 2 \phi_5 \psi_5 \{ \psi_3^2 + \psi_4^2 + 2\psi_5^2 - 2(\psi_3 + \psi_4) \}  + O(n^{-1/2}). \nonumber
\end{align}
Now, using Proposition~\ref{prop:eval_relations}, we have that 
\begin{align}
    \psi_3^2 + \psi_4^2 + 2\psi_5^2 - 2 \psi_3(1 + \psi_5^2 / \psi_3^2) \le \psi_3^2 + \psi_4^2 + 2\psi_5^2 \le 1, \label{eq:bias_e_explicit} \\
    \psi_3^2 + \psi_4^2 + 2\psi_5^2 - 2 \psi_4(1 + \psi_5^2 / \psi_4^2) \le \psi_3^2 + \psi_4^2 + 2\psi_5^2 \le 1, \nonumber \\
    \psi_3^2 + \psi_4^2 + 2\psi_5^2 - 2(\psi_3 + \psi_4) \le \psi_3^2 + \psi_4^2 + 2\psi_5^2 \le  1, \nonumber 
\end{align}
and thus, 
\begin{align}
   n^2 \text{Bias}\left\{ (\hat{V}_{E} )_{22} \right\} 
&\le  (\phi_3 \psi_3  + \phi_4 \psi_4 + 2 \phi_5 \psi_5 ) + O(n^{-1/2}), \nonumber \\
&\le  n^{-3} x^T \Omega x + O(n^{-1/2}). \label{eq:bias_e_bound}
\end{align}
To establish a lower bound on the bias of the exchangeable estimator, using \eqref{eq:phi5_bound} and \eqref{eq:last_bound} in Proposition~\ref{prop:eval_relations},
\begin{align}
    \psi_3\psi_4 &\ge \psi_5^2,\nonumber \\
    \psi_3 + \psi_4 &\ge\psi_3 + \psi_5^2/\psi_3,\nonumber \\
    -(\psi_3 + \psi_4) &\le -(\psi_3 + \psi_5^2/\psi_3), \label{eq:lb_bias1}
\end{align}
and a similar argument holds for $ - \psi_4(1 + \psi_5^2 / \psi_4^2)$.
Then, recalling from Proposition~\ref{prop:eval_relations} that 
\[2 \psi_5^2 \ge \frac{1}{2}\psi_3^2 + \frac{1}{2}\psi_4^2 + \psi_3 \psi_4,\]
and combining with \eqref{eq:lb_bias1}, we obtain the following bound on the key expressions in \eqref{eq:bias_e_explicit},
\begin{align}
     \psi_3^2 + \psi_4^2 + 2\psi_5^2 - 2 \psi_3(1 + \psi_5^2 / \psi_3^2) &\ge \psi_3^2 + \psi_4^2 + 2\psi_5^2 - 2 (\psi_3 + \psi_4), \label{eq:lb_bias_boounds} \\
     &\ge (\psi_3 + \psi_4)^2 - 2 (\psi_3 + \psi_4) + \frac{1}{2}\psi_3^2 + \frac{1}{2}\psi_4^2 - \psi_3 \psi_4, \nonumber \\
     & \ge -1. \nonumber
\end{align}
To obtain the final inequality in \eqref{eq:lb_bias_boounds}, we use the fact that $0 \le \phi_3 + \phi_4 \le 1$ to find that 
\[ (\psi_3 + \psi_4)^2 - 2 (\psi_3 + \psi_4) \ge -1, \]
and that
\[ \frac{1}{2}\psi_3^2 + \frac{1}{2}\psi_4^2 - \psi_3 \psi_4 \ge 0. \]
A similar argument to that in \eqref{eq:lb_bias_boounds} can be obtained starting with $\psi_3^2 + \psi_4^2 + 2\psi_5^2 - 2 \psi_4(1 + \psi_5^2 / \psi_4^2)$.
Then, combining \eqref{eq:lb_bias_boounds} with the complete expression for the bias in \eqref{eq:bias_e_explicit} and the upper bound in \eqref{eq:bias_e_bound},
\begin{align}
    n^2 \Big| \text{Bias} \left\{ (\hat{V}_{E} )_{22}  \right\} \Big|
&\le n^{-3} x^T \Omega x + O(n^{-1/2}). \label{eq:abs_bias_e_bound}
\end{align}

Substituting the bound on the absolute value of the exchangeable bias \eqref{eq:abs_bias_e_bound} into  \eqref{eq:bias_equality}, we have that 
\begin{align}
n^2 \text{Bias}\left\{ (\hat{V}_{DC} )_{22} \right\} 
&\le - 2 n^{-3} x^T \Omega x + O(n^{-1/2}) \le - 2n^2 \Big| \text{Bias} \left\{ (\hat{V}_{E} )_{22} \right\} \Big| + O(n^{-1/2}),
\end{align}
which establishes the result. To complete the proof, we note that by the explicit expression for $n^2 \text{Bias}\{ (\hat{V}_{E} )_{22} \} $ in \eqref{eq:bias_e_asymptotic}, the exchangeable bias scaled by $n^2$ is $O(1)$. Similarly,  $n^2 \text{Bias}\{ (\hat{V}_{DC} )_{22} \} $ is $O(1)$ by \eqref{eq:bias_equality}.
\qed

\subsection{A corollary to Theorem~\ref{thm:bias}}

In the following corollary we expand the conditions of the bias theorem slightly, such that the error array is not strictly in the exchangeable, similar to Corollary~\ref{cor:46}.
\begin{corollary}
\label{cor:extend_bias}
Theorem~\ref{thm:bias} holds under the relaxation exchangeability of $\xi$, where the first two entries in the covariance matrix of $\xi$ are allowed to be heterogenous, that is,  $\E( \xi_{jk}^2 ) = \sigma^2_{jk}$ and  $\E( \xi_{jk} \xi_{kj} ) = \theta_{jk}$, where $|\theta_{jk}| <\sigma^2_{jk}   < C < \infty$. Similarly, the conditions hold when $\E( x_{jk}x_{jk}^T) = A_{jk}$ and $\E( x_{jk}x_{kj}^T) = D_{jk}$, where the entries in matrices $A_{jk}$ and $D_{jk}$ are bounded.  
\end{corollary}
The result follows directly from noting that the proof of Theorem~\ref{thm:bias} depends only on  the values of $\E( x_{jk} x_{lm} )$ and $\E( \xi_{jk} \xi_{lm} )$ for $(jk, lm) \in \Theta_i$ for $i \in \{3,4,5 \}$, at least up to a vanishing constant. Thus, so long as the expectations  $\E( x_{jk} x_{lm} )$ and $\E( \xi_{jk} \xi_{lm} )$ for $(jk, lm) \in \Theta_i$ for $i \in \{1,2 \}$ are bounded, Theorem~\ref{thm:bias} holds.
Corollary~\ref{cor:extend_bias} states that even under heteroskedasticity in the covariance matrix of the error vector $\xi$, the exchangeable estimator still outperforms the dyadic clustering estimator in bias. Only when the entries in $\Omega$ corresponding to configurations (b), (c), and (d) are heterogeneous does the theory fail to hold.

\section{Theoretical properties of the exchangeable estimator}
\label{app:theory}
In this section, we provide additional theoretical results for the proposed exchangeable estimator. Namely, we show that the proposed exchangeable estimator is consistent and provides an improvement in mean-square error over the dyadic clustering estimator, when the assumed model is correct. First, we provide conditions under which the theory holds, and prove asymptotic normality of ordinary least squares (which is necessary for the results that follow). For brevity, we focus on the ordinary least squares coefficient estimator and the case of a singe time period or array observation, where $R=1$.

\subsection{Conditions for theory}
We define the conditions under which we establish the theory, starting with a formal definition of the class of exchangeable covariance matrices.
\begin{definition}
\label{def:exchmat}
An {exchangeable covariance matrix} is defined as $\Omega_{E} = \E[\xi \xi^T]$ arising from mean-zero random vector $\xi = {\rm vec}(\Xi)$, where $\Xi$ is a jointly exchangeable random matrix with entries $\xi_{ij}$ independent $\xi_{kl}$ whenever $\{i,j\}\cap\{k,l\} \not= \varnothing$. 
\end{definition}

For the theoretical assessment of the exchangeable estimator, we take $X$ random, but still evaluate $\hat{\beta}$, $\hat{V}_{DC}$, and $\hat{V}_E$ conditional on $X$. We assume the {rows} of the matrix $X$ are jointly exchangeable, meaning that a reordering of the rows $\{\x_{jk} \}_{j,k=1}^n$ to $\{\x_{\pi(j) \pi(k)} \}_{j,k=1}^n$ leaves the distribution of matrix $X$ invariant
for any permutation $\pi(.)$. As with the dependence in the errors, we assume that two {rows} of $X$ that correspond to relations which do not share an actor are independent, that is row $\x_{ij}^T$ is independent row $\x_{kl}^T$ whenever $\{i,j\}\cap\{k,l\} = \varnothing$. This dependence in the rows of matrix $X$ (along with some assumptions on the finiteness of its moments) implies the following:
\begin{align}
\E ( x_{jk} x_{lm}^T ) = M_i, \ \  \ (jk, lm) \in \Theta_i, \ \ \ (i  = 0,a,b,c,d ), \label{eq_xexch_limit}
\end{align}
where $ \Theta_i$ is the set of  {pairs} of relations $(jk,lm)$ that share a member in the $i$th manner and `$0$' refers to self-relation (i.e. variance). 

The theoretical setting is as follows:

\begin{condition}
\label{cond:1}
Define the following data generating process:
\begin{enumerate}[label=(A\arabic*)]
\item
The true data generating model is ${y} = X \beta + {\xi}$, where the errors ${\xi}$ are mean-zero with exchangeable covariance matrix as  in Definition~\ref{def:exchmat}.
\item
At least one of $\{\phi_b, \phi_c, \phi_d\}$ is nonzero. 
\end{enumerate}
In addition, consider the following regularity conditions:
\begin{enumerate}[label=(B\arabic*)]
\item The covariate matrix $X$ has rows that are jointly exchangeable with at least one of $\{M_i \}_{i \in \{b,c,d \}}$ in \eqref{eq_xexch_limit} nonzero, and where row $\x_{ij}^T$ is independent row $\x_{kl}^T$ whenever $\{i,j\}\cap\{k,l\} \not= \varnothing$.
\item The fourth moments of the errors and the eight moments of the covariates are bounded: $\E( | \xi_{jk}| ^{4} ) < C < \infty$ and $\max_{l \in \{1,2,\ldots,p \}}\E ( | x_{jk}^{(l)}|^8 ) < C' < \infty$ where $\x_{jk} = [x_{jk}^{(1)}, \ldots, x_{jk}^{(p)}]^T$.
\item 
The errors ${\Xi}$ and covariates $X$ are independent.
\item $X$ is full rank. 
\end{enumerate}
\end{condition}

\subsection{Additional theoretical results}

The following theorem establishes asymptotic normality of $\hat{\beta}$, and supports the normal approximation to $\hat{\beta}$ to produce confidence intervals.
\begin{theorem}
\label{thm:AN}
Under Conditions~\ref{cond:1},
\begin{gather*}
n^{1/2}(\hat{\beta} - \beta) \rightarrow_d {\rm N}(0, V_0\big), \\
V_0 = E(x_{jk} x_{jk}^T)^{-1} \big\{ \phi_b E(x_{jk} x_{jl}^T) + \phi_c E(x_{jk} x_{lk}^T) + 2 \phi_d E(x_{jk} x_{lk}^T)  \big\}
    E(x_{jk} x_{jk}^T)^{-1},  
\hspace{.1in} (j \neq k \neq l),
\end{gather*}
where `$\rightarrow_d$' denotes convergence in distribution.
\end{theorem}
Theorem~\ref{thm:AN} is similar to Proposition 3.2 in \citet{tabord2017inference}, although with unbounded errors.  It may be seen as an extension of central limit theorems for sums over infinitely exchangeable arrays \citep{fortini2012central} to regression scenarios, and belongs to a broader literature on distributions of  random processes that satisfy a symmetry property \citep{kallenberg2006probabilistic, austern2018limit}.

We establish consistency of the proposed estimator, and its improvement in performance over the dyadic clustering estimator, in the following three theorems. 
\begin{theorem}
\label{thrm:consistency} 
Under Conditions~\ref{cond:1}, 
the exchangeable covariance estimator is consistent in the sense that 
\begin{align}
n\hat{V}_E - n \var( \hat{\beta} ) \rightarrow_p 0, \hspace{.1in}  n \rightarrow \infty, \nonumber
\end{align}
where `$\rightarrow_p$' denotes convergence in probability.
\end{theorem}

\begin{theorem}
\label{thm_mse}
Under Conditions~\ref{cond:1}, the mean-square error of the exchangeable estimator is less than that of the dyadic clustering estimator with probability approaching 1, that is
\[
\text{pr} \left[  E \left\{  
\left( \hat{V}_{E} - \var(\hat{\beta} \mid X) \right)^2 \mid X \right\} \le E \left\{ \left( \hat{V}_{DC} - \var(\hat{\beta} \mid X) \right)^2 \mid X \right\} \right] \rightarrow 1.
\]
\end{theorem}

\subsection{Corollary to theory}

The proofs of Theorem~\ref{thm:AN} through Theorem~\ref{thm_mse} 
depend only on the uniformity of the covariance entries in $\Omega$ corresponding to $\phi_b$, $\phi_c$, and $\phi_d$. Then, we can relax the assumption of exchangeability of the error vector $\xi$ to allow some heterogeneity, namely heteroskedasticity, under which the theory still holds.

\begin{corollary}
\label{cor:46}
Theorem~\ref{thm:AN} through Theorem~\ref{thm_mse} hold under the relaxation of (A2) and (B1), where the first two entries in the covariance matrix of $\xi$ are allowed to be heterogenous, that is,  $\E( \xi_{jk}^2 ) = \sigma^2_{jk}$ and  $\E( \xi_{jk} \xi_{kj} ) = \theta_{jk}$, where $|\theta_{jk}| <\sigma^2_{jk}   < C < \infty$. Similarly, the conditions hold when $\E( x_{jk}x_{jk}^T) = A_{jk}$ and $\E( x_{jk}x_{kj}^T) = D_{jk}$, where the entries in matrices $A_{jk}$ and $D_{jk}$ are bounded.  
\end{corollary}

\section{Proof of asymptotic normality of ordinary least squares}
\label{app:AN}
\subsection{Notation and outline of proof}
For this proof, and throughout the remainder of the  supplementary material, we adopt slightly different notation to simplify the representation of the exchangeable covariance estimator. Recall that the exchangeable covariance estimator for the OLS estimating equations is 
\[\hat{V}_E = (X^T X)^{-1}X^T \hat{\Omega}_E X (X^T X)^{-1},\]
where $\hat{\Omega}_E$ is the exchangeable estimate of the error covariance matrix, consisting of five averages of residual products. Here we express $\hat{\Omega}_E$ as
\begin{align}
\hat{\Omega}_E &= \sum_{i = 1}^5 \hat{\phi}_i \mathcal{S}_i, \hspace{.15in}
\hat{\phi}_i = \frac{ \sum_{(jk,\ell m) \in \Theta_i} e_{jk} e_{\ell m}}{|\Theta_i|},  \hspace{.15in} (i  = 1,\ldots,5). \label{eq:notation_chge}
\end{align}
This amounts to mapping $\sigma^2 \mapsto \phi_1$, $\phi_a \mapsto \phi_2$, ..., $\phi_d \mapsto \phi_5$, and re-indexing the $\mathcal{S}$ and $M$ matrices accordingly. Additionally, when we consider sequences of jointly exchangeable random variables $\{ W_{ij} \}_{i,j=1}^n$, it is understood that the sequence arises from a relational array such that entries with $i=j$ are undefined. Thus, sums over the sequence are of $n(n-1)$ terms and we define $\sum_{ij} W_{ij}  = \sum_{i \neq j} W_{ij}$.

We work in the asymptotic regime where actors are added incrementally to the relational data set, i.e. $n$ is continually increasing. To establish asymptotic normality of $\hat{\bbeta}$, we wish to show
\begin{align}
&n^{1/2}(\hat{\bbeta} - \bbeta) \rightarrow_d {\rm N}(0, M_1^{-1} \big[ \phi_3 M_3 + \phi_4 M_4 + 2\phi_5 M_5 \big]
   M_1^{-1}  \big),   \nonumber 
\end{align}
where $\{M_i \}_{i \in \{1,3,4,5\}}$ are as in \eqref{eq_xexch_limit} and `$\rightarrow_d$' denotes element-wise convergence in distribution.

The motivation for the proof argument follows from the expression
\begin{align}
n^{1/2}(\hat{\bbeta} - \bbeta) = \left( \frac{\sum_{jk} \x_{jk} \x_{jk}^T}{n(n-1)} \right)^{-1}   \frac{n^{1/2} \sum_{jk} \x_{jk} \xi_{jk}}{n(n-1)}. \label{eq:beta_an1}
\end{align}
We state that $\left( \frac{\sum_{jk} \x_{jk} \x_{jk}}{n(n-1)} \right)^{-1}$ converges in probability to $M_1^{-1}$,  and then show asymptotic normality of the second multiplicative term in \eqref{eq:beta_an1}. 

In analyzing $\{\x_{ij} \xi_{ij} \}_{i,j=1}^n$, 
by condition (B1), the joint exchangeability and  independence of non-overlapping pairs of the sequence $\{\xi_{ij} \}_{i,j=1}^n$ extends to the component sequences in the vectors $\{\x_{ij} \xi_{ij} \}_{i,j=1}^n$.
Thus, to prove asymptotic normality of $\hat{\bbeta}$, we first prove a theorem stating that the average of a mean-zero sequence of jointly exchangeable random variables is asymptotically normal. Specifically, for $\{ W_{ij} \}_{i,j=1}^n$ mean zero and jointly exchangeable, we show
\begin{align}
k_n \frac{\sum_{ij} W_{ij}}{\sigma} \rightarrow_d  {\rm N}(0,1) \label{eq:exch_asymp}
\end{align}
for some normalizing constant $\sigma$ and fixed sequence $k_n \rightarrow 0$ as $n\rightarrow \infty$.

To prove \eqref{eq:exch_asymp}, we rely on a result from \citet{bolthausen1982central}, as well as a supporting lemma which we present here.  Below we outline the significance of these results in the proof. 
\begin{itemize}
\item Lemma~\ref{lem:clt_characteristic} (\cite{bolthausen1982central}):  Provides a sufficient condition for asymptotic normality of a sequence of measures based on the standard normal characteristic function. 
\item {Lemma~\ref{lem:an_var_bound}: } Provides a bound for a variance that surfaces in the proof of asymptotic normality in \eqref{eq:exch_asymp}.
\end{itemize}
 From \eqref{eq:exch_asymp}, we immediately have the marginal asymptotic normality of the sample mean of each of the vector components in the sequence $\{\x_{ij} \xi_{ij} \}_{i,j=1}^n$. To establish joint asymptotic normality, we employ the Cram{\'e}r-Wold device \citep{cramer1936some}, where asymptotic normality of $\{ \v^T \x_{ij} \xi_{ij} \}_{i,j=1}^n$, for all $\v \in \R^p$ with $||\v|| = 1$, establishes joint normality. To achieve the asymptotic normality of this inner product, we simply recognize that this inner product is itself the mean of an exchangeable sequence of random variables. Joint asymptotic normality of the mean of the sequence of vectors $\{\x_{ij} \xi_{ij} \}_{i,j=1}^n$ establishes joint asymptotic normality of $\hat{\bbeta}$
 via \eqref{eq:beta_an1}.

\subsection{Lemmas and theorem in support of Theorem~\ref{thm:AN}}
The following is Lemma 2 in~\cite{bolthausen1982central} and provides a sufficient condition for asymptotic normality. We abuse notation slightly, letting $i$ be the imaginary unit where appropriate. 
\begin{lemma}[\citet{bolthausen1982central}] 
\label{lem:clt_characteristic}
Let $\nu_n$ be a sequence of probability distributions over $\R$ which satisfies
\begin{enumerate}
\item $\ \sup_n \int x^2 d\nu_n(x) < \infty $, and
\item for all $\lambda \in \R$,   $\ \lim_n \int (i \lambda - x){\rm e}^{i \lambda x}d\nu_n(x) = 0.$
\end{enumerate}
Then, $\nu_n \rightarrow_d {\rm N}(0,1).$
\end{lemma} \noindent To provide intuition for Lemma~\ref{lem:clt_characteristic}, the integral in condition (2) is identically zero when $\nu_n$ is the standard normal distribution.

The next lemma provides a sufficient condition on the dependence structure in $\{ W_{ij} \}_{i,j=1}^n$ necessary for the proof of asymptotic normality in \eqref{eq:exch_asymp}. Reecall that terms in $\{ W_{ij} \}_{i,j=1}^n$ with $i=j$ are undefined. 
\begin{lemma} 
\label{lem:an_var_bound}
Let $\{ W_{ij} \}_{i,j=1}^n$ be a sequence of jointly exchangeable random variables as in Definition~\ref{def:exchmat} with $||W_{ij}||_4 < L < \infty$, where $||W_{ij}||_p := \E\left( |W_{ij}|^p \right)^{1/p}$ for $p>0$.  
Then,
\begin{align} 
\frac{1}{n^6} \var \left( \sum_{ij} \sum_{kl \in \Theta_{ij}} W_{ij} W_{kl} \right) < \frac{C L^4 }{n} \rightarrow 0 \text{ as  } n\rightarrow \infty,  \nonumber
\end{align}
for some $C<\infty$, where $\Theta_{ij}$ is the set of ordered pairs $(k,l)$ that share an index with $(i,j)$.
\label{lem:var_bound}
\end{lemma}

\begin{proof}
By definition we write
{
\begin{align} 
\frac{1}{n^6} \var \left(  \sum_{ij} \sum_{kl \in \Theta_{ij}} W_{ij} W_{kl} \right) &= \frac{1}{n^6} \sum_{ij} \sum_{kl \in \Theta_{ij}} \sum_{rs} \sum_{tu \in \Theta_{rs}}{\rm cov}(W_{ij} W_{kl}, W_{rs} W_{tu}).  \label{eq:quad_sum}
\end{align}
}
Each covariance of \eqref{eq:quad_sum} is bounded by $L^4$. To bound the variance, we will show the number of nonzero entries in the sum is ${O}(n^5)$. For ${\rm cov}(W_{ij}W_{kl}, W_{rs}W_{tu}) \neq 0$, there must be overlap between the index sets $\{ i,j,k,l \}$ and $\{ r,s,t,u \}$. Further, the sum in \eqref{eq:quad_sum} is taken over index sets that themselves contain overlap, i.e. $ \{ i, j \} \cap \{ k, l \} \neq \varnothing$ and $ \{ r, s \} \cap \{ t, u \} \neq \varnothing$. For example, the index sets $\{ i,j,i,l \}$ and $\{ i,s,i,u \}$ have nonzero covariance in \eqref{eq:quad_sum}. Since there are 5 unique indices in the union of the sets $\{ i,j,i,l \}$ and  $\{ i,s,i,u \}$, there are ${O}(n^5)$ such index set pairs of this form in total. There are 96 pairs of index sets that result in nonzero covariance ${\rm cov}(W_{ij}W_{kl}, W_{rs}W_{tu})$. For example, another such pair of index sets is $\{i,j,i,l \}$ and $\{i,j,i,j \}$. Each of these 96  pairs of index sets is ${O}(n^5)$. Thus, the sum of covariances in \eqref{eq:quad_sum} is over ${O}(n^5)$ bounded elements.
\end{proof}

It is worth noting that we repeat the counting argument in the proof of Lemma~\ref{lem:an_var_bound} in many of the following proofs, including those in later sections. Now that we have Lemma~\ref{lem:clt_characteristic}~and~\ref{lem:an_var_bound}, we prove that a general sequence of mean-zero exchangeable random variables is asymptotically normal.

\begin{theorem} 
\label{thm:an_exchangeable}
Let $\{ W_{ij} \}_{i,j=1}^n$ be a mean-zero sequence of jointly exchangeable random variables 
with at least one of $\{\phi_3, \phi_4, \phi_5 \}$ nonzero.
If $|| W_{ij} ||_{4} < L < \infty$, then
\begin{align} 
\frac{n^{1/2} \sum_{ij} W_{ij}}{n(n-1)} \rightarrow_d {\rm N}(0, \phi_3 + \phi_4 + 2\phi_5) \text{ as  } n \rightarrow \infty.\label{eq:clt_condition}
\end{align}  

\end{theorem}

\begin{proof}
We first show that $\phi_3 + \phi_4 + 2 \phi_5$ is the correct limiting variance. Writing the variance of the expression on the left hand side of \eqref{eq:clt_condition} explicitly and recalling that entries such that $i=j$ are undefined, we see 
{
\begin{align}
\var \left(\frac{n^{1/2}}{n(n-1)} \sum_{ij} W_{ij} \right) &= \frac{n}{n^2(n-1)^2} \sum_{ij} \sum_{kl \in \Theta_{ij}} {\rm cov}\left( W_{ij}, W_{kl} \right) \nonumber \\ 
  &= \frac{ n^2(n-1) \left( \phi_1 + \phi_2 \right) + n^2(n-1)(n-2) \left( \phi_3 + \phi_4 + 2 \phi_5 \right) }{n^2(n-1)^2} \nonumber \\
  & \rightarrow \phi_3 + \phi_4 + 2 \phi_5 \text{ as  } n\rightarrow \infty, \nonumber
\end{align}
}
by the properties of joint exchangeability of $\{ W_{ij} \}_{i,j=1}^n$ as described in Section~\ref{sec:exch_cov_mat}. This variance is finite and nonzero by assumption.  To prove \eqref{eq:clt_condition}, it is sufficient to show
\begin{align} 
\bar{S}_n := \frac{ \sum_{ij} W_{ij}}{n^{3/2} ( \phi_3 + \phi_4 + 2\phi_5 )^{1/2}} \rightarrow_d N(0, 1). \label{eq:barSn}
\end{align}
Define the limiting variance as $\sigma_n^2 = n^{3} (\phi_3 + \phi_4 + 2\phi_5)$  and the sum $S_n = \sum_{ij} W_{ij}$.

To establish \eqref{eq:barSn}, we employ Lemma~\ref{lem:clt_characteristic}, where $\nu_n$ is the probability measure corresponding to $\bar{S}_n$ for all $n$. The first condition of Lemma~\ref{lem:clt_characteristic} is satisfied since 
\begin{align}
\E ( (\bar{S}_n)^2 ) = \frac{\var \left( \sum_{ij} W_{ij} \right) }{ n^3 (\phi_3 + \phi_4 + 2\phi_5)} < CL^2 \nonumber
\end{align}
for $C<\infty$ and all $n$. Thus, to prove \eqref{eq:barSn}, it is sufficient to show the second condition of Lemma~\ref{lem:clt_characteristic}: for all $\lambda \in \R$, as $n \rightarrow \infty$,
\begin{align}
\E \left( (i \lambda - \bar{S}_n){\rm e}^{i \lambda \bar{S}_n}\right) \rightarrow 0.
\label{eq:lemma2_rewrite}
\end{align}
We decompose the term in the expectation as in~\citet{guyon1995random} and~\citet{lumley2003asymptotics}:
\begin{align}
(i \lambda - \bar{S}_n){\rm e}^{i \lambda \bar{S}_n} = A_1 - A_2 - A_3,  \nonumber 
\end{align}
\vspace{-.3in}
\begin{align}
& A_1 = i \lambda{\rm e}^{i \lambda \bar{S}_n}\left( 1 - \frac{1}{\sigma_n^2}\sum_{i j}  W_{ij} S_{ij,n}\right), \nonumber\\
& A_2 = \frac{{\rm e}^{i \lambda \bar{S}_n}}{\sigma_n}\sum_{i j}  W_{ij}\left(1 - i \lambda \bar{S}_{ij,n} - {\rm e}^{-i \lambda \bar{S}_{ij,n}} \right), \nonumber \\
&A_3= \frac{1}{\sigma_n}\sum_{i  j}  W_{ij} {\rm e}^{i \lambda (\bar{S}_n - \bar{S}_{ij, n})}, \nonumber   \\
&S_{ij,n} = \sum_{kl \in \Theta_{ij}} W_{kl}, \nonumber
\end{align}
\noindent where $\bar{S}_{ij,n} = S_{ij,n}/\sigma_n$.
To satisfy \eqref{eq:lemma2_rewrite} it remains to be shown that $\lim_{n\rightarrow \infty} \E(A_m) = 0$ for each $m \in \{1,2,3\}$.

Beginning with ${A_1}$, $|{\rm e}^{i \lambda \bar{S}_n}| \le 1$. Using this fact and Lemma~\ref{lem:var_bound},
{\footnotesize
\begin{align}
0 \le \E( |A_1| )^2 \le \E( |A_1^2 | )  &\le \lambda^2 \E\left( \left| 1 -  \frac{1}{\sigma_n^2}\sum_{i j} W_{ij}  S_{ij,n} \right|^2 \right) \nonumber \\
  &= \frac{\lambda^2}{\sigma_n^4} \var \left( \sum_{ij} \sum_{kl \in \Theta_{ij}} W_{ij} W_{kl} \right) + \lambda^2 \left\{ 1 - \frac{\var \left( \sum_{i j} W_{ij} \right)}{\sigma_n^2} \right\}^2 \nonumber  \\
  &\le \lambda^2 \frac{CL^4}{n} + \lambda^2  \left(1 - \frac{\sigma^2_n + {O}(n^{-1})}{\sigma^2_n} \right)^2 \nonumber\\
  & = \lambda^2 \left( \frac{CL^4}{n} + \frac{ {O} \left( n^{-2} \right)}{\sigma^2_n} \right)\rightarrow 0 \nonumber
\end{align}
}
for all real-valued $\lambda$. $\E( |A_1| )^2$ limiting to zero implies $\E( |A_1| )$ limits to zero, and hence $\E( A_1 )$ limits to zero.

Now, for ${A_2}$,   
by Taylor expansion of ${\rm e}^{-i \lambda \bar{S}_{ij,n}}$, we can write
\begin{align}
\left|1 - i \lambda \bar{S}_{ij,n} - {\rm e}^{-i \lambda \bar{S}_{ij,n}} \right| \le c \lambda^2 \left( \bar{S}_{ij,n} \right)^2, \nonumber 
\end{align}
for some $0 < c < \infty$ and all $n, \lambda$. Using this bound and the fact that $|\Theta_{ij}| = 4n - 6$, we evaluate $\E( |A_2| )$ directly below: 
\begin{align}
\E( |A_2| ) &\le \frac{1}{\sigma_n} \E \left( \sum_{i  j} \left| W_{ij}  \right|  \left| 1 - i \lambda \bar{S}_{ij,n} - {\rm e}^{-i \lambda \bar{S}_{ij,n}} \right| \right), \nonumber \\ 
&\le  \frac{c \lambda^2}{\sigma_n^3}\sum_{i j}  \E \left(  \left|  W_{ij}  \right|  S_{ij,n}^2 \right), \nonumber \\
&\le \frac{c \lambda^2}{\sigma_n^3}n(n-1)(4n-6)^2 L^3 \rightarrow 0, \nonumber 
\end{align}
for all real $\lambda$. As $\E( |A_2| )$ limits to zero, so does $\E( A_2 )$.

Finally, for ${A_3}$, the expression ${S}_{ij,n}$ sums all terms in the sequence $\{ W_{ij} \}_{i,j=1}^n$ that depend upon $W_{ij}$, including $W_{ij} $ itself. Thus, $W_{ij} $ and $\bar{S}_{n} - \bar{S}_{ij,n}$ are independent. It follows immediately that 
\begin{align}
 \E \left( \frac{1}{\sigma_n}\sum_{i  j} W_{ij}  {\rm e}^{i \lambda (\bar{S}_n - \bar{S}_{ij, n})} \right) =\frac{1}{\sigma_n}\sum_{i  j} \E\big( W_{ij}  \big) \E \big( {\rm e}^{i \lambda (\bar{S}_n - \bar{S}_{ij, n})} \big) = 0, \nonumber
\end{align}
since $\E \left( W_{ij} \right) = 0$ for all ordered pairs $(i,j)$. 

Hence, $\lim_{n\rightarrow \infty} \E( A_m ) = 0$ for each $m \in \{1,2,3\}$ and we have the convergence in~\eqref{eq:lemma2_rewrite}, implying  
$\bar{S}_n \rightarrow_d {\rm N}(0,1)$ by Lemma~\ref{lem:clt_characteristic}, which gives the desired result in~\eqref{eq:clt_condition}.
\end{proof}

\subsection{Proof of Theorem~\ref{thm:AN}}
We begin by writing 
\begin{align}
n^{1/2}(\hat{\bbeta} - \bbeta) = \left( \frac{\sum_{jk} \x_{jk} \x_{jk}^T}{n(n-1)} \right)^{-1}   \frac{n^{1/2} \sum_{jk} \x_{jk} \xi_{jk}}{n(n-1)}, \label{eq:beta_hat_AN}
\end{align}
again emphasizing that entries in the sum with $j=k$ are undefined and omitted. Addressing the first multiplicative term in \eqref{eq:beta_hat_AN}, we recall that the inverse map is continuous. Then, by \eqref{eq_xexch_limit} and the continuous mapping theorem, we have
\begin{align}
\left( \frac{\sum_{jk} \x_{jk} \x_{jk}^T}{n(n-1)} \right)^{-1} \rightarrow_p M_1^{-1}. \label{eq:conv_in_prob_xx}
\end{align}

We now analyze the second multiplicative term in \eqref{eq:beta_hat_AN}. Showing asymptotic normality of this term is sufficient to show asymptotic normality of the expression on the left hand side of  \eqref{eq:beta_hat_AN}. Recall $\x_{jk}^T = [x_{jk}^{(1)}, x_{jk}^{(2)}, \ldots, x_{jk}^{(p)}]$. We wish to show that the sum of vectors
\begin{align}
\U_n := \frac{n^{1/2} }{n(n-1)} \sum_{jk} \x_{jk} \xi_{jk} \rightarrow_d  {\rm N}(0, \Sigma), \label{eq:vector_an}
\end{align}
for some limiting variance $\Sigma$. 
By the Cram{\'e}r-Wold device \citep{cramer1936some}, $\U_n$ is asymptotically normal with asymptotic variance $\Sigma$ if and only if $ \v^T \U_n$ is asymptotically normal with asymptotic variance $\v^T \Sigma \v$ for every vector $\v \in \R^p$ such that $||\v|| = 1$. Clearly, 
\begin{align}
\v^T \U_n := \frac{n^{1/2} }{n(n-1)} \sum_{jk} \tilde{x}_{jk} \xi_{jk}, \nonumber
\end{align}
where we define $\tilde{x}_{jk} = \v^T \x_{jk}$. We wish to apply 
Theorem~\ref{thm:an_exchangeable} to the sequence $\{ \tilde{x}_{jk} \xi_{jk} \}_{j,k=1}^n$. First, the condition of finite moments in (B2) of Theorem~\ref{thm:AN} and $||\v|| = 1$ implies that $|| \tilde{x}_{jk} \xi_{jk} ||_{4} < L$ for some finite $L < \infty$. 
Secondly, by the independence of $X$ and $\Xi$ in (B3) of Theorem~\ref{thm:AN}, the sequence $\{ \tilde{x}_{jk} \xi_{jk} \}_{j,k=1}^n$ is a mean-zero exchangeable sequence of scalar random variables. Taking the variance directly,
\begin{align}
\var \left( \sum_{jk} \tilde{x}_{jk} \xi_{jk} \right) = n^3\v^T \{\phi_1 M_1 {O}(n^{-1}) + \phi_2 M_2 {O}(n^{-1}) + \phi_3 M_3  +  \phi_4 M_4 + 2 \phi_5 M_5\} \v. \label{eq_var_Un}
\end{align}
Then, we apply Theorem~\ref{thm:an_exchangeable} with $\sigma_n^2 = \var( \sum_{jk} \tilde{x}_{jk} \xi_{jk})$ in \eqref{eq_var_Un}, which gives that
\begin{align}
\v^T \U_n \rightarrow_d  {\rm N}(0,  \v^T [ \phi_3 M_3  +  \phi_4 M_4 + 2 \phi_5 M_5 ] \v  ). \nonumber
\end{align}
Thus, by the Cram{\'e}r-Wold device, we get the desired joint asymptotic normality
\begin{align}
 \frac{n^{1/2} \sum_{jk} \x_{jk} \xi_{jk} }{ n(n-1)} 
 &\rightarrow_d  {\rm N} \big(0,  \phi_3 M_3  +  \phi_4 M_4 + 2 \phi_5 M_5 \big). \label{eq:an_normality_vector_xxi}
\end{align}
Combining the convergence in probability in \eqref{eq:conv_in_prob_xx} and the asymptotic normality of \eqref{eq:an_normality_vector_xxi}, we obtain the desired result. 
\qed

\section{Proof of consistency of the exchangeable estimator}
\label{app:consistency}

\subsection{Notation and outline of proof}
For the proof of the consistency of the exchangeable estimator $\hat{V}_E$, we adopt the same change in notation as in Section \ref{app:AN}, defined in \eqref{eq:notation_chge}. We deviate slightly in that we denote $\Theta_{i}$ to denote dyadic pairs $(j,k)$ and $(l,m)$ that share a member in the $i$th manner. For example, for $i=3$ we must have $j=l$ and $m \neq k$. We use the same assumptions as in Theorem~\ref{thm:AN}. 

This proof is outlined as follows. We initially prove that the exchangeable estimator $\hat{V}_E$ is consistent if the exchangeable parameter estimates $\{\hat{\phi}_i: i=1,\ldots,5\}$ are consistent for the true parameters. We then prove consistency of $\{ \hat{\phi}_i \}$ in two steps: (a) we show parameter estimates $\{\widetilde{\phi}_i\}$ based on the unobserved true errors $\Xi$ are consistent and then (b) we show that the parameter estimates $\{\hat{\phi}_i\}$ are asymptotically equivalent to $\{\widetilde{\phi}_i\}$. We require the consistency of $\hat{\bbeta}$ result (implied by Theorem~\ref{thm:AN}) for this last step.

\subsection{Proof of Theorem~\ref{thrm:consistency}: consistency of the exchangeable estimator}
First, from Theorem~\ref{thm:AN}, the order of convergence of  $\hat{\bbeta}$ is ${n}^{1/2}$.  
Thus, we choose the rate $n$ as our asymptotic regime for consistency of $\hat{V}_E$. We wish to show that 
\begin{align}
n\hat{V}_E - n \var( \hat{\beta} ) \rightarrow_p 0. \label{eq:V_consistency}
\end{align}

First we show that to prove consistency of $\hat{V}_E$, it is sufficient to prove the consistency of the parameter estimates $\{\hat{\phi}_i\}$ for the true parameters. We begin by writing the difference of variances $n\hat{V}_E - n \var( \hat{\beta} )$ in \eqref{eq:V_consistency},  as 
{ 
\begin{align}
& n(X^TX)^{-1}X^T \big(\hat{\Omega}_E - \Omega_E \big)  X (X^T X)^{-1} 
\nonumber \\
  &=\frac{n}{n^2(n-1)^2}\left\{ \frac{X^TX}{n(n-1)} \right\}^{-1} \left\{ \frac{X^T \sum_{i = 1}^5|\Theta_i| \left(\hat{\phi}_i - \phi_i \right) \mathcal{S}_i X}{|\Theta_i|} \right\} \left\{ \frac{X^TX}{n(n-1)} \right\}^{-1} \nonumber \\
  &= \sum_{i = 1}^5 \frac{|\Theta_i|  }{n(n-1)^2}  
  \left( \hat{\phi}_i - \phi_i \right)
  \left\{ \frac{X^TX}{n(n-1)}
  \right\}^{-1} 
  \left( \frac{  \sum_{(jk,\ell m) \in \Theta_i} \x_{jk} \x_{\ell m}^T  }{|\Theta_i|} \right) 
  \left\{ \frac{X^TX}{n(n-1)}
  \right\}^{-1}  \nonumber \\
  &= \sum_{i = 1}^5 c_{i,n} \left( \hat{\phi}_i - \phi_i \right) h_{i,n}\left( X \right), \label{eq:c_phi_h}
\end{align}
}
where $c_{i,n} = |\Theta_i|/n(n-1)^2$ and $h_{i,n}(X)$ contains the remaining terms which are functions of $X$. By the counting argument used to show Lemma~\ref{lem:an_var_bound}, each $|\Theta_i|$ is at most ${O}(n^3)$, so each $c_{i,n} \rightarrow d_i$ for some finite constant $d_i$. Namely, $d_i = 0$ for $i\in \{1,2\}$, $d_i = 1$ when $i \in \{3,4\}$, and $d_i=2$ for $i=5$.
To obtain the result in \eqref{eq:V_consistency}, it is sufficient then to show $\hat{\phi}_i - \phi_i \rightarrow_p 0$ and $h_{i,n}(X)$ converges in probability to some constant for all $i$.  The latter comes easily, that is, by assumption and Slutsky's theorem,
\begin{align}
h_{i,n} (X) \rightarrow_p M_1^{-1} M_i M_1^{-1}, \ \ i \in \{1,\ldots, 5 \}. \nonumber
\end{align}
The continuous mapping theorem allows us to take the probability limit of ${X^TX} / ({n(n-1))}$ before inversion, as the inversion map is continuous.

We now consider consistency of the parameter estimates $\hat{\phi}_i$. First, define error averages $\{\widetilde{\phi}_i: i = 1,\ldots,5\}$ analogous to the parameter estimates, such that for each $i$,
\begin{align}
\widetilde{\phi}_i  &=  \frac{1 }{|\Theta_i|} \sum_{(jk,\ell m) \in \Theta_i} \xi_{j k} \xi_{\ell m}. \nonumber
\end{align}
We will show $\widetilde{\phi}_i -\phi_i$ converges in probability to zero, and then do the same for $\hat{\phi}_i - \widetilde{\phi}_i$. This is sufficient for showing $\hat{\phi}_i - {\phi}_i \rightarrow_p 0$ as $\hat{\phi}_i - {\phi}_i  = (\hat{\phi}_i - \widetilde{\phi}_i) + (\widetilde{\phi}_i - \phi_i)$. \\

To show convergence in probability of $\widetilde{\phi}_i -\phi_i$ to zero, we use the argument that the bias and variance both tend to zero. By assumption (A1), $\E[\xi_{jk} \xi_{\ell m}] = \phi_i$ for every relation pair $(jk, \ell m) \in \Theta_i$. Thus, $\E[\widetilde{\phi}_i -\phi_i] = 0$ for all $n$ and $i \in \{1,\ldots,5\}$. We now turn to the variance:
\begin{align}
\var \left(  \widetilde{\phi}_i  \right) &= \frac{1}{|\Theta_i|^2} \sum_{(jk,\ell m) \in \Theta_i} \sum_{(rs,tu) \in \Theta_i}  {\rm cov} \big( \xi_{j k} \xi_{\ell m}, \xi_{r s} \xi_{t u} \big).  \nonumber 
\end{align}
We again make a counting argument similar to that in Lemma \eqref{lem:an_var_bound}. By condition (B2), each of the $|\Theta_i|^2$ covariances in the sum above are bounded. The covariance between $\xi_{jk}\xi_{\ell m}$ and $\xi_{rs} \xi_{tu}$ is nonzero only if there is overlap between their two index sets. This reduces the number of nonzero covariances from the maximum possible $|\Theta_i|^2$ by a factor of at least $n$. Again, consider the case of $i=3$ where $|\Theta_3| = {O}(n^3)$. Each pair of relations in $\Theta_3$ must be of the form $(jk,jm)$, and thus the second set of indices must be of the form $(js,ju)$, for example, for the covariance to be nonzero. The set of indices $\{j,k,j,m,j,s,j,u \}$ is of order ${O}(n^5) = |\Theta_3|^2n^{-1}$. There are other forms of relation pairs in the second sum that give rise to nonzero covariance, such as $(ks,ku)$ and so on. However, there are nine such forms, each of which is  ${O}(n^5)$. Thus, the number of nonzero covariances is ${O}(n^5)$, and hence, we have
\begin{align}
\var \left(  \widetilde{\phi}_i  \right) &= \frac{|\Theta_i|^2 {O}(n^{-1})}{|\Theta_i|^2} \rightarrow 0. \label{eq:tilde_phi_n}
\end{align}
This same argument holds for all $i$, and thus, we have the desired consistency: $
\widetilde{\phi}_i - \phi_i  \rightarrow_p 0$ for $i=1,\ldots,5$. \\

We now show that $\hat{\phi}_i - \widetilde{\phi}_i$ converges in probability to zero. We first write the expression in terms of the estimated coefficients $\hat{\bbeta}$:
{
\begin{align}
\hat{\phi}_i \; - \; \widetilde{\phi}_i   &= \frac{\sum_{(jk,\ell m) \in \Theta_i}e_{jk} e_{\ell m} - \xi_{jk} \xi_{\ell m} }{|\Theta_i|} \notag \\
&= \frac{1}{|\Theta_i|} \sum_{(jk,\ell m) \in \Theta_i} \left \{ (\beta - \hat{\beta})^T(\x_{jk} \x_{\ell m}^T ) (\beta - \hat{\beta})-(\beta - \hat{\beta})^T(\xi_{jk} \x_{\ell m} + \xi_{\ell m} \x_{j k}) \right \}. 
\label{eq:ResidDiff}
\end{align}
}
By Theorem~\ref{thm:AN}, $\hat{\bbeta} - \bbeta$ converges to zero in probability. By Slutsky's theorem, if the terms in \eqref{eq:ResidDiff} involving elements of $X$ and $\xi$ converge in probability to any constant, then $\hat{\phi}_i - \widetilde{\phi}_i$ converges in probability to zero.  By (B1) and \eqref{eq_xexch_limit} we have the convergence in probability of the term
involving $\x_{jk} \x_{\ell m}^T$. Furthermore, by condition (B3), we have that $\E( \xi_{jk} \x_{\ell m} ) = \E( \xi_{\ell m} \x_{jk} ) = 0$. It remains to be shown that the variance of the error-covariate averages tend to zero. Consider the variance 
of the first error-covariate averages:  
{
\begin{align}
\var \left( \frac{1}{|\Theta_i|} \sum_{(jk,\ell m) \in \Theta_i} \xi_{jk} \x_{\ell m} \right) &= \frac{1}{|\Theta_i|^2} \sum_{(jk,\ell m) \in \Theta_i} \sum_{(rs,tu) \in \Theta_i} {\rm cov}\left( \xi_{jk} \x_{\ell m}, \xi_{rs} \x_{tu}\right), \nonumber \\
&= \frac{1}{|\Theta_i|^2} \sum_{(jk,\ell m) \in \Theta_i} \sum_{(rs,tu) \in \Theta_i} \E \left( \x_{\ell m} \x_{tu}^T \right) {\rm cov}\left( \xi_{jk}, \xi_{rs} \right).
\label{eq:v_xijk_xlm}
\end{align}
}
In writing \eqref{eq:v_xijk_xlm}, we use condition (B3) and simplify by conditioning on $X$ and using the law of total variance. By the same counting arguments used to establish  \eqref{eq:tilde_phi_n}, there are $|\Theta_i|^2 {O}(n^{-1})$ nonzero bounded covariances in \eqref{eq:v_xijk_xlm}. Thus, we have  
\begin{align}
\var \left( \frac{1}{|\Theta_i|} \sum_{(jk,\ell m) \in \Theta_i} \xi_{jk} \x_{\ell m} \right) &= \frac{|\Theta_i|^2 {O}(n^{-1})}{|\Theta_i|^2} \rightarrow 0.  \nonumber
\end{align}
Since the expectation and variance both tend to zero, we have 
\begin{align}
\frac{1}{|\Theta_i|} \sum_{(jk,\ell m) \in \Theta_i} \xi_{jk} \x_{\ell m} \rightarrow_p 0.  \nonumber
\end{align}
The same argument applies to the second error-covariate term in \eqref{eq:ResidDiff}. Thus, we have shown that consistency of $\hat{\bbeta}$ implies  
\begin{align}
\hat{\phi}_i - \widetilde{\phi}_i \rightarrow_p 0. \nonumber
\end{align}
\qed

\section{Proof of Theorem~\ref{thm_mse}: mean-square error of the exchangeable and dyadic clustering estimators}
\label{app:mse}

In this section, we prove that the mean-square error of the estimator of $\var ( \hat{\bbeta} )$, conditional on $X$, is lower when using the exchangeable estimator than that when using the dyadic clustering estimator with high probability in $X$, assuming that the error structure is exchangeable. Before proving the theorem, we provide a lemma that states
that the mean-square error of the each estimator is asymptotically equivalent to the mean-square error of each estimator based on the true errors, which vastly simplifies the proof of the mean-square error theorem. Even so, we must consider higher order moments of $\xi$ than the covariances ${\rm cov} \left(\xi_{jk}, \xi_{lm} \right)$. So, we also provide a lemma in which we define the covariance of any pair of product of error relations ${\rm cov} \left(\xi_{jk} \xi_{lm}, \xi_{rs} \xi_{tu} \right)$ and define the limiting values of the covariance of the error averages, $n{\rm cov} ( \tilde{\phi}_v, \tilde{\phi}_w )$, for every pair $( v,w ) \in \{1,\ldots,5 \} \times \{1,\ldots,5 \} $. 

In this Section, we use the notation ${O}(n^{a})$ and $\bTheta(n^{a})$, for some $a \in \R$, to denote the convergence a sequence of numbers to a constant (possibly zero) and a nonzero constant, respectively, as $n$ grows to infinity. In other words, $X_n = {O}(n^{a})$ means that the sequence $n^{-a} X_n$ converges to a constant that may be zero. The notation $X_n = \bTheta(n^{a})$ means that the sequence $n^{-a} X_n$ converges to a nonzero constant.  
Lastly, it follows that $X_n =  {O}(n^{a - \epsilon})$ means that the sequence $n^{- a}X_n$ converges to zero. 

We use similar notation for convergence of sequences of random variables. The notation $X_n = {O}_p(n^a)$ for some $a \in \R$ means that the sequence $n^{-a}X_n$ converges in distribution to a random variable (possibly a constant). The notation $X_n = o_p(n^a)$ for some $a \in \R$ means that the sequence $n^{-a}X_n$ converges in probability to zero. Finally, we define $X_n = \bTheta_p(n^a)$ to mean that $n^{-a}X_n$ converges in distribution to a random variable with distribution that is not a point mass at zero, and thus possibly a nonzero constant (as will always be the case in this section).

\subsection{Lemmas in support of Theorem~\ref{thm_mse}}
The first lemma describes the covariances of parameter estimates based on the errors, which arise in the proof of the mean-square error theorem. Of interest are the covariances ${\rm cov} ( \tilde{\phi}_v, \tilde{\phi}_w )$ for $( v,w ) \in \{3,4,5 \} \times \{3,4,5 \} $, as there are $\bTheta(n)$ times as many of these covariances in $\hat{V}_E$ as those covariances where at least one of $ v$ or $w$ is in $\{ 1,2\}$. However, we provide limiting values of all covariances for completeness. The proof of this lemma follows from recognizing that $\tilde{\phi}_v$ is a sample average and from defining all possible covariances that make up ${\rm cov} ( \tilde{\phi}_v, \tilde{\phi}_w )$ and their multiplicities.

\begin{lemma}
\label{lem:covariance_rates}
If ${\xi}$ is a mean zero random vector with positive definite covariance matrix in the exchangeable class, $\Omega = \sum_{i=1}^5 \phi_i \mathcal{S}_i$, and $\E[\xi_{jk}^4] < L < \infty$, then the covariance $n{\rm cov}\left( \tilde{\phi}_v, \tilde{\phi}_w \right)$ for $( v,w ) \in \{1,\ldots,5 \} \times \{1,\ldots,5 \} $ converges to
\begin{align} 
n{\rm cov}\left( \tilde{\phi}_v, \tilde{\phi}_w \right) &\rightarrow 
\begin{cases}
\sum_{i = 1}^4 \alpha_i  \beta_v \beta_w C \left(v,w \right)_i & (v,w) \in \{3,4,5\} \times \{3,4,5\} \\
\sum_{j=1}^3 \gamma_j F (v,w)_j & (v,w) \in \{1,2\} \times \{1,2\}, \\
\sum_{k = 1}^4 \gamma_k D (v,w)_k  &o.w.,
\end{cases}  \label{eq_cov_limit_lemma}\\
\text{where } &\alpha_i := 1 + {1}_{[i > 1]} + {1}_{[i = 4]}, \quad (i = 1,\ldots,4 ), \nonumber \\
&\beta_i := 1 + {1}_{[i = 5]}, \quad i=(1,\ldots,5 ), \nonumber  \\
&\gamma_i := 1 + {1}_{[i > 2]}, \quad i=(1,\ldots,4), \nonumber  
\end{align}
and $C   \left(v,w \right)_i$,  $D   \left(v,w \right)_i$, and $F   \left(v,w \right)_i$ are unknown finite constants equal to ${\rm cov}\left( \xi_{jk} \xi_{lm}, \xi_{rs} \xi_{tu}\right)$ for various configurations of the sets $\{j,k,l,m\}$ and $\{r,s,t,u \}$.
\end{lemma}

\begin{proof}
By definition, 
\begin{align}
n{\rm cov}\left( \tilde{\phi}_v, \tilde{\phi}_w \right) = n|\Theta_v|^{-1} |\Theta_w|^{-1} \sum_{(jk,lm) \in \Theta_v} \sum_{(rs,tu) \in \Theta_w} {\rm cov}\left(\xi_{jk} \xi_{lm}, \xi_{rs} \xi_{tu} \right). \label{eq_cov_lemma_def}
\end{align}
The sum is over $|\Theta_v| |\Theta_w|$ terms. Whenever $\{ j,k,l,m \}  \cap \{ r,s,t,u \} =  \varnothing$, the covariance is zero. This removes a power of $n$ from the sum in \eqref{eq_cov_lemma_def}, such that the sum is over ${O}( |\Theta_v| |\Theta_w| n^{-1} )$ possibly nonzero covariances. The scaled sum 
in \eqref{eq_cov_lemma_def} converges -- provided that the number of values that ${\rm cov}\left(\xi_{jk} \xi_{lm}, \xi_{rs} \xi_{tu} \right)$ can take is finite -- as each covariance is finite by assumption and the sequence of covariances is homogeneous as $n$ grows by exchangeability. In the remainder of the proof, we enumerate and define the covariances  ${\rm cov}\left(\xi_{jk} \xi_{lm}, \xi_{rs} \xi_{tu} \right)$ in \eqref{eq_cov_lemma_def} for particular pairs $( v,w ) \in \{1,\ldots,5 \} \times \{1,\ldots,5 \}$, showing that the number of values  that ${\rm cov}\left(\xi_{jk} \xi_{lm}, \xi_{rs} \xi_{tu} \right)$ can take is finite.
This is sufficient to establish convergence. 

We begin by analyzing the case of interest, that is  when both $v$ and $w$ are members of $\{3,4,5 \}$. As an example, we focus on $v=3$ and $w = 4$, where the first product of error relations corresponds to the same-sender covariance (b) in Figure~\ref{fig:directeddyad} and the second corresponds to the same-receiver covariance (c) in Figure~\ref{fig:directeddyad}. In this case, both $|\Theta_3| = |\Theta_4| = \bTheta(n^3)$.
When $v=3$ and $w=4$,  the covariance in \eqref{eq_cov_lemma_def} becomes
\begin{align}
n{\rm cov}\left( \tilde{\phi}_3, \tilde{\phi}_4 \right) =_a n^{-5} \sum_{jk} \sum_{l \notin \{j,k \}} \sum_{rs} \sum_{t \notin \{r,s \}} {\rm cov}\left(\xi_{jk} \xi_{jl}, \xi_{sr} \xi_{tr} \right), \label{eq_cov_lemma_vw34}
\end{align}
where `$=_a$' denotes equality in the limit as $n$ grows to infinity.

Only pairs of relation products that share a single actor will remain in the limit, as there are an order of $n$ fewer covariances resulting from pairs of relation products that share two actors. One such pair of relation products that share a single actor correspond to the case when $s=j$, i.e. ${\rm cov}\left(\xi_{jk} \xi_{jl}, \xi_{jr} \xi_{tr} \right)$, of which there are $\bTheta(n^{5})$ in the sum in \eqref{eq_cov_lemma_vw34}. There are $\bTheta(n^4)$ covariances corresponding to the case when $s=j$ and $r=k$, i.e. ${\rm cov}\left(\xi_{jk} \xi_{jl}, \xi_{jk} \xi_{tk} \right)$.
The values of all covariances in \eqref{eq_cov_lemma_vw34} are finite by assumption and not equal in general.
However, by exchangeability, covariances resulting from pairs of relations that share an actor in the same way {are} equal. For example, the covariance corresponding to $s=j$ is the same regardless of the node labeling, that is ${\rm cov}\left(\xi_{jk} \xi_{jl}, \xi_{jr} \xi_{tr} \right) = {\rm cov}\left(\xi_{ab} \xi_{ac}, \xi_{ad} \xi_{ed} \right)$ for any set $\{a,b,c,d,e \} \subset \{1,\ldots,n \}$ with $|\{a,b,c,d,e \}| = 5$. 

There are nine ways that we may have $|\{j,k,l \} \cap \{r,s,t \}| = 1$ in \eqref{eq_cov_lemma_vw34}, i.e. there are nine ways that exactly one of $\{j,k,l \}$ equals exactly one of $\{r,s,t \}$. However, these reduce into four unique covariance values for each pair $( v,w ) \in \{3,4,5 \} \times \{3,4,5 \}$. As an example, when $t=j$ the covariance is the same as that when $s=j$, that is ${\rm cov}\left(\xi_{jk} \xi_{jl}, \xi_{jr} \xi_{tr} \right) = {\rm cov}\left(\xi_{jk} \xi_{jl}, \xi_{sr} \xi_{jr} \right)$.
Now we define these four covariance values and their multiplicities out of the nine possible ways that exactly one of $\{j,k,l \}$ equals exactly one of $\{r,s,t \}$:
\begin{itemize}
\item When $r=j$, we define the covariance $C (3,4)_1 = {\rm cov}\left(\xi_{jk} \xi_{jl}, \xi_{sj} \xi_{tj} \right)$, of which there is one out of nine possible;
\item When $s=j$, the covariance is the same as when $t=j$ (multiplicity two), and we define this covariance $C (3,4)_2 = {\rm cov}\left(\xi_{jk} \xi_{jl}, \xi_{jr} \xi_{tr} \right)$;
\item When $r=k$, the covariance is the same as when $r=l$ (multiplicity two), and we define this covariance  $C(3,4)_3 = {\rm cov}\left(\xi_{jk} \xi_{jl}, \xi_{sk} \xi_{tk} \right)$;
\item We define the covariance when $s=k$ to be $C (3,4)_4 = {\rm cov}\left(\xi_{jk} \xi_{jl}, \xi_{kr} \xi_{tr} \right)$, of which there are four, the remaining terms of which correspond to $t=k$, $s=l$, and $t=l$.
\end{itemize}
Now, noting that there are $n^5 + \bTheta(n^{4})$ covariances in the sum \eqref{eq_cov_lemma_vw34}
corresponding to each of the nine possible ways that exactly one of $\{j,k,l \}$ equals exactly one of $\{r,s,t \}$, we see that
\begin{align}
n{\rm cov}\left( \tilde{\phi}_3, \tilde{\phi}_4 \right)
&\rightarrow{\rm cov}\left(\xi_{jk} \xi_{jl}, \xi_{sj} \xi_{tj} \right) + 2{\rm cov}\left(\xi_{jk} \xi_{jl}, \xi_{jr} \xi_{tr} \right)  \nonumber \\
&\hspace{.3in}  + 2{\rm cov}\left(\xi_{jk} \xi_{jl}, \xi_{sk} \xi_{tk} \right) + 4{\rm cov}\left(\xi_{jk} \xi_{jl}, \xi_{kr} \xi_{tr} \right),\label{eq_cov_vw34} \\
& = C(3,4)_1 + 2C(3,4)_2 + 2C(3,4)_3 + 4C(3,4)_4, \nonumber
\end{align} 
where `$\rightarrow$' denotes convergence in the limit as $n$ goes to infinity.  Under appropriate definition of $C(v,w)_i$ for $i \in \{1,2,3,4 \}$, the same argument applies when both $v$ and $w$ are one of $\{3,4 \}$. When $w=5$ (relation products of the form $\{\xi_{jk} \xi_{kl} \}$ and $\{\xi_{jk} \xi_{lj} \}$) and $v=3$, however, we then must consider covariances ${\rm cov}\left(\xi_{jk} \xi_{jl}, \xi_{rs} \xi_{tr} \right)$ {and} ${\rm cov}\left(\xi_{jk} \xi_{jl}, \xi_{sr} \xi_{rt} \right)$ from $w=5$, which doubles the coefficients in \eqref{eq_cov_vw34}. This accounts for $\beta_w = 2$ when $w=5$ and $\beta_w=1$ otherwise  in \eqref{eq_cov_limit_lemma}.  The same argument applies when $v=5$.

We now analyze both $v$ and $w$ in $\{1,2 \}$, corresponding to variance and the reciprocal covariance (a) in Figure~\ref{fig:directeddyad}. In this case, both $|\Theta_v| = |\Theta_w| = n(n-1)$. Taking $v=1$ and $w=1$ as an example, in the limit, the covariance in \eqref{eq_cov_lemma_def} is
\begin{align}
n{\rm cov}\left( \tilde{\phi}_1, \tilde{\phi}_2 \right) =_a n^{-3} \sum_{jk} \sum_{rs} {\rm cov}\left(\xi_{jk}^2, \xi_{rs}^2 \right). \label{eq_cov_def_vw11}
\end{align}
Again, we only consider covariances corresponding to pairs of relation products that share a single actor as only these covariances survive in the limit. There are four possible ways that $\{j, k\}$ shares exactly one actor with $\{r,s \}$. We define the three unique covariances and their multiplicities corresponding to the four ways that $\{j, k\}$ shares exactly one actor with $\{r,s \}$ as follows:
\begin{itemize}
\item When $r=j$, we define the covariance $F(1,1)_1 ={\rm cov}\left(\xi_{jk}^2, \xi_{js}^2 \right)$, of which there is one out of the four possibilities;
\item When $s=k$, 
we define the covariance $F(1,1)_2 ={\rm cov}\left(\xi_{jk}^2, \xi_{rk}^2 \right)$, of which there is one;
\item When $s=j$, we define the covariance $F(1,2)_3 = {\rm cov}\left(\xi_{jk}^2, \xi_{rj}^2 \right)$, which is the same as when $r=k$, accounting for the remaining two possibilities. 
\end{itemize}
Now, the fact that there are $n^3 + \bTheta(n^{2})$ covariances in the sum \eqref{eq_cov_def_vw11}
corresponding to each of the four possible ways that $\{j, k\}$ shares exactly one actor with $\{r,s \}$ gives that 
\begin{align}
n{\rm cov}\left( \tilde{\phi}_1, \tilde{\phi}_1 \right) &\rightarrow {\rm cov}\left(\xi_{jk}^2, \xi_{js}^2 \right) + {\rm cov}\left(\xi_{jk}^2, \xi_{rk}^2 \right) + 2{\rm cov}\left(\xi_{jk}^2, \xi_{rj}^2 \right) \nonumber \\ 
& = F (1,1)_1 +  F (1,1)_2 + 2 F (1,1)_3. \nonumber
\end{align}
Of course, the same argument applies to any $v$ and $w$ both in $\{1,2 \}$. In the case where $v=1$ and $w=2$, by symmetry, $F(1,2)_1 = F(1,2)_2$. Similarly, for $v=w=2$, all $F(2,2)_1 = F(2,2)_2 = F(2,2)_3$.

Similar counting arguments to those in the previous paragraphs apply when one of $v,w$ is in $\{ 3,4,5\}$ and the other is in $\{1,2 \}$. As an example, consider $v=1$ and $w = 3$. Once again, only pairs of relations that share a single actor will remain in the limit. Then, in the limit, the covariance in \eqref{eq_cov_lemma_def} becomes
\begin{align}
n{\rm cov}\left( \tilde{\phi}_2, \tilde{\phi}_3 \right) =_a n^{-4} \sum_{jk} \sum_{rs} \sum_{t \notin \{r,s \}} {\rm cov}\left(\xi_{jk}^2 , \xi_{rs} \xi_{rt} \right). \label{eq_cov_lemma_vw23}
\end{align}
Now, there are six ways in which the first pair of relations share an actor with the second pair, i.e. all sets with exactly one actor from $\{j,k \}$ equal to exactly one other from $\{r,s,t \}$. We define the covariances corresponding to the six possibilities below:
\begin{itemize}
\item When $r=j$, we define the covariance $D(1,3)_1 = {\rm cov}\left(\xi_{jk}^2 , \xi_{js} \xi_{jt} \right)$, of which there is one out of the six possibilities;
\item When $r=k$, we define the covariance $D(1,3)_2 = {\rm cov}\left(\xi_{jk}^2 , \xi_{ks} \xi_{kt} \right)$, of which there is one;
\item
The overlaps where $s=j$ and $t=j$ result in the same covariance (multiplicity two), which we define $D (1,3)_3 ={\rm cov}\left(\xi_{jk}^2 , \xi_{rj} \xi_{rt} \right)$;
\item
The overlaps where $s=k$ and $t=k$ result in the same covariance (multiplicity two), which we define $D (1,3)_4 ={\rm cov}\left(\xi_{jk}^2 , \xi_{rk} \xi_{rt} \right)$.
\end{itemize}
Then, noting that there are $n^4 + \bTheta(n^{3})$ covariances in the sum \eqref{eq_cov_lemma_vw23}
corresponding to each of the six possible ways that exactly one actor from $\{j,k \}$ is equal to exactly one other from $\{r,s,t \}$, we have that $n{\rm cov}\left( \tilde{\phi}_2, \tilde{\phi}_3 \right) $ converges to
{ 
\begin{align}
 & {\rm cov}\left(\xi_{jk}^2 , \xi_{js} \xi_{jt}\right) + {\rm cov}\left(\xi_{jk}^2 , \xi_{ks} \xi_{kt}\right) + 2{\rm cov}\left(\xi_{jk}^2 , \xi_{rj} \xi_{rt}\right) + 2{\rm cov}\left(\xi_{jk}^2 , \xi_{rk} \xi_{rt}\right), \label{eq_cov_vw13} \\
& = D (1,3)_1 + D (1,3)_2 + 2D (1,3)_3 + 2D (1,3)_4. \nonumber
\end{align}
}
When $D(v,w)_k$ for $k\in\{1,2,3,4 \}$ is appropriately defined, the same argument applies for all settings where one of $v,w$ is in $\{3, 4 \}$ and the other is in $\{1,2\}$. When $w=5$ (relation products of the form $\{\xi_{jk} \xi_{kl} \}$ and $\{\xi_{jk} \xi_{lj} \}$), however, we then must consider covariances in \eqref{eq_cov_lemma_vw23} ${\rm cov} (\xi_{jk}^2, \xi_{rs} \xi_{tr} )$ {and} ${\rm cov} ( \xi_{jk}^2, \xi_{sr} \xi_{rs} )$, which doubles the coefficients in \eqref{eq_cov_vw13}. This
accounts for $\beta_w = 2$ when $w=5$ and $\beta_w=1$ otherwise  in \eqref{eq_cov_limit_lemma}. When $v=2$, for example, we have the simplification that $D(2,3)_1 = D(2,3)_2$ and $D(2,3)_3 = D(2,3)_4$. 
\end{proof}

The expressions for the estimators based on the errors are simpler to analyze than those based on the residuals. For example, when comparing the mean-square errors of the exchangeable and dyadic clustering estimators, it is desirable to analyze $MSE_\xi ( \tilde{V}_E \mid X )$ instead of $MSE_\xi ( \hat{V}_E  \mid X )$. The following lemma allows us to do just this. This lemma states that mean-square errors of the estimators based on the errors are asymptotically equivalent to the mean-square error of those based on the residuals. The proof consists of first evaluating the mean-square error conditional on $X$.  We then show that $n^3 MSE_\xi ( \tilde{V}_E \mid  X )$ converges in $X$-probability to a nonzero constant in general and that $n^3 MSE_\xi ( \tilde{V}_E \mid  X ) - n^3 MSE_\xi ( \hat{V}_E \mid  X )$ converges in $X$-probability to zero, implying that the  difference between $MSE_\xi ( \tilde{V}_E \mid  X )$ and $MSE_\xi ( \hat{V}_E \mid X )$ is asymptotically negligible. We repeat the procedure for $MSE_\xi ( \tilde{V}_{DC} \mid X )$ and $MSE_\xi  ( \hat{V}_{DC} \mid X )$. 

\begin{lemma}
\label{lem_sufferrors1}
Assuming $\E\left( |x_{jk}^{(l)}|^8 \right) < L^8 < \infty$ for all $(l = 1, \ldots,p)$ and under the assumptions of Theorem~\ref{thm:AN}, the mean-square error for both the exchangeable and dyadic clustering estimators based on the residuals is asymptotically equivalent to the mean-square error of each respective estimator based on the errors. That is, 
\begin{align}
n^3 MSE_\xi \left( \hat{V}_{E} \mid X \right) =  n^3 MSE_\xi \left( \tilde{V}_{E} \mid X \right) +  {O}_p(n^{-1/2}) = {O}_p(1), \label{eq_lem_suff_errors_def}
\end{align}
and analogously for dyadic clustering. 
\end{lemma}

\begin{proof}
We will focus on the exchangeable estimator first, and then the dyadic clustering estimator. Throughout, we drop the conditioning on $X$ in the mean-square error as it is understood, for example $MSE_\xi ( \hat{V}_{E} ) = MSE_\xi ( \hat{V}_{E} \mid X )$.

We begin with the exchangeable estimator.
By definition, the mean-square error of the exchangeable estimator is
{ \small
\begin{align}
MSE_\xi \left( \hat{V}_E \right) &= \E \left\{ \left( \hat{V}_E - V^*\right)^2 \ \mid \ X \right\}, \nonumber \\
&= MSE_\xi \left( \tilde{V}_E \right) + \E \left\{  \left( \hat{V}_E - \tilde{V}_E \right)^2 \ \mid \ X \right\} + 2 \E \left\{ \left( \hat{V}_E - \tilde{V}_E \right) \left( \tilde{V}_E - V^* \right) \ \mid \ X \right\}, \label{eq_breakup_MSE_E}
\end{align} 
}
where $V^* = \var( \hat{\beta} )$, the true variance of $\hat{\beta}$. By the Cauchy-Schwarz inequality, 
{\small
\begin{align}
\E \left\{ \left( \hat{V}_E - \tilde{V}_E \right) \left( \tilde{V}_E - V^* \right)  \ \mid \ X \right\}  
&\le \left[  MSE_\xi \left( \tilde{V}_E  \right)  \E \left\{  \left( \hat{V}_E - \tilde{V}_E \right)^2\ \big| \ X \ \right\} \right]^{1/2}. \label{eq_E_cauchy_schwarz}
\end{align}
} 
If we show that $n^3 MSE_\xi \left( \tilde{V}_E \right)$ converges in $X$-probability to a constant, i.e. $MSE_\xi \left( \tilde{V}_E \right) = {O}_p(n^{-3})$, and that $\E \left\{  \left( \hat{V}_E - \tilde{V}_E \right)^2 \ \mid \ X \right\} = {O}_p(n^{-4})$, then \eqref{eq_E_cauchy_schwarz}
implies that the third additive term of
\eqref{eq_breakup_MSE_E} is ${O}_p(n^{-7/2})$. This is sufficient to establish \eqref{eq_lem_suff_errors_def}. 
We begin with showing $n^3 MSE_\xi \left( \tilde{V}_E \right) = {O}_p(1)$. By definition, the scaled mean-square error is $n^{3}MSE_\xi \left( \tilde{V}_{E} \right) 
= n^{3} \E \left\{ {\rm tr} \left( \tilde{V}_{E}^2 \right) \ \big| \ X \right\}$, which is equal to
{\footnotesize
\begin{align}
&
 \sum_{v = 1}^5 \sum_{w = 1}^5  n{\rm cov}\left(\tilde{\phi}_v, \tilde{\phi}_w \right){\rm tr}\left\{ \left( \frac{X^T X}{n^2} \right)^{-1} \left( \frac{X^T \s{S}_v X}{n^{3}} \right)\left( \frac{X^T X}{n^2} \right)^{-2}  \left(  \frac{X^T \s{S}_w X}{n^{3}} \right) \left( \frac{X^T X}{n^2} \right)^{-1} \right\}. \label{eq_mse_xi_E_1}
\end{align} 
}
By Lemma~\ref{lem:covariance_rates}, $n{\rm cov} ( \tilde{\phi}_v, \tilde{\phi}_w  )$ converges to a finite constant for every $( v,w ) \in \{1,\ldots,5 \} \times \{1,\ldots,5 \} $. The convergence in probability of each multiplicative term in \eqref{eq_mse_xi_E_1} containing $X$ is defined by assumption (B1); only those with both $v$ and $w$ in $\{3,4,5 \}$ survive in the limit as these have $|\Theta_v| = \bTheta(n^3)$ whereas $|\Theta_v| = \bTheta(n^2)$ for $v \in \{1,2 \}$. Thus, we have that
\begin{align}
n^{3}MSE_\xi \left( \tilde{V}_{E} \right) 
&\rightarrow_{pr_X}  \sum_{v=3}^5 \sum_{w=3}^5 \sum_{i=1}^4 \alpha_i \beta_v \beta_w C(v,w)_i {\rm tr} \left( M_1^{-1} M_v M_1^{-2} M_w M_1^{-1} \right), \label{eq_exch_error_mse_convergence}
\end{align} 
which is finite., and where `$\rightarrow_{pr_X} $' denotes convergence in $X$-probability. 

It remains to show that $\E \left\{  \left( \hat{V}_E - \tilde{V}_E \right)^2 \ \mid \ X \right\} = {O}_p(n^{-4})$. To establish this fact, it is sufficient to show that $\hat{V}_E - \tilde{V}_E = {O}_p(n^{-2})$, and then, by the continuous mapping theorem, $\left( \hat{V}_E - \tilde{V}_E \right)^2 = {O}_p(n^{-4})$, which implies the desired result. Writing directly, 
\begin{align}
\hat{V}_E - \tilde{V}_E &= \frac{1}{n}\sum_{v = 1}^5  \left(\hat{\phi}_v  - \tilde{\phi}_v  \right) \left( \frac{X^T X}{n^2} \right)^{-1} \left(\frac{X^T \s{S}_v X}{n^{3}} \right)\left( \frac{X^T X}{n^2} \right)^{-1}.
\end{align}
By assumption (B1), the multiplicative terms involving $X$ converge in probability to constants. To establish $\hat{V}_E - \tilde{V}_E = {O}_p(n^{-2})$, it is sufficient show that $\hat{\phi}_v  - \tilde{\phi}_v = {O}_p(n^{-1})$ for all $v \in \{1,\ldots,5 \}$. Writing this expression directly, the difference $\hat{\phi}_v  - \tilde{\phi}_v$ is
{\small
\begin{align}
 -(\hat{\bbeta} - \bbeta)^T \left( \sum_{(jk,lm) \in \Theta_v}\frac{\x_{jk} \xi_{lm} + \x_{lm} \xi_{jk}}{| \Theta_v |}  \right) + (\hat{\bbeta} - \bbeta)^T \left( \sum_{(jk,lm) \in \Theta_v} \frac{\x_{jk} \x_{lm}^T}{|\Theta_v|}  \right)(\hat{\bbeta} - \bbeta). \label{eq_hat_tilde_phi}
\end{align}
}
By Theorem~\ref{thm:AN}, $\hat{\bbeta} - \bbeta = {O}_p(n^{-1/2})$. Also, by assumption (B1), the sum involving $X$ in the second term converges in probability to a constant; thus, the second additive term in \eqref{eq_hat_tilde_phi} is ${O}_p(n^{-1})$. Turning to the first additive term, we notice its expectation is zero, that is $\E( \x_{jk} \xi_{lm} ) = 0$ for all relations $jk$ and $lm$. The variance is 
{\small
\begin{align}
 \var \left( \sum_{(jk,lm) \in \Theta_v}\frac{\x_{jk} \xi_{lm}}{|\Theta_v|} \right)
 &= \frac{1}{|\Theta_v|^2}\sum_{(jk,lm) \in \Theta_v} \sum_{(rs,tu) \in \Theta_v} \E \left( \x_{jk} \x_{rs}^T \xi_{lm} \xi_{tu} \right) = {O}(n^{-1}), \nonumber
\end{align}
}
where we use the fact that $\E \left( \x_{jk} \x_{rs}^T \xi_{lm} \xi_{tu} \right)$ is only nonzero when relation $lm$ shares an actor with relation $tu$ since $\E \left( \xi_{lm} \xi_{tu} \right) = 0$ whenever $\{j, k\} \cap \{l,m \} = \varnothing$ and $X$ is independent ${\xi}$ by assumption (B3). This fact removes at least a factor of $n$ from the sum. Thus, we have that $\sum_{(jk,lm) \in \Theta_v} {\x_{jk} \xi_{lm} + \x_{lm} \xi_{jk}} / {| \Theta_v |} = {O}_p(n^{-1/2})$, which gives that $\hat{\phi}_v  - \tilde{\phi}_v = {O}_p(n^{-1})$ and 
\begin{align}
\E \left\{  \left( \hat{V}_E - \tilde{V}_E \right)^2 \ \mid \ X \right\} ={O}_p(n^{-4}), \nonumber 
\end{align}
which establishes \eqref{eq_lem_suff_errors_def} for the exchangeable estimator. 

The same argument following \eqref{eq_breakup_MSE_E} applies to the dyadic clustering estimator. To establish \eqref{eq_lem_suff_errors_def} for the dyadic clustering estimator, it is thus sufficient to show that $n^3 MSE_\xi \left( \tilde{V}_{DC} \right)$ converges in $X$-probability to a constant and that $\E \left\{  \left( \hat{V}_{DC} - \tilde{V}_{DC} \right)^2 \ \mid \ X \right\} = {O}_p(n^{-4})$. We begin with the former. By definition, $n^{3} MSE_\xi \left( \tilde{V}_{DC} \right) $ is
{
\begin{align}
&\frac{1}{n^{5}} \sum_{(jk,lm) \in \Theta_0} \sum_{(rs,tu) \in \Theta_0} {\rm cov}\left(\xi_{jk} \xi_{lm}, \xi_{rs} \xi_{tu} \right) \times \ldots \nonumber \\
&\hspace{.2in}\ldots \times {\rm tr}\left\{ \left( \frac{X^T X}{n^2}\right)^{-1} \x_{jk} \x_{lm}^T \left( \frac{X^T X}{n^2}\right)^{-2} \x_{rs} \x_{tu}^T \left( \frac{X^T X}{n^2}\right)^{-1} \right\}, 
\end{align} \nonumber
}
where $\Theta_0$ is the set of relation pairs that share an actor in {any} manner. 
Then, substituting the asymptotic values for ${\rm cov}\left(\xi_{jk} \xi_{lm}, \xi_{rs} \xi_{tu} \right)$ from Lemma~\ref{lem:covariance_rates} and separating the sum by the five ways that two relations may share an actor, 
{
\begin{align}
n^{3} MSE_\xi \left( \tilde{V}_{DC} \right) 
&=_a \frac{1}{n^5}\sum_{v = 3}^5 \sum_{w = 3}^5 \sum_{i=1}^4 \sum_{T(v,w)_i}   C(v,w)_i \  \x_{lm}^T M_1^{-2} \x_{rs} \x_{tu}^T M_1^{-2} \x_{jk},\label{eq_mse_xi_DC_2}
\end{align} 
}
where `$=_a$' denotes equality in the limit and $T(v,w)_i$ is the set of relations $(jk,lm,rs,tu)$ such that  $(jk,lm) \in \Theta_v$ and $(rs, tu) \in \Theta_w$ and such that the pairs of relations $(jk,lm)$ and $(rs,tu)$ share a single actor as appropriate for $C(v,w)_i$  in Lemma~\ref{lem:covariance_rates}. In \eqref{eq_mse_xi_DC_2}, we substitute the limit of $\left( {X^T X} / {n^2}\right)^{-2}$ from assumption (B1). Also in \eqref{eq_mse_xi_DC_2}, only terms with $v$ and $w$ both in $\{3,4,5 \}$ survive in the limit as the set $|T(v,w)_i| = \bTheta(n^5)$ (as detailed in Lemma~\ref{lem:covariance_rates}), while the order is less for either $v$ or $w$ in $\{1,2 \}$, so these terms vanish in the limit. Evaluating the vector products, the expression on the right hand side of \eqref{eq_mse_xi_DC_2} equal to
{\small
\begin{align}
\sum_{v = 3}^5 \sum_{w = 3}^5 \sum_{i=1}^4 \sum_{a=1}^p \sum_{b=1}^p \sum_{c=1}^p \sum_{d=1}^p C(v,w)_i \ \left( m_1^{-2} \right)_{ab} \left( m_1^{-2} \right)_{cd} \left( \frac{1}{n^{5}} \sum_{T(v,w)_i}  x_{lm}^{(a)} x_{rs}^{(b)} x_{tu}^{(c)} x_{jk}^{(d)} \right), \label{eq_DC_singles}
\end{align}
}
where $\left( m_1^{-2} \right)_{ab}$ is the $(a,b)$ entry in $M_1^{-2},$ e.g., and $x_{jk}^{(a)}$ is the entry in $X$ pertaining to column $a$ and relation $jk$. Further, the variance $\var \left( {n^{-5}} \sum_{T(v,w)_i}  x_{lm}^{(a)} x_{rs}^{(b)} x_{tu}^{(c)} x_{jk}^{(d)} \right)$ is
{
\begin{align}
&\frac{1}{n^{10}}\sum_{T(v,w)_i} \sum_{U(v,w)_i} {\rm cov}\left( x_{lm}^{(a)} x_{rs}^{(b)} x_{tu}^{(c)} x_{jk}^{(d)}, x_{ef}^{(a)} x_{gh}^{(b)} x_{np}^{(c)} x_{yz}^{(d)}\right),  \label{eq_var_dc_singles} \\
&= \frac{\bTheta(n^9)}{n^{10}} \rightarrow 0, \nonumber
\end{align}
}
where $U(v,w)_i = T(v,w)_i$ and $(lm,rs,tu,jk)$ indexes the first sum and $(ef,gh,np,yz)$  indexes the second sum. The convergence is the result of the independence portion of assumption in (B1) and the bounded moment assumption on $X$. The variance in \eqref{eq_var_dc_singles} converges to zero for every set of covariates $\{a,b,c,d \}$, every relation type  $v$ and $w$ both in $\{3,4,5 \}$, and every covariance type $i \in \{1,2,3,4 \}$. Thus, provided that the expectation of $n^{-5} \sum_{T(v,w)_i}  x_{lm}^{(a)} x_{rs}^{(b)} x_{tu}^{(c)} x_{jk}^{(d)}$ converges to a constant, this expression converges in probability to that same constant. This expectation is
\begin{align}
\E\left( x_{lm}^{(a)} x_{rs}^{(b)} x_{tu}^{(c)} x_{jk}^{(d)} \right) &= {\rm cov}\left(x_{jk}^{(d)} x_{lm}^{(a)}, x_{rs}^{(b)} x_{tu}^{(c)}  \right) + \left( m_v \right)_{ad} \left( m_w \right)_{bc}, \label{eq_dc_covs_abcd}
\end{align}
where $\left( m_v \right)_{ij}$ is the $(i,j)$ entry in $M_v$ and we use the symmetry of $M_v$ for all $v \in  \{ 3,4,5\}$. Unlike ${\xi}$, for a given $i \in \{1,2,3,4 \}$, the covariances in \eqref{eq_dc_covs_abcd} may not be the same for two relation sets in $T(v,w)_i$. However, by assumption (B1) and taking $i=1$, $v=3$, and $w=4$ for example, we still have that 
\begin{align}
{\rm cov}\left(x_{jl}^{(d)} x_{sj}^{(a)}, x_{tj}^{(b)} x_{jk}^{(c)} \right) = {\rm cov}\left(x_{ef}^{(d)} x_{ge}^{(a)}, x_{he}^{(b)} x_{ep}^{(c)} \right), \label{eq_cov_X_ways}
\end{align}
for $ |\{j,k,l,s,t,e,f,g,h,p \}| = 10$. That is, covariances that share an actor in the same way are still equal. So, for fixed $i \in \{1,2,3,4 \}$ and pair of $v$ and $w$ both in $\{3,4,5 \}$, we may collect the $\alpha_i$ possible covariances and average them to attain the convergent value. We thus define the limit
{
\begin{align}
 \frac{1}{n^5}\sum_{T(v,w)_i}{\rm cov}\left(x_{jk}^{(d)} x_{lm}^{(a)}, x_{rs}^{(b)} x_{tu}^{(c)} \right) \rightarrow & \ \alpha_i \beta_v \beta_w \frac{1}{\alpha_i} \sum_{W(v,w)_i} {\rm cov}\left(x_{jk}^{(d)} x_{lm}^{(a)}, x_{rs}^{(b)} x_{tu}^{(c)} \right), \label{eq_cov_conv_Xa} \\
 &:=\alpha_i \beta_v \beta_w C_X^{(d,a,b,c)}(v,w)_i, \label{eq_cov_conv_X}
\end{align}
}
where $W(v,w)_i$ is the set of  $\alpha_i$ {ways} that $(jk, lm, rs, tu)$ correspond to $T(v,w)_i$.   
For example, when $i=4$, $v=3$, and $w=4$, $W(v,w)_i$ contains four index sets corresponding to the four multiplicities of $C(v,w)_4$ as defined in Lemma (3).   
The convergence of \eqref{eq_cov_conv_Xa} results from \eqref{eq_cov_X_ways}. As the average over $W(v,w)_i$ is over a finite number of terms since each $\alpha_i$ is bounded, there is no possibility of divergence. 
The covariances in \eqref{eq_cov_conv_Xa} are finite by assumption (B2) and $\beta_v \beta_w$ arises from the asymptotic limit of $n^{-5}|T(v,w)_i|$.
Taking \eqref{eq_cov_conv_Xa} together with \eqref{eq_dc_covs_abcd}, we have the convergence of the  expectation of $n^{-5} \sum_{T(v,w)_i}  x_{lm}^{(a)} x_{rs}^{(b)} x_{tu}^{(c)} x_{jk}^{(d)}$. Along with \eqref{eq_var_dc_singles}, convergence of the expectation of $n^{-5} \sum_{T(v,w)_i}  x_{lm}^{(a)} x_{rs}^{(b)} x_{tu}^{(c)} x_{jk}^{(d)}$  establishes the convergence of $n^{-5} \sum_{T(v,w)_i}  x_{lm}^{(a)} x_{rs}^{(b)} x_{tu}^{(c)} x_{jk}^{(d)}$ to the same limit.

Now, for a particular $v$ and $w$ both in $\{3,4,5 \}$ and $i \in \{ 1,2,3,4\}$, we collect the set of $\{ C_X^{(d,a,b,c)}(v,w)_i \}$ for every $a,b,c,$ and $d$ in $\{1,2,\ldots,p \}$ from \eqref{eq_cov_conv_X}  into a $p^2 \times p^2$ matrix defined $D_X(v,w)_i$. Substituting this definition into \eqref{eq_DC_singles} while noting each $\{ M_v\}_{v=1}^5$ is symmetric, the convergent value for the dyadic clustering estimator is 
{
\begin{align}
&n^{3} MSE_\xi \left( \tilde{V}_{DC} \right) 
\rightarrow_{pr_X} \sum_{v = 3}^5 \sum_{w = 3}^5 \sum_{i=1}^4  \alpha_i \beta_v \beta_w C(v,w)_i \times \ldots \label{eq_mse_xi_DC_convergence} \\ 
&\hspace{.2in}\ldots \times \left\{ {\rm vec}\left( M_1^{-2} \right)^T D_{X}(v,w)_i{\rm vec}\left( M_1^{-2} \right) + {\rm tr}\left(M_1^{-2} M_v M_1^{-2} M_w \right)\right\}. \nonumber
\end{align} 
}

Since the convergent value in \eqref{eq_mse_xi_DC_convergence} is a finite constant, it remains to show that \[\E \left\{  \left( \hat{V}_{DC} - \tilde{V}_{DC} \right)^2 \ \mid \ X \right\} = {O}_p(n^{-4}).\] As with the exchangeable estimator, it is sufficient to show that $\hat{V}_{DC} - \tilde{V}_{DC} = {O}_p(n^{-2})$. Using the residual definition $e_{jk} = \xi_{jk} + \x_{jk}^T (\bbeta - \hat{\bbeta})$, the expression for $\hat{V}_{DC} - \tilde{V}_{DC}$ is
{\footnotesize
\begin{align}
 &\ \frac{1}{n}\sum_{v = 1}^5 \sum_{(jk,lm) \in \Theta_v}  \left( \frac{X^T X}{n^2} \right)^{-1} \left\{ \frac{\x_{jk} \x_{lm}^T }{n^{3}} \left(e_{jk} e_{lm} - \xi_{jk} \xi_{lm} \right) \right\} \left( \frac{X^T X}{n^2} \right)^{-1}  \label{eq_dc_e_xi_diff} \\
& \hspace{.25in}  =_a \frac{1}{n}\sum_{v = 1}^5 \sum_{(jk,lm) \in \Theta_v}  M_1^{-1} \left[ \frac{\x_{jk} \x_{lm}^T }{n^{3}} \left\{  \x^T_{jk} \left( \hat{\bbeta} - \bbeta \right)  \left( \hat{\bbeta} - \bbeta \right)^T \x_{lm} \right\}  -  2\x^T_{jk} \left( \hat{\bbeta} - \bbeta \right) \xi_{lm}   \right] M_1^{-1}, \nonumber 
\end{align}
}
where we substitute the convergence of $\left( {X^T X} / {n^2} \right)^{-1}$ to $M_1^{-1}$ and have used the exchangeability property to get the factor of two on the second additive term in the center of \eqref{eq_dc_e_xi_diff}. Analyzing the first additive term in the center of \eqref{eq_dc_e_xi_diff}, 
{
\begin{align}
&\sum_{v = 1}^5\sum_{(jk,lm) \in \Theta_v} \frac{\x_{jk} \x_{lm}^T }{n^{3}} \left\{  \x^T_{jk} \left( \hat{\bbeta} - \bbeta \right)  \left( \hat{\bbeta} - \bbeta \right)^T \x_{lm} \right\} \nonumber \\
&\hspace{.2in} = \sum_{jk} \sum_{lm \in \Theta_{jk}}  \left(  \frac{\x_{jk} \x_{jk}^T}{n^2} \right)  \left( \hat{\bbeta} - \bbeta \right)  \left( \hat{\bbeta} - \bbeta \right)^T  \left( \frac{\x_{lm} \x_{lm}^T}{n} \right), \label{eq_dc_e_xi_diff_term1}\\
&\hspace{.23in} =\bTheta_p(1) {O}_p(n^{-1/2}) {O}_p(n^{-1/2}) \bTheta_p(1) = {O}_p(n^{-1}), \nonumber
\end{align}
}
 recalling the notation that $\Theta_{jk}$ is the set of all relations that share an actor with relation $jk$ and that $|\Theta_{jk}| = \bTheta(n)$. We attain the convergence rate by noting that the $X$-terms in \eqref{eq_dc_e_xi_diff_term1} converge in probability to constants by assumption (B1) and $\hat{\bbeta} - \bbeta = {O}_p(n^{-1/2})$ by Theorem~\ref{thm:AN}. The convergences in \eqref{eq_dc_e_xi_diff_term1} are for $p\times p$ matrices; these convergences are element-wise.

We now analyze the convergence rate of the second additive term in the center of \eqref{eq_dc_e_xi_diff}, 
{\small
\begin{align}
\sum_{v = 1}^5\sum_{(jk,lm) \in \Theta_v} \frac{\x_{jk} \x_{lm}^T }{n^{3}} \left\{  \x^T_{jk} \left( \hat{\bbeta} - \bbeta \right)  \xi_{lm}  \right\} 
&= \sum_{lm} \sum_{jk \in \Theta_{lm}} \left( \frac{\x_{jk} \x_{jk}^T}{n} \right) \left( \frac{\xi_{lm} \x_{lm}^T }{n^2} \right)  \left( \hat{\bbeta} - \bbeta \right), \nonumber \\
&=\bTheta_p(1) {O}_p(n^{-1/2}) {O}_p(n^{-1/2}) =  {O}_p(n^{-1}). \nonumber
\end{align}
}
Again, the convergence of the first multiplicative term is a result of assumption (B1) and  $\hat{\bbeta} - \bbeta = {O}_p(n^{-1/2})$ by Theorem~\ref{thm:AN}. The mean $n^{-2} \sum_{lm} \x_{lm}\xi_{lm}$ is expectation zero and ${O}_p(n^{-1/2})$ by previous arguments, for example in \eqref{eq:vector_an}. Thus, we have $\hat{V}_{DC} - \tilde{V}_{DC} = {O}_p(n^{-2})$, and the dyadic clustering estimator satisfies the relation in \eqref{eq_lem_suff_errors_def}.
\end{proof}

\subsection{Proof of Theorem~\ref{thm_mse}}
We now establish that the mean-square error of the exchangeable estimator is less than that of the dyadic clustering estimator with high probability. To do so, we show that the value to which the difference in mean-square errors converges is nonnegative. Throughout the proof, we drop the conditioning on $X$ in the mean-square error as it is understood, for example $MSE_\xi ( \hat{V}_{E} ) = MSE_\xi ( \hat{V}_{E} \mid X )$.

The asymptotic difference in mean-square errorss is as follows, where we substitute the expressions for the estimators based on the errors in \eqref{eq_exch_error_mse_convergence} and \eqref{eq_mse_xi_DC_convergence}, as justified by Lemma~\ref{lem_sufferrors1}:
{ 
\begin{align}
&n^{3} \left\{ MSE_\xi \left( \hat{V}_{DC} \right) - MSE_\xi \left(  \hat{V}_{E} \right)\right\}   \nonumber \\
&\hspace{.3in}\rightarrow_{pr_X}  \sum_{v=3}^5 \sum_{w=3}^5 \sum_{i=1}^4 \alpha_i\beta_v \beta_w  C(v,w)_i {\rm vec}\left( M_1^{-2} \right)^T D_{X}(v,w)_i{\rm vec}\left( M_1^{-2} \right). \label{eq_diff_single_cov1}
\end{align}
}
It remains to show that this is a nonnegative constant. To do so, we show that the matrix in the quadratic form in \eqref{eq_diff_single_cov1} is the limit of a variance matrix, and thus positive semi-definite. We will show that the scaled variance
{ 
\begin{align}
&\frac{1}{n^5} \var \left\{ \sum_{v=1}^5 \sum_{jk, lm \in \Theta_v} \xi_{jk} \xi_{lm} {\rm vec} \left( \x_{jk} \x_{lm}^T \right) \right\}   \nonumber \\
&\hspace{.3in}= \frac{1}{n^5} \sum_{v=1}^5 \sum_{w=1}^5 \sum_{jk, lm \in \Theta_v} \sum_{rs, tu \in \Theta_w} {\rm cov}\left\{ \xi_{jk}\xi_{lm} {\rm vec} \left( \x_{jk} \x_{lm}^T \right),\  \xi_{rs}\xi_{tu} {\rm vec} \left( \x_{rs} \x_{tu}^T \right) \right\}, \label{eq_vlimit_mse}
\end{align}
}
converges to the desired matrix. First, the sum in \eqref{eq_vlimit_mse} is $\bTheta(n^5)$ as the relations $jk$ and $lm$ must share at least one actor with the relations $rs$ and $tu$ for the covariance to be nonzero. Then, by the arguments in Lemma~\ref{lem:covariance_rates}, 
only pairs of relations that share a single actor survive in the limit. Finally, by assumption (B3), $X$ is independent ${\xi}$ and the variance in \eqref{eq_vlimit_mse} is asymptotically equivalent to
\begin{align}
&\frac{1}{n^5}\sum_{v=3}^5 \sum_{w=3}^5 \sum_{jk, lm \in \Theta_v} \sum_{rs, tu \in \Theta_w} \E\left\{ \left(\xi_{jk} \xi_{lm} - \phi_v\right) \left(\xi_{rs} \xi_{tu} - \phi_w\right) \right\}\times \ldots \nonumber \\ 
&\ldots \times \E\left[ \left\{ {\rm vec} \left( \x_{jk} \x_{lm}^T \right) - {\rm vec} \left( M_v \right) \right\} \left\{ {\rm vec} \left( \x_{rs} \x_{tu}^T \right) - {\rm vec} \left( M_w \right) \right\}^T \right],  \nonumber
\end{align}
where only terms with both $v$ and $w$ in $\{3,4,5 \}$ survive in the limit. Then, by assumptions (A1) and (B1) and applying Lemma~\ref{lem:covariance_rates}, the variance converges to 
{
\begin{align}
&\frac{1}{n^5} \var \left\{ \sum_{v=1}^5 \sum_{jk, lm \in \Theta_v} \xi_{jk} \xi_{lm} {\rm vec} \left( \x_{jk} \x_{lm}^T \right) \right\}  \rightarrow  \sum_{v=3}^5 \sum_{w=3}^5 \sum_{i=1}^4 \alpha_i \beta_v \beta_w C(v,w)_i D_X(v,w)_i, \label{eq_central_matrix_limit}
\end{align}
}
where we substitute the definition of $D_X(v,w)_i$ following \eqref{eq_cov_conv_X}.  Thus, the matrix in \eqref{eq_central_matrix_limit} is positive semi-definite. Now, \eqref{eq_diff_single_cov1} becomes
{
\begin{align}
&n^{3} \left\{ MSE_\xi \left(  \hat{V}_{DC} \right) - MSE_\xi \left(  \hat{V}_{E} \right)\right\}  \nonumber \\
& \hspace{.25in} \rightarrow_{pr_X} {\rm vec}\left( M_1^{-2} \right)^T \left\{ \sum_{v=3}^5 \sum_{w=3}^5 \sum_{i=1}^4  \alpha_i \beta_v \beta_w C(v,w)_i D_X(v,w)_i \right\} {\rm vec}\left( M_1^{-2} \right) \ge 0. \nonumber 
\end{align}
}
\qed

\section{Simulation study}
\label{app_sim}

\subsection{Details of simulation study}
We simulated data from a linear regression model with an exchangeable error model, independent and identically distributed errors, and a non-exchangeable error model. 
For each error setting, we employed the following three-covariate regression model:
\begin{equation}
y_{ij}=\beta_1+\beta_2 {1}_{(x_{2i} \in C)} {1}_{(x_{2j} \in C)}+\beta_3|x_{3i}-x_{3j}|+\beta_4x_{4ij}+\xi_{ij}, 
\label{simreg}
\end{equation}  
where $1_{(.)}$ denotes the indicator function. 
In this model, $\beta_1$ is an intercept; $\beta_2$ is a coefficient on a binary indicator of whether individuals $i$ and $j$ both belong to a pre-specified class $C \subset (1,\ldots,n )$; $\beta_3$ is a coefficient on the absolute difference of a continuous, actor-specific covariate $x_{3i}$; and $\beta_4$ is that for a pair-specific continuous covariate $x_{4ij}$. For all simulations, we fixed true coefficients $\beta = (1,1,1,1)^T$. We drew each $x_{2i}$ from a  Bernoulli$(1/2)$ distribution independently. In the rare event that  $x_{2i} = x_{2j}$ for all $(i,j)$ pairs, one realization $x_{2k}$ was randomly flipped to a $1$ or $0$. All $x_{3i}$ and $x_{4ij}$ were drawn independently from a standard normal distribution.

Each error setting was specified to have the same total variance: 
\[\sum_{i j} \var (\xi_{ij}) = 3n(n-1).\] This variance was chosen so that the variance of the error would be similar to that of the regression mean model $\beta^T x_{ij}$. In the independent and identifically distributed errors setting,  $\xi_{ij} \sim_{iid}  {\rm N}(0, 3)$ for all $(i,j)$. To generate the non-exchangeable errors, a mean-zero random effect was added to the upper left quadrant of $\var( \xi )$ under independent and identically distributed errors. The errors for the non-exchangeable error setting may be written
\[\xi_{ij} = \tau {1}_{( i \le \left \lfloor{n/2} \right \rfloor )} {1}_{( j \le \left \lfloor{n/2}\right \rfloor ) } + \epsilon_{ij}, \ \ \ \ \tau \sim {\rm N}\left(0, \frac{9n}{4\left \lfloor{n/2}\right \rfloor } \right), \ \ \ \ \epsilon_{ij} \sim_{iid} {\rm N}(0, 3/4),\]
where `${1}_{( j \le \left \lfloor{n/2}\right \rfloor)}$' is an indicator of index $i$ less than or equal to the floor of $n/2$, for example. 
Since noise was added to actor relations in the same position in $\Xi$ in each simulation run, the distribution of the relations was not exchangeable:  the distribution of the errors would be different for a reordering of the rows and columns of the error array.

Finally, the distribution of the exchangeable (bilinear mixed effects model) error setting is defined in~\cite{hoff2005bilinear} and detailed in full below. The bilinear mixed effects model  is a generalization of the ``additive common shocks'' error structure used in simulation studies to justify the dyadic clustering estimator (\cite{cameron2014robust, aronow2015cluster, tabord2017inference}). 
\begin{align}
y_{ij} &= \beta^T x_{ij} +  \xi_{ij}, \hspace{.6in} \xi_{ij} = a_i + b_j + z_i^Tz_j + \gamma_{(ij)} + \epsilon_{ij},  \label{eq:ble_full}  \\
\noalign{\medskip}  \nonumber
(a_i,b_i) &\sim \text{N}(0_2,\Sigma_{ab}), \hspace{.6in} \Sigma_{ab} =  \left( \begin{array}{cc}
\sigma^2_{a} & \rho_{ab}\sigma_a \sigma_b  \nonumber  \\
\rho_{ab} \sigma_a \sigma_b & \sigma^2_b  
 \end{array} \right),  \nonumber  \\
\noalign{\medskip}  \nonumber
 z_i, z_j  \sim  \text{N}_d(0,&\sigma_z^2 I_d),    \hspace{.4in} \gamma_{(ij)} = \gamma_{(ji)} \sim \text{N}(0, \sigma^2_{\gamma}),    \hspace{.4in} \epsilon_{ij} \sim \text{N}(0, \sigma^2_{\epsilon}),  \nonumber  
\end{align}
where $a_i$, $b_j$, $z_i$, $z_j$, and $\epsilon_{ij}$ are independent.  
We selected the dimension of the latent space to be $d=2$, the correlation between sender and receiver effects as $\rho_{ab} = 1/2$, and the sender variance to be twice that of receiver variance: $\sigma_a^2 = 2 \sigma_b^2$. We further specified $\sigma_z = \sigma_\gamma = \sigma_b$. Finally, we selected $\sigma^2_\epsilon = 3/4$. With the aforementioned choices, the restriction $\sum_{i j}\var ( \xi_{ij} ) = 3n(n-1)$ generated a quadratic equation in $\sigma_b$. The standard deviations that resulted from solving this quadratic equation are shown in Table \ref{tab:sim_exchsds}.

\begin{table}
    \centering
    \caption{\label{tab:sim_exchsds}Approximate standard deviations for exchangeable error setting.}
    \begin{tabular}{ccccc}
    $\sigma_\epsilon$ & $\sigma_a$ & $\sigma_b$ & $\sigma_\gamma$ & $\sigma_z$ \\
0.866 & 0.957  & 0.677 & 0.677 & 0.677 \\
    \end{tabular}
\end{table}

We generated 500 random realizations of covariates for relational data sets with numbers of actors $n=20,40,80,160,320$. For each covariate realization, 1,000 random error realizations were generated for each of the three error settings: independent and identically distributed, exchangeable, and non-exchangeable. Using the model \eqref{simreg}, a simulated data set was created from each covariate realization and error realization pair.  The regression model was fit to each data set using ordinary least squares, and standard errors were estimated using the exchangeable and dyadic clustering
sandwich variance estimators.  Confidence interval coverage was estimated {for each covariate realization} by counting the fraction of confidence intervals  that contained the true coefficient. Bias and variance of the standard error estimators was estimated {for each covariate realization} relative to the known true standard errors of the ordinary least squares estimator, conditional on the covariate realization.

\subsection{Bias in simulation study}
\label{sec:sim_bias}

In addition to the estimated coverages in Figure~\ref{fig:sim_coverage}, we plot the estimated bias of the standard error estimators in Figure~\ref{fig:sim_bias}. We note that the bias of the dyadic clustering estimator is generally substantially larger in absolute value, especially in the independent and exchangeable error cases, and most standard error estimators are biased downwards.

 \begin{figure}[ht!]
\centering
 \begin{tabular}{rccc}
  \hspace{-.2in}&\hspace{-.2in}
    \hspace{1pt} binary
  \hspace{-.22in}&\hspace{-.22in}
  \hspace{1pt} positive 
  \hspace{-.22in}&\hspace{-.22in}
  \hspace{1pt} real-valued \\ 
  \begin{sideways}\hspace{.4in} Estimated bias (\% error) \end{sideways} 
  \hspace{-.2in}&\hspace{-.2in}
\includegraphics[width=.315\textwidth]{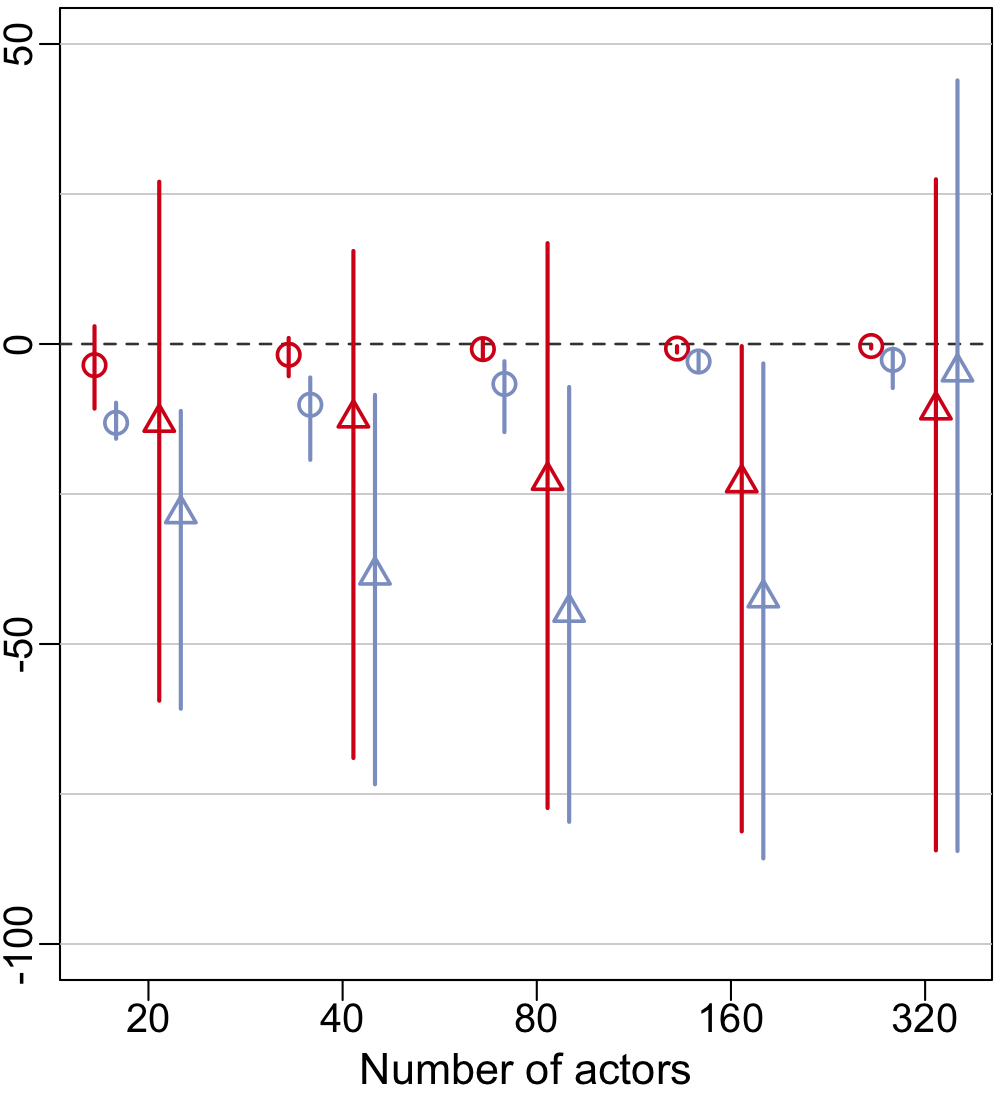} 
  \hspace{-.22in}&\hspace{-.22in}
\includegraphics[width=.315\textwidth]{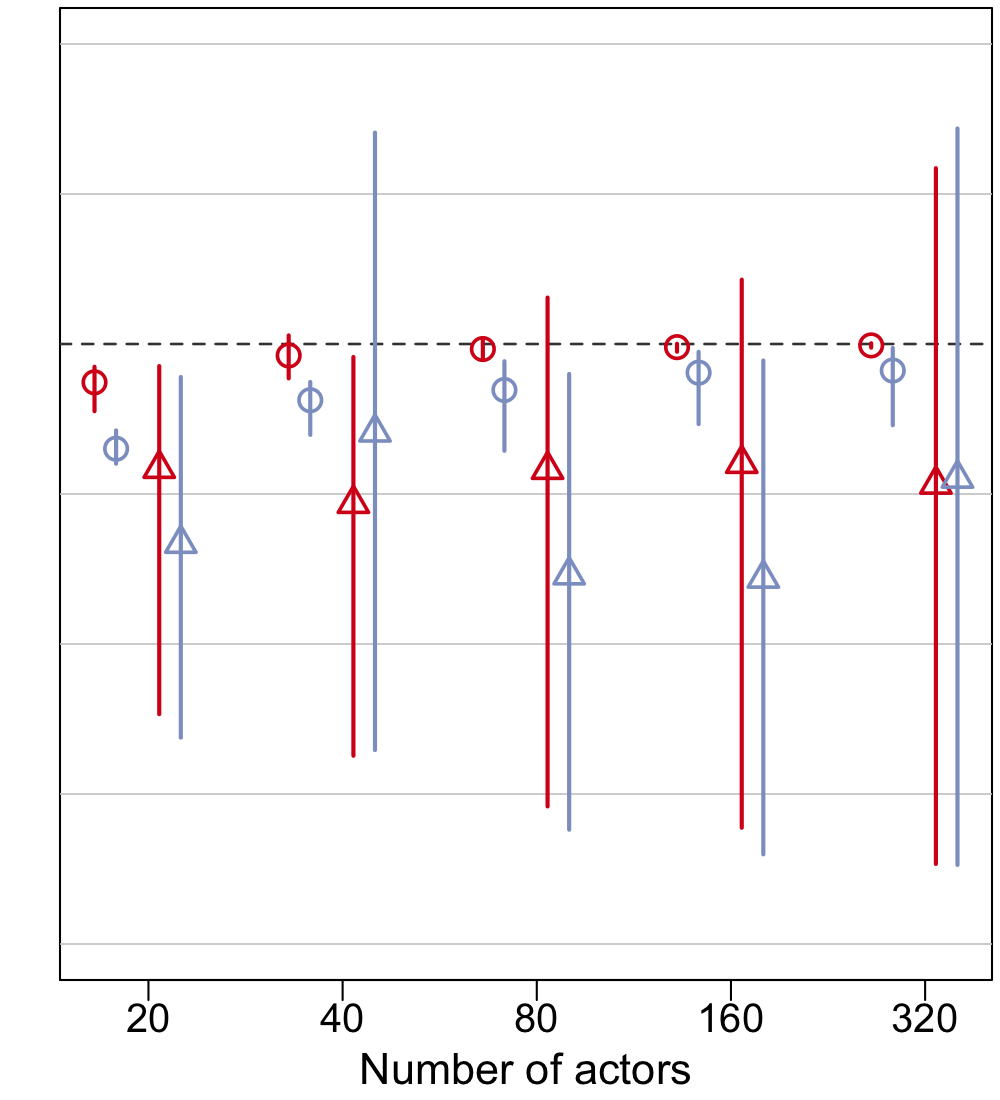} 
  \hspace{-.22in}&\hspace{-.22in}
\includegraphics[width=.315\textwidth]{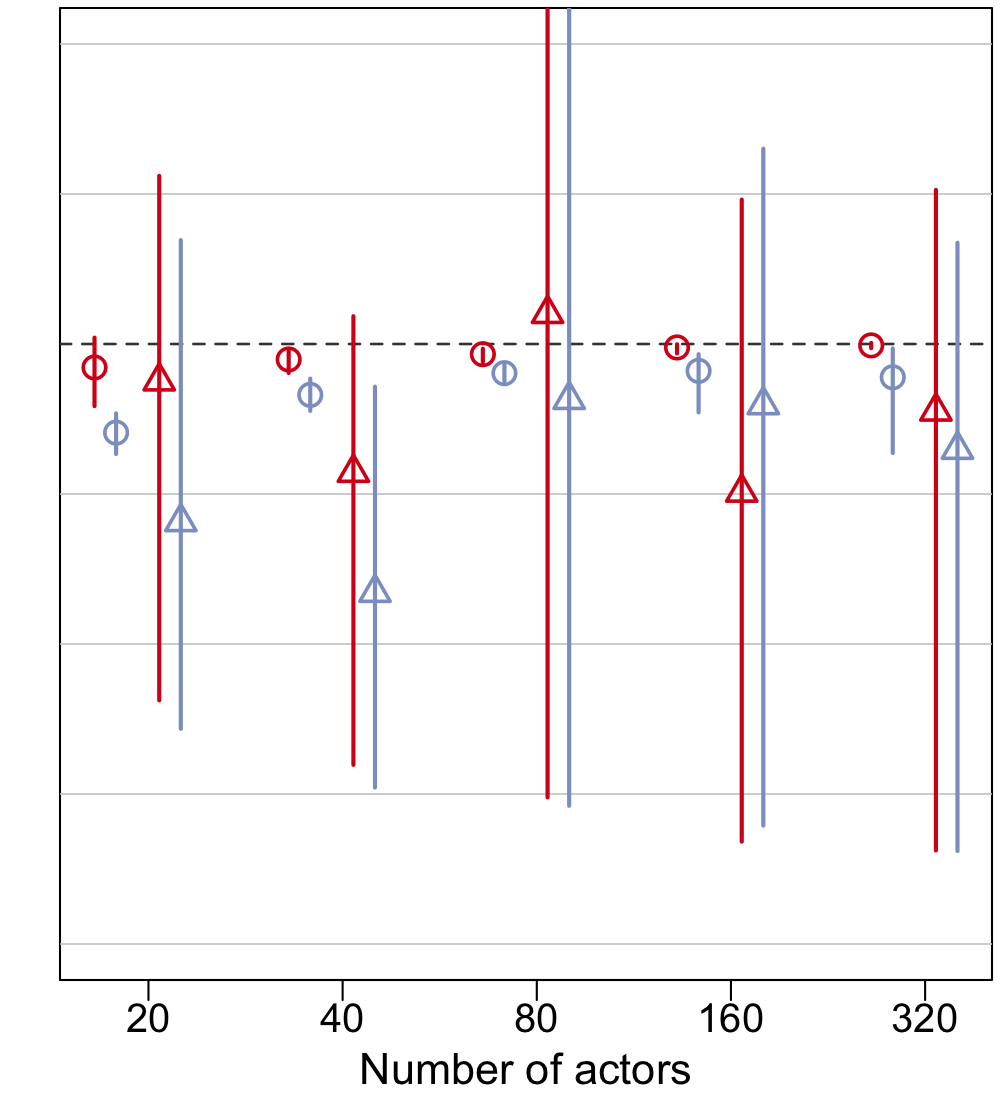} 
  \\
  \end{tabular}
  \caption{Estimated bias of standard error estimators for each of three covariates (binary, positive, and real-valued) when the errors are generated from  exchangeable (circles) and non-exchangeable (triangle) error models. 
Points represent mean estimated percent error (of standard error estimators) and lines represent 95\% quantiles of percent error for exchangeable (red) and dyadic clustering (blue)
estimators. }
      \label{fig:sim_bias}
\end{figure}

To verify Theorem~\ref{thm:bias}, we compared the biases of the dyadic clustering and exchangeable estimators of the variances of the regression coefficients in the simulation studies with exchangeable and independent errors. We computed the mean bias for each variance estimator across all simulations for each error structure, data set size, coefficient, and estimator. Then, we took the ratio of mean absolute biases of the dyadic clustering variance estimator to that of the exchangeable estimator. The ratios, averaged for all $n$, are given in Table~\ref{tab:sim_bias}. 
We note that, in all cases, the mean bias of the dyadic clustering estimator was negative, and larger in absolute value than that of the exchangeable estimator. The fact that every ratio is greater than 2 confirms Theorem~\ref{thm:bias}.

\begin{table}[ht]
\centering
\caption{Mean (across values of $n$) ratio of average (across covariate and error realizations) absolute bias in dyadic clustering variance estimator to average absolute bias of exchangeable variance estimator for independent and exchangeable error models. }
\label{tab:sim_bias}
\begin{tabular}{rrrr}
 & binary & positive & real-valued \\ 
independent & 96.48 & 5.03 & 123.45 \\ 
  exchangeable & 3.04 & 5.43 & 5.02 \\ 
\end{tabular}
\end{table}

\section{Dyadic clustering covariance matrix invertibility}
\label{app_DC_inv}
Ideally, for a covariance matrix estimate $\hat{\Omega}$ to be of utmost utility, it must be invertible. For example, if we wish to employ feasible generalized least squares, 
the estimate of the the covariance matrix must be nonsingular. However, in many cases the dyadic clustering estimator is singular and hence cannot be inverted.  In cases when the dyadic clustering estimator $\hat{\Omega}_{DC}$ is singular, it can still be used to  produce the dyadic clustering estimator $\hat{V}_{DC}$ in \eqref{eq:varc}.

\begin{theorem} The dyadic clustering estimate of the error variance, $\hat{\Omega}_{DC}$, is singular for directed data with $R=1$.
\label{theoremDC}
\end{theorem}
\begin{proof}
The dyadic clustering estimator can be written as the Hadamard product between the outer product of the residuals and a matrix of indicators of whether the dyad indices share a member,
\[ \hat{\Omega}_{DC}=ee^{T}\circ {1}{\{(i,j) \cap (k,l)\neq \varnothing \} }, \]
where ${1}{\{(i,j) \cap (k,l)\neq \varnothing \} }$ is an $n(n-1) \times n(n-1)$ matrix of indicators of indices $(i,j)$ sharing an index with $(k,l)$.
The rank of the outer product of the residuals is one: ${\rm rank}(ee^T) = 1$. The rank of the indicator matrix is at most $n(n-1)/2$, since the indices $(i,j)$ share a member with an arbitrary pair $(k, l)$ if and only if the indices $(j,i)$ do as well. Thus, the column of ${1}{\{(i,j) \cap (k,l)\neq \varnothing \} }$ corresponding to $(i,j)$ is the same as that corresponding to $(j,i)$.

For any two square matrices of equal size $A$ and $B$, ${\rm rank}(A \circ B) \le {\rm rank}(A){\rm rank}(B)$. Thus, 
 \begin{align*}
{\rm rank}(\hat{\Omega}_{DC}) &= {\rm rank}[ ee^{T}\circ {1}{\{(i,j) \cap (k,l)\neq \varnothing \} } ], \\
&\le {\rm rank}(ee^{T}) {\rm rank}[{1}{\{(i,j) \cap (k,l)\neq \varnothing \} }], \\
&\le \frac{n(n-1)}{2}.
\end{align*}
$\hat{\Omega}_{DC}$ is therefore not full rank. 
\end{proof}

\begin{remark}
Theorem~\ref{theoremDC} does not hold for undirected data when $R=1$. If the data are undirected, then the bound does not guarantee singularity of $\hat{\Omega}_{DC}$ since the dimension of $\hat{\Omega}_{DC}$ is exactly $n(n-1)/2$. In practice, we find that $\hat{\Omega}_{DC}$  may be full rank in this case. 
\end{remark}

\begin{remark}
The result of Theorem~\ref{theoremDC} holds for both directed and undirected data when $R>1$. In this case, 
the column in the indicator matrix ${1}{\{(i,j) \cap (k,l)\neq \varnothing \} }$ corresponding to the indices $(i,j,s)$ is the same as that column corresponding to $(i,j,r)$ for all values of $r= 1,\ldots,R$. Thus, again $\hat{\Omega}_{DC}$ is not full rank. 
\end{remark}

\section{Efficient inversion of the exchangeable covariance matrix}
\label{app_Omega_inv}
To perform the generalized least squares estimation procedure as described in Section~\ref{sec:trade}, we must invert the exchangeable covariance matrix ${\Omega}_E$ as defined in Figure~\ref{exchCOVMAT}.  For now, we work in the case where $R=1$.  Since ${\Omega}_E$ is a real symmetric matrix, its inverse is real and symmetric as well. However, we can say more about the patterns in the inverse ${\Omega}_E^{-1}$. Recall that  ${\Omega}_E$ has at most six unique terms; call these parameters $\phi$. We find that the inverse ${\Omega}_E^{-1}$ has at most six unique terms as well. If we define the parameters in ${\Omega}_E^{-1}$ as ${p}$, we can write
\begin{align}
{\Omega}_E({\phi}) {\Omega}_E^{-1}({p}) = I, \quad {\phi}, {p} \in \mathbb{R}^6, \nonumber
\end{align}
where $I$ is the $n(n-1) \times n(n-1)$ identity. Lastly, we make the conjecture that the parameter pattern in ${\Omega}_E^{-1}$ is exactly the same as that in ${\Omega}_E$; we find this conjecture to be true in practice. One caveat is that the locations in which we assume zeros in ${\Omega}_E$ are not zero in ${\Omega}_E^{-1}$ in general. 

We can find the inverse parameters ${p}$ from ${\phi}$ without inverting the entire matrix ${\Omega}_E$ by instead solving the following linear system 
\begin{align}
C({\phi}, n){p} = (1,0,0,0,0,0)^T, \quad C({\phi}, n) \in \mathbb{R}^{6 \times 6},  \nonumber
\end{align}
where $C({\phi}, n)$ is a set of six linear equations based on the parameters ${\phi}$ and the number of actors $n$ and is depicted in Figure~\ref{fig:cmat}. Thus, we replace the need to invert the $n(n-1) \times n(n-1)$ matrix $\Omega_E$ by the inversion of the $6 \times 6$ matrix $C({\phi}, n)$. Using this procedure, the computational cost associated with finding the inverse of $\Omega_E$ is independent of the  number of actors $n$. 
\begin{sidewaysfigure}
\centering
\includegraphics[width=.99\textwidth]{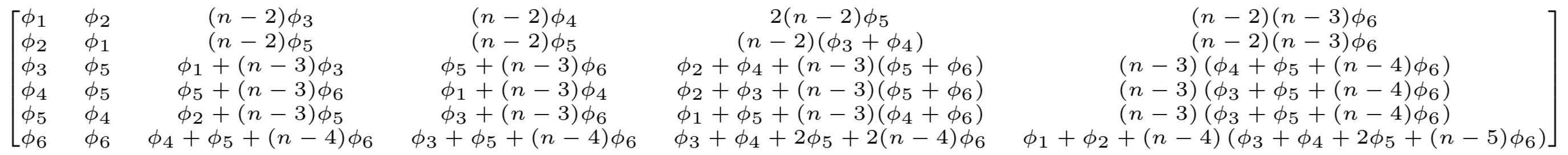}
\caption{Matrix $C(\bs{\phi}, n)$.}
\label{fig:cmat}
\end{sidewaysfigure}

Now consider the case of array data with $R>1$. Inversion of the exchangeable covariance $\Omega_E$ in Figure~\ref{fig:directeddyad} requires consideration of the  patterns in the block matrices. Recalling that $\Omega_E$ is  parametrized by twelve terms when $R>1$, we denote the first six parameters as ${\phi}^{(1)}$ and the second six ${\phi}^{(2)}$, corresponding to $\Omega_1$ and $\Omega_2$ respectively. Again the inverse $\Omega^{-1}$ has the exact same block matrix pattern as $\Omega$. Thus, the inverse may be parametrized by ${p}^{(1)}$ and ${p}^{(2)}$, each with length six, defined by the following linear equations.
\begin{align}
C({\phi}^{(1)}, n){p}^{(1)} + (R-1)C({\phi}^{(2)}, n){p}^{(2)} & = (1,0,0,0,0,0)^T, \nonumber  \\  
C({\phi}^{(2)}, n){p}^{(1)} + C({\phi}^{(1)}, n){p}^{(2)} + (R-2)C({\phi}^{(2)}, n){p}^{(2)} & = 0. \nonumber
\end{align}
 This is twelve linear equations in ${p}^{(1)}$ and ${p}^{(2)}$. In this formulation we reduce a $Rn(n-1) \times Rn(n-1)$ inversion to a $12 \times 12$ inversion for calculation of $\Omega_E^{-1}$. Again, there is no dependence of the complexity of the inversion on the array dimensions $n$ and $R$.

\section{Trade data analysis}
\label{app_trade}

\subsection{Trade model details}
We analyzed international trade among 58 countries over 20 years. The data were analyzed and made available (DOI:\url{10.1214/10-AOAS403}) by~\citet{westveld2011mixed}, which used a modified gravity mean model to represent the logarithm of yearly trade between each pair of countries as linear function of seven covariates in years 1981-2000.
The gravity model, proposed by~\citet{tinbergen1962shaping}, posits that the total trade between countries is proportional to overall economic activity of the countries weighted by the inverse of the distance between them (raised to a power).  Following~\citet{ward2007persistent} we added an indicator for whether the nations' militaries cooperated in the given year and a measure of democracy, i.e. polity, which ranges from $0$ (highly authoritarian) to $20$ (highly democratic).
The complete model has the form:
\begin{align}
\text{log} \mbox{Trade}_{ijt}=& \beta_{0t} + \beta_{1t} \text{log}  \text{GDP}_{it} + \beta_{2t} \text{log}  \text{GDP}_{jt} + \beta_{3t} \text{log}  \text{D}_{ijt} \label{eq:trade} \\
& \hspace{.1in} +\beta_{4t} \, \text{Pol}_{it} +\beta_{5t}\, \text{Pol}_{jt} +\beta_{6t} \, \text{CC}_{ijt} +\beta_{7t}\left(\text{Pol}_{it} \times \text{Pol}_{jt}\right)+\epsilon_{ijt}, \nonumber
\end{align}
where $\text{log} \mbox{Trade}_{ijt}$ is the (log) volume of trade sent from country $i$ to country $j$ in year $t$; $\text{log} \text{GDP}_{it}$ and $\text{log} \text{GDP}_{jt}$ are the (log) Gross Domestic Product of nations $i$ and $j$, respectively; $\text{log} \text{D}_{ijt}$ is the (log) geographic distance between nations; $\text{CC}_{ijt}$ is the measure of cooperation in conflict (coded as +1 if nations were on the same side of a dispute and -1 if they were on opposing sides); and $\text{Pol}_{it}$ and $\text{Pol}_{jt}$ are the polity measures for $i$ and $j$, respectively, where polity ranges from $0$ (highly authoritarian) to $20$ (highly democratic).

\subsection{Estimation}
Let $\beta$ denote the length $8 T$ vector of regression coefficients defined in \eqref{eq:trade}, where $t=1,\ldots, T$ and $T=20$ is the final time period. We estimated $\beta$ using feasible generalized least squares as in \eqref{eq:gls}, where $y$ is a vector consisting of $( \text{log} \mbox{Trade}_{ijt} )$ where $(i=1, \ldots, n; j=1,\ldots, n; i \neq j; t = 1,\ldots, T)$, $X$ is a matrix of corresponding covariates defined by the model in \eqref{eq:trade}, and we assume $\Omega$ is jointly exchangeable. We initialized $\beta$ using ordinary least squares. Then, we estimated $\Omega$ using the residuals and the joint exchangeability assumption, as in \eqref{eq:exch.est} and \eqref{eq:resid_avg}. We then used 
\eqref{eq:gls} with $\Omega = \hat{\Omega}$ to re-estimate $\beta$. We iterated estimates of $\Omega$ and $\beta$ until convergence. We defined convergence as an absolute change in weighted residual inner product being less than a particular tolerance $\epsilon > 0$. Explicitly, convergence is defined to occur when
\begin{align}
    &|Q^{\gamma} - Q^{\gamma-1}| < \epsilon , \nonumber \\
    &Q^\gamma = (y - X \hat{\beta}^{\gamma})^T (\hat{\Omega}^{\gamma})^{-1} (y - X \hat{\beta}^{\gamma}), \ (\gamma = 1,2,\ldots), \nonumber
\end{align}
where $\hat{\beta}^{\gamma}$ and $\hat{\Omega}^\gamma$ are the estimators at iteration $\gamma$ of the regression coefficients and error covariance matrix, respectively. For the trade data analysis, we take $\epsilon = 10^{-6}$.
After convergence, we set $\Omega = \tilde{\Omega}$ in \eqref{eq:gls_var} to obtain standard errors for the estimated regression coefficients, which is the standard approach in feasible generalized least squares.

\subsection{In-sample evaluation}
\label{sec:trade_insample}
We compared the coefficients and confidence intervals resulting from  the proposed feasible generalized least squares approach to the estimated coefficients and credible intervals from the mixed effects model of \cite{westveld2011mixed}, and to coefficients estimated by ordinary least squares and confidence intervals using the dyadic clustering estimator $\hat{V}_{DC}$.  
The estimated coefficients and corresponding 95\% confidence intervals and posterior credible intervals are shown in Figure~\ref{fig:tradecompare}.

\begin{figure}
\centering
\begin{tabular}{c c }
 \includegraphics[width=.47\textwidth]{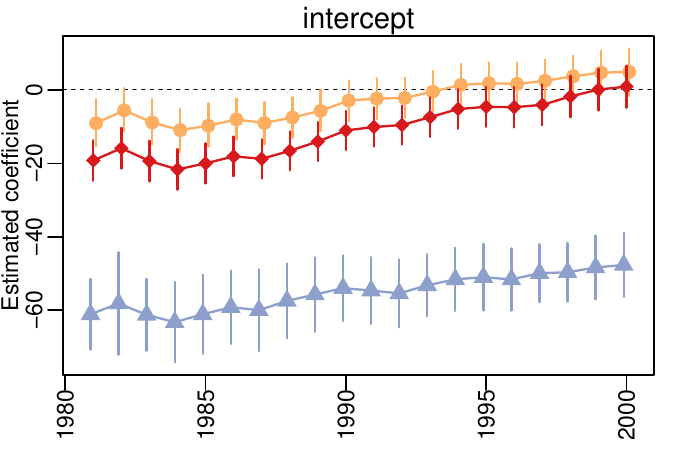}  &  \hspace{-.2in}
 \includegraphics[width=.47\textwidth]{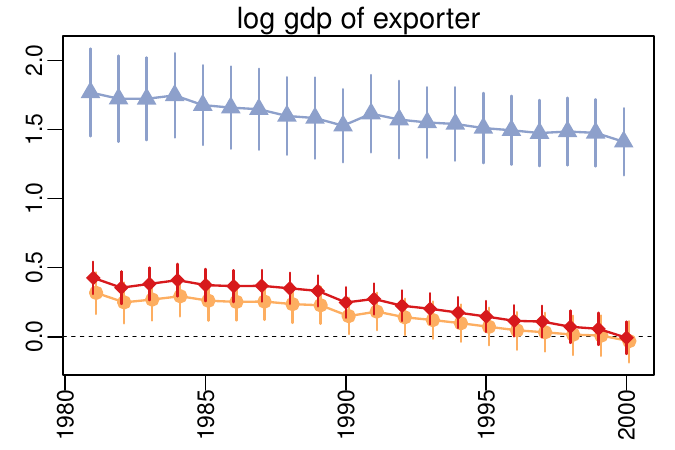}  \\[-2pt]
 \includegraphics[width=.47\textwidth]{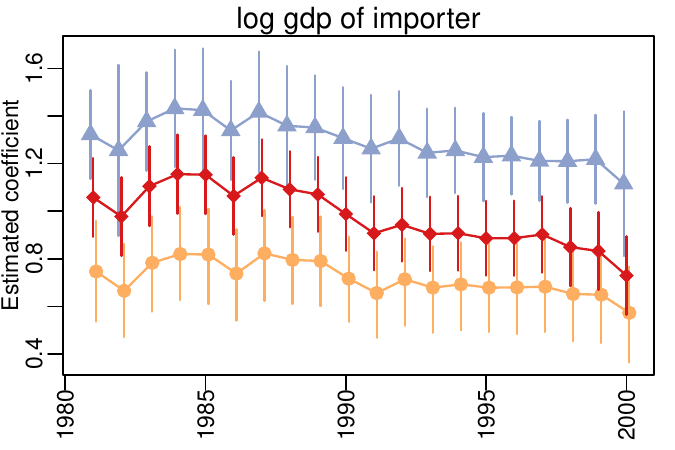}  &  \hspace{-.2in}
 \includegraphics[width=.47\textwidth]{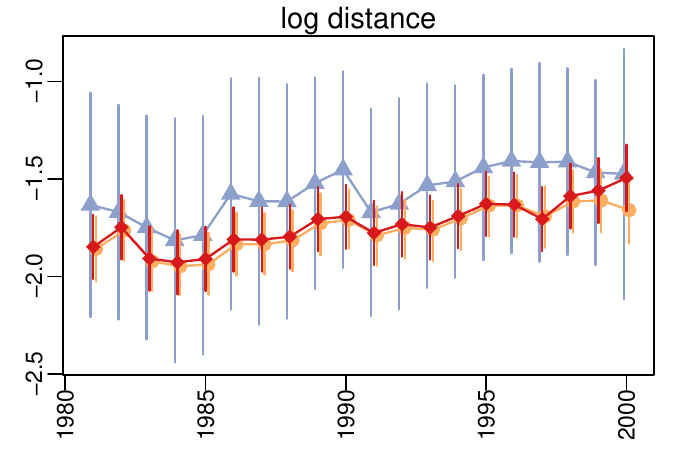}  \\[-2pt]  
 \includegraphics[width=.47\textwidth]{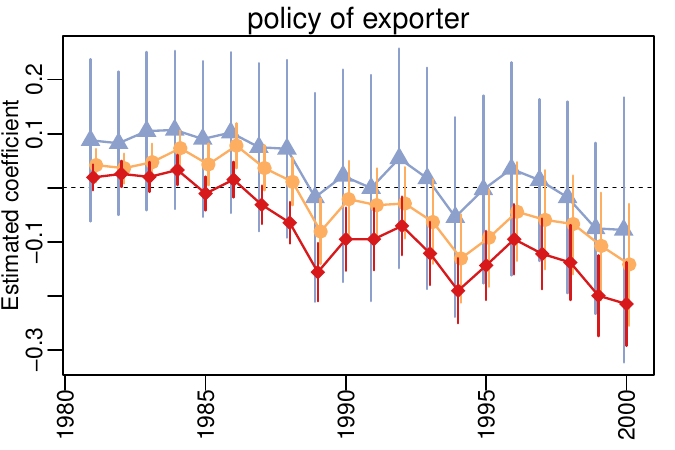}  &  \hspace{-.2in}
 \includegraphics[width=.47\textwidth]{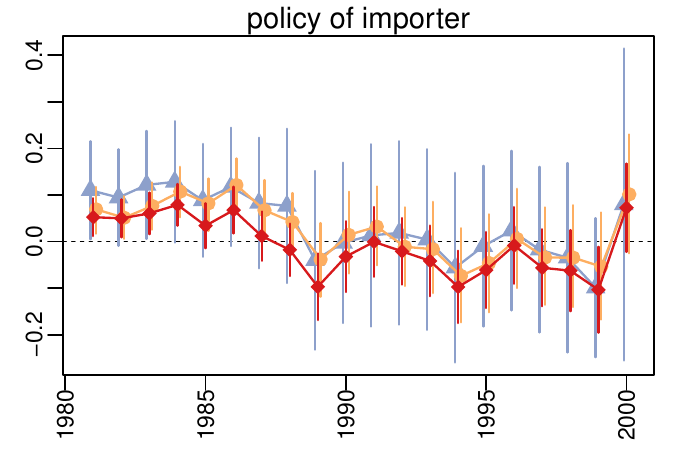}  \\[-2pt]
 \includegraphics[width=.47\textwidth]{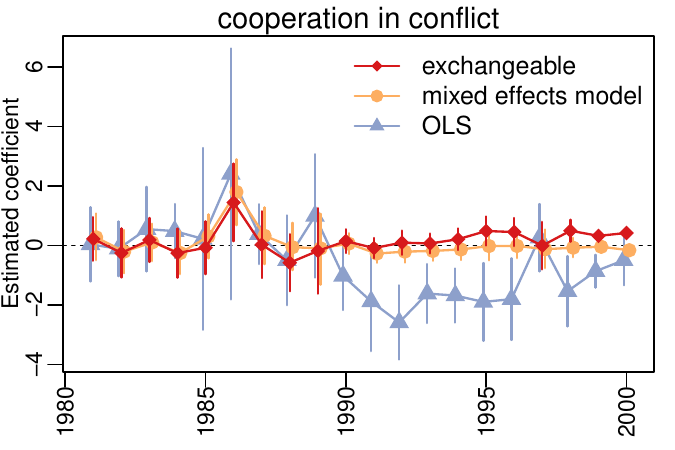}  &  \hspace{-.2in}
 \includegraphics[width=.47\textwidth]{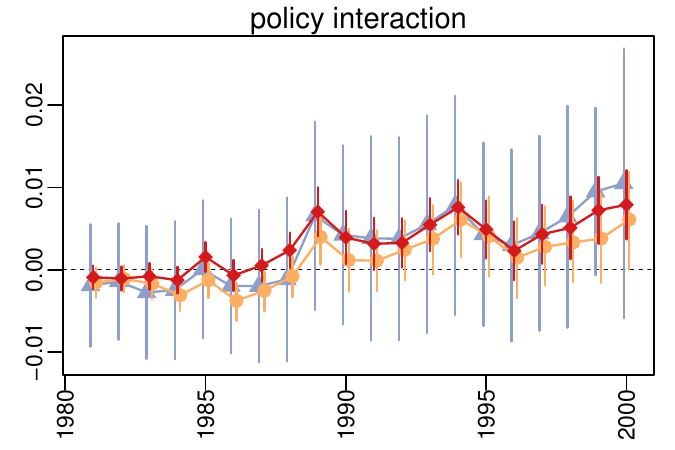}    
    \end{tabular}
        \vspace{-.1in} 
  \caption{Estimated coefficients and 95\% confidence/credible intervals using three different estimation techniques: exchangeable (red/diamonds), mixed effects model (orange/circles), and ordinary least squares with dyadic clustering estimator (``OLS,'' triangles/blue). }
\label{fig:tradecompare}
\end{figure}

\subsection{Prediction study}
To compare the ability of the proposed exchangeable approach and the model from \citet{westveld2011mixed}  -- which we refer to as the hierarchical, longitudinal mixed effects model --  to represent the trade data,
we examined the  out-of-sample predictive performance of the estimators. We detail the computation of predicted values here. 
First, we establish some notation for this section. Recall that the trade data set has covariate measures in $X$ that vary by year. Thus, we rewrite the linear model in \eqref{eq:trade} as
\begin{align*}
    \y_t= \X_t \bbeta_t + \bbxi_t, \hspace{.2in} (t = 1,\ldots, T ),
\end{align*}
where $\y_t$ represents the $n(n-1)$ vector of relations among the $n=58$ countries in year $t$, $\X_t$ is a matrix eight covariates corresponding to year $t$, and $\bbxi_t$ is a vector of errors for year $t$. 

We estimated both models on the first 4 through 19 years of data, and then used these estimates to predict the trade in the following year. To generate predictions from each model, we computed the conditional expectation $E( {\y}_{T} \mid \{\y_t \}_{t=1}^{T-1} )$ based on the assumption that $\y_{T}$ and $( \y_r )$, for $r = 1, \ldots, T-1$, are jointly normal. 
The exchangeable approach and mixed effects model correspond to different models of the variance-covariance matrices of $\bbxi_t$ and the covariances matrices between vectors $\bbxi_t$ and $\bbxi_{t+h}$.  As a baseline, we included ordinary least squares, assuming independence of each year, i.e. $\bbxi_t$ independent of $\bbxi_{t+h}$. 

We first discuss the ordinary least squares estimation and prediction procedure used for this study, as this approach is the simplest. We estimated the coefficients in the prediction time period, $\bbeta_{T}$, with the coefficients from the previous time period and assumed independent and identically distributed entries in all $\bbxi_t$.
Thus, the ordinary least squares estimator of the trade in year $T$ is
\begin{align*}
    E( {\y}_{T}  \mid \{\y_t \}_{t=1}^{T-1} )_{OLS} = \X_{T} \hat{\bbeta}_{T-1}.
\end{align*}

For the exchangeable procedure, we again set $\hat{\bbeta}_{T} = \hat{\bbeta}_{T-1}$. Based on the model underlying the exchangeable procedure (Figure 6(a)), the variance $\var(\y_t) = \bOmega_1$ for all $t =1,\ldots, T$ and the covariance $\text{cov} (\y_t, \y_{t+h}) = \bOmega_2$ for all $h$. Further, it can be shown that the precision corresponding to the concatenated  vector $z_{T-1}^T = (\y_1^T, \y_2^T, \ldots, \y_{T-1}^T)$ is of the same pattern as the variance $\var (z_{T-1})$, which has $\bOmega_1$ along the diagonal blocks and $\bOmega_2$ in the off-diagonal blocks (as in Figure 6(a)). We define the diagonal blocks of $\var (z_{T-1})^{-1}$ as $\bPsi_1$ and the off-diagonal blocks $\bPsi_2$. Then, when the relations $\{\y_t \}_{t=1}^{T}$ are jointly normally distributed, we define the prediction from the exchangeable procedure is
{\small
\begin{align*}
    E( {\y}_{T}  \mid \{\y_t \}_{t=1}^{T-1} )_{E} = \X_{T} \hat{\bbeta}_{T-1} + \bOmega_2 \{ \bPsi_1 + (T-2)\bPsi_2 \}  \sum_{t=1}^{T-1} \left(\y_t - \X_t \hat{\bbeta}_t \right).
\end{align*}
}

Finally, to detail predictions from the mixed effects model, we first review some terms used in \citet{westveld2011mixed}. The proposed model is 
\begin{align*}
    y_{ij, t} &= x_{ij,t}^T \bbeta_t + s_{i,t} + r_{j,t} + g_{ij,t},
\end{align*}
where $s_{i,t}$ and $r_{j,t}$ are sender and receiver random effects, respectively, and $g_{ij,t}$ is a reciprocal random effect. These effects evolve according to autoregressive order one processes, such that 
{\footnotesize
\begin{align*}
    \begin{bmatrix}
    s_{i,t}\\
    r_{i,t}
    \end{bmatrix}
&= 
    \bPhi_{sr}
    \begin{bmatrix}
    s_{i,t-1} \\
    r_{i,t-1}
    \end{bmatrix}
    + \bepsilon_{i,t}, 
    \hspace{.2in} (i = 1,\ldots, T; \ t = 1, \ldots, T-1 ), \\
    \begin{bmatrix}
    g_{ij,t}\\
    g_{ji,t}
    \end{bmatrix} &= 
    \bPhi_{g}
    \begin{bmatrix}
    g_{ij,t-1} \\
    g_{ji,t-1}
    \end{bmatrix}
    + {e}_{ij,t}, 
        \hspace{.2in} (i,j = 1,\ldots, T; \ i \neq j ;\ t = 1, \ldots, T-1 ), 
\end{align*}
}
where $\bPhi_{sr}$ is a $2\times 2$ autoregressive matrix and $\bPhi_{g}$ is a symmetric autoregressive matrix. The 
error terms $\bepsilon_{i,t}$ and ${e}_{ij,t}$ are mean-zero independent bivariate Gaussian random variables. Please see \citet{westveld2011mixed} for more details of the model. With the notation and model defined, and again setting  $\hat{\bbeta}_{T} = \hat{\bbeta}_{T-1}$, algebra reveals that the prediction from the mixed effects model is
{\footnotesize
\begin{align*}
    E(  y_{ij, T}   \mid \{\y_t \}_{t=1}^{T-1} )_{MEM} 
&= x_{ij, T}^T \hat{\bbeta}_{T-1} 
    +
    \begin{bmatrix}
    1 & 0
    \end{bmatrix}^T
            \hat{\bPhi}_{sr}
                \begin{bmatrix}
    \hat{s}_{i,T-1} \\
    \hat{r}_{i,T-1}
    \end{bmatrix} \ldots
    \\
     &\hspace{.25in}+
     \begin{bmatrix}
    0 & 1
    \end{bmatrix}^T
        \hat{\bPhi}_{sr}
    \begin{bmatrix}
    \hat{s}_{j,T-1} \\
    \hat{r}_{j,T-1}
    \end{bmatrix}
    +
        \begin{bmatrix}
    1 & 0
    \end{bmatrix}^T
        \hat{\bPhi}_{g}
    \begin{bmatrix}
    \hat{g}_{ij,T-1} \\
    \hat{g}_{ji,T-1}
    \end{bmatrix}.
\end{align*}
}
To estimate the mixed effects model, \citet{westveld2011mixed} generated a Markov chain of length 55,000, of which the first 10,000 were discarded and every 20th iteration saved, giving 2,250 samples to approximate the posterior distributions of $\bbeta_{T-1}$, $\Phi_{sr}$, $\Phi_{g}$, and the random effects at $t = T-1$ in the above equations.  The mean of these 2,250 samples from the joint posterior distribution, for each pair $i\neq j$, was used to construct $E( \y_{T}  \mid \{\y_t \}_{t=1}^{T-1} )_{MEM} $.

We evaluate performance in the predicted year using the mean square prediction error, defined
\begin{align*}
    MSPE = \frac{1}{n(n-1)} \Big | \Big| E ( {\y}_{T} \mid \{\y_t \}_{t=1}^{T-1} ) - \y_{T} \Big | \Big|^2_2,
\end{align*}
where the expectation is replaced by the expectation of the appropriate prediction estimator (ordinary least squares, exchangeable, or mixed effects).

\section{A test for exchangeability}
\label{app_test}

\subsection{Motivation and details}
In this section we briefly motivate and describe a permutation test for assumption of exchangeability of the errors. For simplicity, we focus on the case where ordinary least squares is used to estimate regression coefficients $\beta$, and where $R=1$.  

If the errors are truly non-exchangeable in a way that impacts inference, then the dyadic clustering estimator should be extreme relative to the distribution of the dyadic clustering estimator under the assumption of exchangeability. The idea of the proposed test is to generate the null (exchangeable) distribution using permutations that mimic the exchangeability assumption.

If the exchangeable assumption is correct, then by definition, a simultaneous permutation of the rows and columns of the error matrix leaves the distribution unchanged. That is, 
\begin{align}
    \P( \xi_{ij} ) = \P( \xi_{\pi(i) \pi(j)} ), \nonumber
\end{align}
where $\P( \xi_{ij}  ) $ denotes the joint distribution of the set $(\xi_{ij})$ over all relations such that $i \neq j$. Although the errors are unobservable, the residuals are observable and can be used to approximate the errors. Thus, the residuals should be approximately exchangeable, such that 
\begin{align}
    \P( e_{ij} ) \approx \P( e_{\pi(i) \pi(j)}  ), \nonumber
\end{align}
where $e_{ij} = y_{ij} - \hat{\beta}^T x_{ij}$ is the residual for relation $(i,j)$. We rewrite the dyadic clustering estimator of the variance $\var(\hat{\beta} \mid X)$ as
\begin{align}
\hat{V}_{DC} = \sum_{(jk, lm) \in \Theta_0} e_{jk} e_{lm} \tilde{x}_{jk} \tilde{x}_{lm}^T,  \label{eq:eq_dc_sum}
\end{align}
where $\Theta_0$ represents the set of pairs of relations $(jk, lm)$ that share an actor and where $\tilde{x}_{jk} = (X^T X)^{-1} x_{jk}$.  Applying the joint exchangeability assumption to the residuals, one may generate an approximate sample of the dyadic clustering estimator under the null assumption of exchangeability
\begin{align}
\tilde{V}_{DC}^{(\pi)} & \approx  \sum_{(jk, lm) \in \Theta_0} e_{\pi(j) \pi(k)} e_{\pi(l) \pi(m)} \tilde{x}_{jk} \tilde{x}_{lm}^T. \label{eq:vdc_test_stat}
\end{align}

We use an $L_2$ norm to evaluate the discrepancy between the observed  dyadic clustering estimator based on the observed residual matrix in \eqref{eq:eq_dc_sum} and the distribution of dyadic clustering estimators under the exchangeable null in \eqref{eq:vdc_test_stat}. The squared distance can be computed using
\begin{align}
    \bar{d}^{(\pi)} = || \hat{V}_{DC} - \bar{V}_{DC}^{(\pi)} ||_2^2. \label{eq:dist_obs}
\end{align}
where $\bar{V}_{DC}^{(\pi)}$ is the element-wise sample mean dyadic clustering estimator of the exchangeable permutations of $\tilde{V}_{DC}^{(\pi)}$ in \eqref{eq:vdc_test_stat}. A distribution of these squared distances under the exchangeable null may be obtained by replacing the dyadic clustering estimator based on the observed residuals with the permuted samples $\tilde{V}_{DC}^{(\pi)}$,  
\begin{align}
    {d}^{(\pi)} = || \tilde{V}_{DC}^{(\pi)} - \bar{V}_{DC}^{(\pi)} ||_2^2. \label{eq:dist_distr}
\end{align}
If $\bar{d}^{(\pi)}$ in \eqref{eq:dist_obs} is extreme relative to the distribution of ${d}^{(\pi)}$ in  \eqref{eq:dist_distr} under the exchangeable null, then this constitutes evidence that the errors are non-exchangeable. We obtain an empirical $p$-value by computing the fraction of the permuted, squared distances ${d}^{(\pi)}$ that are greater than $\bar{d}^{(\pi)}$.

\subsection{Evaluating the test}
We simulated from exchangeable and non-exchangeable relational regression models to assess the ability of our proposed hypothesis testing procedure to detect non-exchangeability. We evaluated the power and size of the test by simulating data from null and alternative models and evaluating the test on each simulated dataset. 

We begin by specifying the null exchangeable model. In addition to an intercept, we include a single relational covariate, the coefficient of which is the target of inference. The model is specified as
\begin{align}
    &y_{ij} = \beta_0 + \beta_1 x_{ij} + \xi_{ij}, \label{eq:sim_null} \\
    &v_i \stackrel{iid}{\sim} N(0, 1), \ i = 1,2,\ldots,n, \nonumber \\
    &x_{ij} = \frac{1}{\sqrt{2}}(v_i - v_j), \nonumber \\
    &w_i \stackrel{iid}{\sim} N(0, 1), \ i = 1,2,\ldots,n, \nonumber \\
    &\epsilon_{ij} = \frac{1}{{2}}(w_i - w_j), \nonumber \\
    & \xi_{ij} = \epsilon_{ij} + \frac{1}{\sqrt{2}}z_{ij}, \nonumber \\
    & z_{ij}\stackrel{iid}{\sim} N(0, 1). \nonumber
\end{align}
Such a model produces true variance parameters $\phi^T = [1, 1/2, 1/4, 1/4, -1/4, 0]$. Further, $\var(x_{ij}) = 1$, such that the variance of the errors and variance of the covariate $x_{ij}$ are equivalent. We take $\beta_0 = \beta_1 = 1$.

The non-exchangeable model we consider can be expressed as follows:
\begin{align}
    &y_{ij} = \beta_0 + \beta_1 x_{ij} + \xi_{ij}, \label{eq:sim_alt} \\
    &v_i \sim N(0, (i/n)^2), \ i = 1,2,\ldots,n, \nonumber \\
    &x_{ij} = \frac{1.7}{\sqrt{2}}(v_i - v_j), \nonumber \\
    &w_i \stackrel{\indep}{\sim} N(0, v_i^2), \ i = 1,2,\ldots,n, \nonumber \\
    &\epsilon_{ij} = 1.7(w_i - w_j), \nonumber \\
    & \xi_{ij} = (\epsilon_{ij} + \frac{1}{\sqrt{2}}z_{ij})/1.6, \nonumber \\
    & z_{ij}\stackrel{iid}{\sim} N(0, 1). \nonumber
\end{align}
Under this model, both $X$ and the errors are  non-exchangeable and the variance of the errors errors depend on values of $X$. The numerical prefactors were chosen such that $\var(x_{ij}) \approx 1$ and $\var(\xi_{ij}) \approx 1$ as in the exchangeable model in \eqref{eq:sim_null}.

To evaluate the power and size of the test, 
we simulated 1,000 datasets from each of the null and alternative models in \eqref{eq:sim_null} and \eqref{eq:sim_alt}, respectively. For each simulated dataset, we evaluated the empirical two-sided $p$-value using 1,000 permutations of the form of \eqref{eq:vdc_test_stat}, and rejected the null hypothesis at the $\alpha = 0.05$ level. 
The test results in Table~\ref{tab:test_size} show approximate size, i.e. the probability of rejecting the null when the null is true, of $0.025$. Under the considered alternative distribution, Table~\ref{tab:test_size} gives the estimate approximate power of about 0.839. Thus, for the given alternative model at moderate relational data size ($n=50$), the proposed permutation test has a good chance of detecting the non-exchangeability in the alternative model, while failing to reject a true null with probability less than $0.05$. 
\begin{table}[h]
    \centering
    \caption{Estimated probability of a Type I error (size) and one less the probability of a Type II error (power) for the proposed permutation test.}
    \label{tab:test_size}
    \begin{tabular}{cc}
    Size & Power \\
    0.025 & 0.839 
    \end{tabular}
\end{table}
